\documentclass[]{amsart}
\usepackage[utf8]{inputenc}
\usepackage{amsaddr}
\usepackage{amsmath,amssymb,amsthm}
\usepackage{enumitem}
\usepackage{tikz}
\usetikzlibrary{patterns}
\usepackage[backend=biber,style=alphabetic]{biblatex}
\usepackage{hyperref}
\usepackage{xurl}
\usepackage{cleveref}
\hypersetup{breaklinks=true}
\DeclareMathOperator{\var}{Var}
\DeclareMathOperator{\tr}{tr}
\newcommand{\Z}{\mathbb{Z}}
\newcommand{\R}{\mathbb{R}}

\newcommand{\Q}{\mathbb{Q}}

\renewcommand{\P}{\mathbb{P}}
\newcommand{\E}{\mathbb{E}}

\newcommand{\ve}{\varepsilon}

\newtheorem{thm}{Theorem}[section]

\newtheorem{claim}[thm]{Claim}
\newtheorem{lem}[thm]{Lemma}

\newtheorem{remm}[thm]{Remark}
\newtheorem{deffo}[thm]{Definition}
\newtheorem{prop}[thm]{Proposition}
\newtheorem{ex}[thm]{Example}

\numberwithin{equation}{section}

\title[Nonstationary localization and unique continuation on $\Z^2$]{Localization and unique continuation for non-stationary Schr\"odinger operators on the 2D lattice}

\author{Omar Hurtado}
\address{University of California, Irvine\\
ohurtad1@uci.edu}

\newcounter{cnstcntC}

\newcounter{cnstcntc}

\begin{document}
\begin{abstract}
	We extend methods of Ding and Smart in \cite{Ding-Smart} which showed Anderson localization for certain random Schr\"odinger operators on $\ell^2(\Z^2)$ via a quantitative unique continuation principle and Wegner estimate. We replace the requirement of identical distribution with the requirement of a uniform bound on the essential range of potential and a uniform positive lower bound on the variance of the variables giving the potential. Under those assumptions, we recover the unique continuation and Wegner lemma results, using Bernoulli decompositions and modifications of the arguments therein. This leads to a localization result at the bottom of the spectrum. 
\end{abstract}
\maketitle
\section{Introduction}
\subsection{Main results}
The tight binding Anderson model in dimension $d$ is a random Schr\"odinger operator acting on $\ell^2(\Z^d)$ of the form
\begin{equation}\label{schrod} H =  -\Delta + V \end{equation} where $\Delta$ is the discrete Laplacian
\[[\Delta\psi](n) = \sum_{|m-n|=1} (\psi(m)-\psi(n)) \]
and $V = (V_n)$ is a random potential, with $V_n$ (usually) independent and identically distributed. Such random operators model the movement of an electron through disordered materials. There are related continuum models, where the finite difference operator $-\Delta$ is replaced with the usual negative Laplacian, and $V$ is a random potential in $L^\infty(\R^d)$, see e.g. \cite{Bourgain2005}. (We will not make this precise, as we do not examine continuum models in any great detail here.)

In this paper we will extend methods introduced by Ding and Smart in \cite{Ding-Smart} in their breakthrough work proving localization at the bottom of the spectrum for the two-dimensional Anderson Bernoulli model to produce unique continuation and Wegner-type estimates on the probability of resonances in a broader setting than originally considered in \cite{Ding-Smart}. Specifically, throughout we will assume the following on the family of real valued distributions $V = (V_n)_{n\in\Z^2}$:
\begin{enumerate}[label=(\Roman*)]\label{helo}
	\item The $V_n$ are jointly independent \label{ind}
	\item There is a real number $M > 0$ such that  $\P[0 \leq V_n \leq M] = 1$ for all $n \in \Z^2$ \label{unifbdd}
	\item The variables have uniformly positive variance, i.e. $\inf_{n\in\Z^2} \var V_n > 0$ \label{unifvar}
\end{enumerate}
Under these assumptions, we can recover the major results of \cite{Ding-Smart}, which were originally formulated for $V_n$ i.i.d. non-trivial and bounded. (By non-trivial, we mean supported on at least two points.) We obtain a unique continuation result \Cref{ucucthmfinal} analogous to \cite[Theorem 3.5]{Ding-Smart}. We present here a variant which is easier to formulate, and let $\mathcal{E}_{uc}(\Lambda,\alpha, \ve)$ denote the event that for a box $\Lambda \subset \Z^2$, the hypotheses
\begin{equation}
	\begin{cases}
		|\overline{E} - E| \leq e^{-C(L\log L)^{1/2}}\\
		H\psi = E\psi \text{ in }\Lambda\\
		|\psi| \leq 1 \text{ in a } 1-\ve(L\log L)^{-1/2} \text{ fraction of }\Lambda
	\end{cases}
\end{equation}
taken together imply $|\psi| \leq e^{CL\log L}$. And so by contrapositive if $\|\psi\|_{\ell^\infty(\Lambda)} = 1$, there is a $\ve L^{3/2-}$ sized set where $|\psi| \geq e^{-CL\log L/2}$ (if our potential is such that this event holds). This is a property of the potential and thus a random event. The simplified variant of our two dimensional probabilistic unique continuation result is as follows:
\begin{thm}
	Let $V = (V_n)_{n\in\Z^2}$ be a random potential satisfying conditions \ref{ind}, \ref{unifbdd}, and \ref{unifvar}. 
	Then for all sufficiently small $\ve$, there is a corresponding $\alpha > 1$ such that for any fixed $\overline{E} \in [0,8+M]$ and square $\Lambda \in \Z^2$ with side length $L \geq \alpha$, the unique continuation event $\mathcal{E}_{uc}$ defined above satisfies the bound
	\[ \P[\mathcal{E}_{uc}] \geq 1 - e^{-\ve L^{1/4}}\]
	
\end{thm}
Using our unique continuation result, we can obtain a certain variant of the Wegner-type estimate on the probability of a resonance, i.e. larger than expected resolvent at scale $L$ given the size of resolvents at smaller scales. Roughly the result says that if at scale $L$, eigenfunctions (of the finite volume operator on a box of length $L$) in a certain band of energies obey a certain quasi-localization condition then the probability of a resonance is of order $L^{-1/2}$. The precise result is \Cref{bigweg}.

Finally, this Wegner lemma allows the proof of certain resolvent bounds, which in turn imply a localization result for energies near $N$. Our notions of localization are the following:
\begin{deffo}
We say an operator $H$ acting on $\ell^2(\Z^2)$ is Anderson localized in an interval $I$ if it has no continuous spectrum in $I$ and moreover the eigenfunctions associated to eigenvalues in $I$ are exponentially decaying, i.e. for all such eigenfunctions $\psi$, there are positive $c$ and $C$ so that
\[ |\psi(n)| \leq Ce^{-c|n|}\]
\end{deffo}
\begin{deffo}We say a random operator $H$ acting on $\ell^2(\Z^d)$ is strongly dynamically localized (SDL) of order $(s_1,s_2)$ in expectation in $I$ if
		\[ \E\left[\sup_{t\in\R}\left\|\langle X \rangle^{s_1} e^{-itH}\chi_{I}(H)\delta_0\right\|^{s_2}\right] < \infty\]
		(where $\langle X \rangle$ is the multiplication operator \textit{}$\psi(n) \mapsto (n^2+1)^{1/2}\psi(n)$, essentially the position operator.)
\end{deffo}

We obtain two results concerning localization for non-stationary two dimensional random Schr\"odinger operators, the first implying the second.
\begin{thm}\label{sdlocalization}
	Let $V = (V_n)_{n\in\Z^2}$ be a random potential satisfying conditions \ref{ind}, \ref{unifbdd}, and \ref{unifvar}. Then for any positive $s_1$ and any sufficiently small positive $s_2$ depending on $s_1$, there is $E_0>0$ (depending on $s_1$, $s_2$, and the potential $V$)  such that $H$ is strongly dynamically localized in expectation in the interval $[N,N+E_0]$.
\end{thm}
It was shown by Germinet and Klein \cite{GK2012} that SDL in expectation follows from the resolvent bounds obtained via the Bourgain and Kenig MSA in \cite{Bourgain2005}, but this was in the continuum context. Rangamani and Zhu in \cite{rz2023dynamical} extended this work to the discrete context, so that in particular \Cref{sdlocalization} is a corollary of the aforementioned resolvent bounds.  As was mentioned previously, dynamical localization in expectation implies weaker forms of localization; this was first proved in \cite{ks81}; one can also see e.g. \cite{del1996operators} for further discussion of the relationship between spectral and dynamical notions of localization. As a consequence, from the dynamical result we also obtain the following spectral result:
\begin{thm}\label{andlocal}
	Let $V = (V_n)_{n\in\Z^2}$ be a random potential satisfying conditions \ref{ind}, \ref{unifbdd}, and \ref{unifvar}. Then $H$ is almost surely Anderson localized in $[N,N+E_0]$, where $E_0$ is the same as in \Cref{sdlocalization}.
\end{thm}
As we mentioned previously, we do not prove localization directly; instead \cite{rz2023dynamical} allows us to reduce the problem to showing certain bounds on the resolvent of finite volume truncations. It is worth emphasizing that Anderson localization was obtained from such resolvent estimates by the earlier paper \cite{Ding-Smart}, and this used the general strategy of \cite{Bourgain2005}, so that in particular \Cref{andlocal} does not require the methods introduced by \cite{rz2023dynamical}. The necessary resolvent bounds are the following:
\begin{thm}\label{endofmsa1}
	Given a random potential satisfying \ref{ind}, \ref{unifbdd} and \ref{unifvar}, for any $0 < \gamma < 1/2$, there are $\alpha > 1 > \ve > 0$ and $E_0>0$ such that for every $0\leq \overline{E} \leq E_0$ and every square $\Lambda \subset \Z^2$ satisfying $L \geq \alpha$, we have the bound
	\[ \P[|(H_\Lambda-\overline{E})^{-1}(x,y)| \leq e^{L^{1-\ve}-\ve|x-y|}\text{ for all }x,y \in \Lambda] \geq 1-L^{-\gamma}\]
	where $H_\Lambda$ is the truncation of $H$ to a box $\Lambda \subset \Z^2$, and again $L$ is the side length.
\end{thm}
We note that many authors take as part of their definition of localization the supposition that the spectrum of the operator in question intersects the interval non-trivially, for the simple reason that localization is true vacuously over all intervals intersecting the spectrum trivially. We will not do this in order to state our results more succinctly. However, in the full generality where we only assume the conditions \ref{ind}, \ref{unifbdd}, and \ref{unifvar} hold, there is no guarantee that these spectral results are not vacuous in this sense. Much of what can usually be expected in the stationary case, e.g. a more or less deterministic spectrum and explicit description thereof as a finite union of intervals, is not necessarily true here. An example of some strange spectral phenomena for non-stationary potentials in one dimension is presented in \cite[Appendix A]{gorodetski2024nonstationary}. However, there is a wide class of potentials for which our results are non-vacuous, and where we really do obtain ``localization at the bottom of the spectrum''; this is the content of \Cref{suffcond}. (Roughly speaking, what is necessary is that $N$ is in the essential range of our potential at most sites in a uniform way.)

We also mention briefly that the restriction to $[0,E_0]$ comes from the availability in that energy range of an ``initial scale estimate'' which serves as the base case of an inductive argument. It is expected that, at least in the stationary (i.i.d. regime), localization holds throughout the spectrum. The only obstacle to proving this using the methods of Ding and Smart is the initial scale estimate; the same is true for our work. In any energy range where one can prove the appropriate estimate, one obtains localization. We discuss this further in \Cref{ltail}.

\subsection{Background} The study of localization for random Schr\"odinger operators is very well developed at this point, initiated (as far as mathematical results are concerned) by the celebrated paper \cite{GolMolPas77} of Gol'dsheid-Molchanov-Pastur. We refer the reader to e.g. \cite{JitErg}, \cite{awbook15} for comprehensive accounts. Our results, like those of Ding and Smart in \cite{Ding-Smart} and those also of Li and Li and Zhang in \cite{Li20} and \cite{Li-Zhang} respectively, concern localization for random operators where the random potential is singular and where the dimension $d > 1$.

Historically, the study of operators with singular noise has always proven technically challenging. Indeed, all the earliest localization results required regularity: for models of the form (\ref{schrod}) the first results appear in \cite{ks81}, and required that the noise be absolutely continuous, with bounded and compactly supported density. Soon after, Fr\"ohlich and Spencer introduced the multiscale analysis (MSA) in \cite{FS83}, which has since been a central method in the study of random Schr\"odinger operators, and shortly there after Fr\"ohlich, Martinelli, Scoppola and Spencer were able to use this method in \cite{Frohlich1985} to prove localization (in the ``expected'' regimes of extreme energy or high disorder) for random operators in arbitrary dimension, provided the noise was regular. 

MSA is an inductive method and in principle can be used for any type of noise so long as certain estimates hold; however in general these estimate are much easier to obtain in the setting of regular noise. Another central method in the study of random Schr\"odinger operators, the fractional moment method (FMM) introduced by Aizenman and Molchanov in \cite{Aizenman1993}, is technically much simpler and in certain ways more flexible, but depends in a fundamental way on the noise being regular.

The first breakthrough for singular noise came in the seminal paper \cite{Carmona1987} by Carmona, Klein, and Martinelli, which proved localization throughout the spectrum in $d=1$. In the one-dimensional setting, they were able to show the requisite estimates hold using the transfer matrix method; in particular the estimates were consequences of results in the Furstenberg theory of random matrix products, detailed in e.g. \cite{bouglac85}. Since the original work in \cite{Carmona1987}, other proofs of this fact have been found; see e.g. \cite{Shubin1998}, \cite{SJZhu19}, \cite{GORODETSKI}  for work in this vein.

For a long time, localization in higher dimensions was inaccessible for the case of singular noise. In the landmark paper \cite{Bourgain2005}, Bourgain and Kenig considered random Schr\"odinger operators in the continuum setting with Bernoulli noise, and were able to show localization at the bottom of the spectrum. Bourgain and Kenig showed that a quantitative form of unique continuation principle held for eigenfunctions of this operator, which gave a lower bound on the magnitude of finite volume eigenfunctions. With this lower bound, it was possible to obtain a Wegner-type estimate via eigenvalue variation. Such an estimate is one of the two major ingredients of MSA, and the other necessary ingredient, the initial scale estimate, is easily available even for singular noise. Hence the work of Bourgain and Kenig was enough to show Anderson localization at the bottom of the spectrum for the model under consideration.

It must be noted that the resolvent estimates obtained by the methods of Bourgain and Kenig were weaker than those used in the strongest known variant of MSA at the time, the bootstrap MSA of Germinet and Klein \cite{Germinet2001}. And so besides producing this estimate the authors also developed a novel variant of the multi scale analysis compatible with their weaker estimate. Germinet and Klein give a comprehensive account in \cite{GK2012}, also extending the Bourgain and Kenig results in a few directions e.g. showing that the results held for arbitrary bounded singular noise, proving dynamical localization, and treating certain non-stationary potentials.

In the discrete case, this approach faced a considerable obstacle; even in a strictly qualitative sense, there is no unique continuation principle available in this setting for eigenfunctions of $H$. This fact was shown in e.g. \cite{JitErg}. Nevertheless, in \cite{Ding-Smart} it was shown that, with high probability, a certain analogue of this quantitative unique continuation principle holds, for $d=2$ specifically. That is, outside a small probability set of pathological configurations, the necessary lower bounds hold on a large enough subset of the space. This unique continuation principle, together with the new variant of multi scale analysis introduced in \cite{Bourgain2005}, enabled Ding and Smart to show Anderson localization for Schr\"odinger operators on $\Z^2$ with i.i.d. bounded potential.

The work of Ding and Smart has since been extended, to $d=3$ by Li and Zhang in \cite{Li-Zhang}, who produce a version of unique continuation suitable for the lattice $\Z^3$ and thereby showed localization near the bottom of the spectrum. In two dimensions, Li showed that for high disorder localization held outside finitely many small intervals, but only for certain kinds of Bernoulli noise \cite{Li20}. Our work is also an extension of the work of Ding and Smart in  \cite{Ding-Smart}, and our main  contribution is to introduce methods by which to treat potentials which are non-stationary. We believe the requirement that the $V_n$ have uniformly bounded essential range is not necessary; certainly one should have a potential bounded below for the notion of localization at the bottom of the spectrum to make sense but we expect that e.g. one sided Gaussian tails should pose no problem.

A lower bound on the variance is also an essential part of argument. Some control is necessary: if we consider e.g. $V_n$ = $e^{-|n|}\xi_n$, where $\xi_n$ is i.i.d. Bernoulli noise, then such operator surely has absolutely continuous spectrum in $[0,4]$; it is a trace class perturbation of $-\Delta$. It is possible that the situation is more complicated if the decay is slow, as shown in \cite{Delyon1985}, \cite{Kotani1988} in a one-dimensional context.

Because a strict lower bound on the variance is at the very least essential to our proof, we note that it is a ``one scale'' analogue of H\"older continuity of the noise.  This is in fact key to our argument; assuming that $N\leq V_n \leq M$ almost surely, $\var V_n \geq \sigma^2 > 0$ is equivalent to the existence of $\gamma > 0$ and $\rho>0$ so that
\begin{equation}\label{unifac} \sup_{n\in\Z^d, x\in \R} \P[V_n \in [x,x+\gamma]] \leq 1-\rho\end{equation}

Note that H\"older regularity in the context of random Schr\"odinger operators roughly amounts to the existence of $\kappa \in (0,1)$ such that
\begin{equation}\label{holderreg} \sup_{n\in\Z^d, x\in \R} \P[V_n \in [x,x+t]] = O(t^\kappa)\end{equation}
as $t \rightarrow 0$.

For treating non-stationary potentials, FMM is robust enough to recover our localization results if our uniform variance condition is replaced with the assumption of sufficient regularity, as discussed in e.g. \cite{awbook15}. Thus, roughly speaking, our main improvement is to show that one can replace (\ref{holderreg}) with (\ref{unifac}) in the non-stationary setting. To do this, we must use MSA and not FMM. (We note that this only works in the extreme energy regime; in the regular potential setting localization has also been obtained throughout the spectrum in the high disorder regime, which we do not achieve.)

Like the requirement of regularity, the requirement of stationarity is not essential to MSA, but the necessary estimates are often harder to obtain outside the stationary context. There has been significant study of decaying potentials; for potentials with ``slow decay'', localization has sometimes been obtained in e.g. \cite{Delyon1985}, \cite{Kotani1988}, \cite{FGKM}. There has also been study of Delone-Anderson models, where noise in only present at certain sites which are at large scales roughly spatially homogeneous, see e.g. \cite{GMR14}, \cite{Muller2022}. 

These models mentioned thus far are non-stationary, but are fundamentally different models. In ours, strength of the noise does not decay at infinity, and there is noise at every site. For such models in one dimension, the author showed that for ``mild'' non-stationarities which in some sense are spatially localized, usual transfer matrix arguments could be used to obtain localization in \cite{Hurtado23}. Very recently a significant advance was made in the one-dimensional case; Gorodetski and Kleptsyn have shown localization for a very wide class of non-stationary potentials in one dimension; in fact they obtain Anderson localization throughout the spectrum under the conditions \ref{ind}, \ref{unifbdd}, and \ref{unifvar} in \cite{gorodetski2024nonstationary}.  

We note that while our hypotheses on the potentials coincide, the methods we use for our two-dimensional results and which Gorodetski and Kleptsyn use to obtain their one-dimensional results are very different. Indeed, their work is based on the transfer matrix method, a central tool in the study of one-dimensional Schr\"odinger operators. Gorodetski and Kleptsyn leverage a non-stationary Furstenberg theorem proved by the same authors in \cite{gorodetski2023nonstationary} to obtain localization for the systems in question; these methods do not extend to the multi-dimensional setting.

\subsection{Strategy and organization of the paper} As mentioned previously our work roughly follows the strategy of \cite{Ding-Smart} which itself iterated on \cite{Bourgain2005}. The methods of Ding and Smart seem not to admit an easy adaptation to the assumptions we make. The Wegner estimate specifically has a strong combinatorial flavor, and at a glance the unique continuation exploits the stationary structure in a non-trivial way. What is necessary for both (though not sufficient for the Wegner estimate) is precisely the equivalence of condition \ref{unifvar} with (\ref{unifac}) (under the assumption of an almost sure bound). This fact is technically simple, following more or less immediately from a simple application of the second moment method. However, it is also surprisingly powerful and moreover crucial for our argument.

Once the problem has been reformulated in terms of ``uniform anti-concentration'', i.e. in terms of potentials satisfying (\ref{unifac}), unique continuation is obtained by a conceptually straightforward (if somewhat technically involved) implementation of the ideas in \cite[Section 3]{Ding-Smart}. This is not the case for the Wegner estimate: the Wegner estimate of Ding and Smart, \cite[Lemma 4.6]{Ding-Smart} relies very much on specific combinatorial bounds coming from Sperner theory. The ``resonant configurations'' are shown to satisfy a certain combinatorial condition, called the $\rho$-Sperner condition in \cite{Ding-Smart} and $\kappa$-Sperner in this work. In the Bernoulli i.i.d. context, Ding and Smart obtained a probabilistic bound for events with such combinatorial structure. Via the theory of Bernoulli decompositions (specifically results in \cite{Aizenman2009}) one can disintegrate general i.i.d. systems in such a way that it is possible to essentially reduce things to Bernoulli i.i.d., demonstrate the $\kappa$-Sperner condition for the set of resonant configurations, and leverage the existing bound.

Without stationarity, there is no hope of reducing to the Bernoulli i.i.d. case. However, we prove the existence of Bernoulli decompositions satisfying certain quantitative and effective bounds for variables which are almost surely bounded and anti-concentrated in \Cref{decompgoodalt}. Our result is essentially an effective version of the bounds in \cite[Remark 2.1]{Aizenman2009}\footnote{We thank Abel Klein for pointing this out to us.}. This allows us to essentially reduce things to the case of Bernoulli variables, independent and satisfying certain uniform bounds but not identically distributed. After we completed this manuscript, we became aware that the existence of such Bernoulli decompositions was used in \cite{cv}; to our knowledge a full proof of this theorem has not previously appeared in print.

The existing estimates on events with the $\kappa$-Sperner property did not suffice for our purposes, and so we needed to prove versions which worked for more general distributions. We used results of Yehuda and Yehudayoff regarding general product distributions on the discrete hypercube in \cite{yehuda2021slicing} to prove the necessary bound in \Cref{combbound}; said theorem plays the role of \cite[Theorem 4.2]{Ding-Smart} in our non-stationary context and controls the probability of resonances in the proof of the Wegner estimate. Our main technical result in this regard is \Cref{yylym}, which generalizes the Lubell-Yamamoto-Meshalkin type bound \cite[Theorem 3]{yehuda2021slicing}. Once this bound is attained, \Cref{combbound} follows by standard arguments. The decomposition result, together with the bounds so obtained, suffice to prove the Wegner estimate \Cref{bigweg} more or less along the lines laid out in \cite{Ding-Smart}. This summarizes the key technical novelties of our argument, the rest of the localization proof follows the same paper quite closely.

Our paper is organized as follows: In \Cref{prelim}, we introduce key definitions and reformulate our results in terms of uniform anti-concentration, proving the key equivalence between positive variance and said condition. We also discuss the conditions under which our results guarantee localization at the bottom of the spectrum, and introduce probabilistic and asymptotic notation used throughout the paper.

In \Cref{ucsec}, we prove a key lemma for a unique continuation result suitable for uniformly anti-concentrated distributions; after reformulating the problem in terms of uniform anti-concentration as mentioned above, it is a reasonably straightforward adaptation of the ideas in \cite[Section 3]{Ding-Smart}; there are nevertheless some technical details to be worked out. From the key lemma, unique continuation is proved in Appendix \ref{key2uc}; the latter half of the proof is moved to an appendix because our work here is primarily of a clarifying nature; we hew quite closely to \cite{Ding-Smart} for this portion of the proof.

In \Cref{antisec}, we prove the key probabilistic and combinatorial results necessary for the Wegner lemma. We first prove the existence of uniform Bernoulli decompositions for uniformly anti-concentrated variables, \Cref{decompgoodalt}, which may be of some independent interest. We then prove the necessary probabilistic estimates for $\kappa$-Sperner families, \Cref{combbound}.

In \Cref{wegnersec}, we prove a Wegner lemma, similar to the main result of \cite[Section 5]{Ding-Smart}. Our approach is again quite similar to \cite{Ding-Smart}; the main new idea here is the use of Bernoulli decompositions. The introduction of these decompositions introduces technical details not previously present. (We are not the first to use these decompositions in the study of random Schr\"odinger operators, see \cite[Theorem 4.2]{Aizenman2009}. Our use of the decompositions is similar in philosophy, but not particularly similar in the details.)

Finally, in Appendices \ref{detmsasec} and \ref{msasec}, we straightforwardly carry out the program of \cite[Sections 6-8]{Ding-Smart}, which itself iterated on the work in \cite{Bourgain2005} to prove the necessary resolvent bounds in \Cref{endofmsa1}.

\section*{Acknowledgments} We would like to thank Lana Jitomirskaya for suggesting the problem and for many helpful conversations. We would also like to thank Nishant Rangamani and Lingfu Zhang for reading earlier drafts and making helpful suggestions, and Abel Klein and Charles Smart for helpful discussions. Finally, we would like to thank Georgia Tech University and University of California, Berkeley for their hospitality; this work was completed in part during visits to these institutions. This work was supported in part by NSF DMS-2052899, DMS-2155211, and Simons
681675.

\section{Preliminaries}\label{prelim}
We will consider a random Schr\"odinger operator $H = -\Delta + V$, where $-\Delta$ is the discrete Laplacian and $V = \sum_{n\in\Z^2} V_n \delta_n$, with $V_n$ independent but not necessarily identically distributed random variables. Moreover, we make a uniform boundedness assumption that $0 \leq V_n \leq M$, for some $M$ uniform in $n$.

\begin{remm}
    Our choice of lower bound 0 is not so important; any fixed lower bound would in fact suffice and the general case is easily reduced to case of a non-negative potential by an additive normalization that does not meaningfully affect the spectral theory.
\end{remm}

Under our supposition of positivity, we say the $V_n$ are uniformly bounded if there is $M>0$ such that $0 \leq V_n \leq M$ almost surely for all $n$. In this case the spectrum is almost surely contained in $[0,8+M]$. An important concept for us is ``uniform anti-concentration'', which is a quantitative formulation of variables being sufficiently ``non-trivial''. 

\begin{deffo}
    We call a random variable $X$ anti-concentrated with gap $\gamma > 0$ and remainder $\rho > 0$ if
    \begin{equation}\label{unifnoncon} \inf_{\substack{t \in [0,M]\\n \in\Z^2}} \P\left[|X - t| > \frac{\gamma}{2}\right]  \geq \rho \end{equation}
\end{deffo}

\begin{deffo}\label{unifacdef}
We say the potential $V = \{V_n\}_{n\in\Z^2}$ is uniformly anti-concentrated with gap $\gamma> 0$ and remainder $\rho >0$ if
all $V_n$ are anti-concentrated with gap $\gamma$ and remainder $\rho$; more generally we will call a potential uniformly anti-concentrated if there are $\gamma > 0$ and $\rho > 0$ such that it is uniformly anti-concentrated with gap $\gamma$ and remainder $\rho$.
\end{deffo}
(Note that all of these definitions are essentially statements about the L\`evy concentration functions of the variables involved. Because these phenomena are not the main focus of the paper we chose to forgo explicit use of the concentration functions.) Most of our proofs for facts about uniformly anti-concentrated families ultimately amount to showing effective bounds in terms of the parameters $\rho$ and $\gamma$ for any $X$.  Already, we will demonstrate this strategy in proving the following simple but essential fact used throughout our work:
\begin{prop}
If $V_n$ are a family of random variables and there is some $M$ such that $0 \leq V_N \leq M$ almost surely for all $n$, the following are equivalent:
\begin{enumerate}\label{var2ac}
	\item There are $\gamma > 0$ and $\rho > 0$ such that $V_n$ are uniformly anti-concentrated with gap $\gamma$ and remainder $\rho$
	\item There is a uniform positive lower bound for $\text{Var }V_n$.
\end{enumerate}
\end{prop}
One direction is obvious; if the $V_n$ are all anti-concentrated with parameters $\rho$ and $\gamma$, then $\var V_n \geq \rho \gamma^2/4$. The other direction is a corollary of the following:
\begin{prop}\label{pzcor}
	Let $X$ be a random variable with $0 \leq X \leq M$ almost surely and $\var X \geq \sigma^2$. Then
	\[ \P\left[|X-r| \geq \frac{\sigma}{2}\right] \geq \frac{9}{16}\frac{\sigma^4}{\sigma^4+M^4}\]
	for any $r \in \R$.
\end{prop}
\begin{proof} 

First we recall the Paley-Zygmund inequality: for any non-negative random variable $X$ with finite second moment and any $\theta \in [0,1]$, we have
\begin{equation}\label{pz}
    \P[X > \theta \E[X]] \geq (1-\theta)^2\frac{\E[X]^2}{\E[X]^2 + \var X}
\end{equation}

For any $r \in \R$, let $Z_r = |X-r|^2$. $\E[|X-r|^2]$ is minimized by $r = \E X$, so that $\var X \leq \E[Z_r]$. The result follows from an application of the second moment method to the variables $Z_r$.

Applying (\ref{pz}) to $Z_r$, we obtain for all $r \in [0,M]$ and $\theta \in (0,1)$:
	\begin{align*}
		\P[Z_r > \theta \var X] &\geq \frac{(1-\theta)^2 \E[Z_r]^2}{\E[Z_r]^2 + \var Z_r}
        \end{align*}
        Now we consider the function $\frac{x}{x+y}$; note that for $x, y$ both positive, we have $\frac{\partial}{\partial x} \frac{x}{x+y} > 0$ and $\frac{\partial}{\partial y} \frac{x}{x+y} < 0$. In particular, replacing $\E[Z_r]^2$ with something smaller decreases the quantity, and replacing $\var Z_r$ with something larger also decreases the quantity. Hence we obtain:
        \begin{align*}
		\frac{(1-\theta)^2 \E[Z_r]^2}{\E[Z_r]^2 + \var Z_r} &\geq \frac{(1-\theta)^2 (\var X)^2}{(\var X)^2 + \var Z_r}\\
		&\geq \frac{(1-\theta)^2 \sigma^4}{\sigma^4+ M^4}
	\end{align*} As $Z_r > \theta \var X$ implies $|X-r| > \theta^{1/2}\sigma$, we get $\P[|X-r| > \theta^{1/2}\sigma] \geq  \frac{(1-\theta)^2\sigma^4}{\sigma^4 + M^4}$. We conclude by taking $\theta = 1/4$.
\end{proof}

Despite the fact that (in our uniformly bounded context) uniform anti-concentration and uniformly positive variances are equivalent, we introduce the former notion because the specific parameters $\gamma$ and $\rho$ will be used many places in our proofs. Hence going forward, our assumptions on the potential $(V_n)_{n\in\Z^2}$ are that
\begin{enumerate}[label=(\Roman*')]
	\item The $V_n$ are jointly independent\label{aprime}
	\item There is $M > 0$ such that $\P[0 \leq V_n \leq M] = 1$ for all $V_n$\label{bprime}
	\item There are $\gamma > 0$ and $\rho > 0$ such that the $V_n$ are uniformly anti-concentrated with gap $\gamma$ and remainder $\rho$.\label{cprime}
\end{enumerate}
So \Cref{sdlocalization} can be reformulated as follows:
\begin{thm} \label{localizationalt}
	Let $V = (V_n)_{n\in\Z^2}$ be a random potential satisfying conditions \ref{aprime}, \ref{bprime}, and \ref{cprime}. Then there is $E_0 > 0$ (depending on the joint distribution $V$) such that $H = -\Delta +V$ is strongly dynamically localized in expectation for energies in $[0,E_0]$.
\end{thm}
Throughout, we will take $\gamma$, $\rho$ and $M$ as fixed. In particular, many of our constants will have an implicit dependence on these; we will periodically recall this fact. As was mentioned in the introduction, our localization theorem is a consequence of work in \cite[Theorem 1]{rz2023dynamical} together with appropriate resolvent bounds. We briefly discussed this before, but we here define explicitly the truncated operators $H_\Lambda$; letting $P_\Lambda$ be the orthogonal projection onto the subspace spanned by $\delta_n$ for $n \in \Lambda$, $H_\Lambda := P_\Lambda H P_\Lambda$. Then the necessary result on resolvent bounds is the following:
\begin{thm}\label{endofmsa2}
	Given a random potential satisfying \ref{aprime}, \ref{bprime} and \ref{cprime}, for any $0 < \gamma < 1/2$, there are $\alpha > 1 > \ve > 0$ and $E_0>0$ such that for every $0\leq \overline{E} \leq E_0$ and every square $\Lambda \subset \Z^2$ with side length $L \geq \alpha$, we have the bound
	\[ \P[|(H_\Lambda-\overline{E})^{-1}(x,y)| \leq e^{L^{1-\ve}-\ve|x-y|}\text{ for all }x,y \in \Lambda] \geq 1-L^{-\gamma}\]
\end{thm}
Note that (by \Cref{pzcor}) this theorem is just \Cref{endofmsa1} reformulated. We now introduce a sufficient condition for our localization results to be non-vacuous, i.e. that $[0,E_0]$ contains the bottom of the spectrum.
\begin{prop}\label{suffcond}
    Fix $x \geq 0$. If, for any $\varepsilon > 0$, there is $\delta > 0$ such that $\mathbb{P}[x \leq V_n \leq x+\varepsilon] \geq \delta$ holds for a density 1 subset of $\mathbb{Z}^2$, then $[x,x+4] \subset \sigma(H)$ almost surely, so in particular if $x < E_0$,  $[x,\min\{x+4,E_0\}]$ is contained in the spectrum and $H$ is localized in this interval.
\end{prop}
\begin{proof}
	Without loss of generality we will take $x=0$. Throughout this proof, $\|\cdot\|$ denotes the $\ell^2$ norm on $\ell^2(\Z^2)$. We let $q \in [0,4]$ be arbitrary; we will show $q \in \sigma(H)$ almost surely. By a union bound over all $q \in [0,4]\cap \Q$, the result follows. By Weyl's criterion and $q \in \sigma(-\Delta)$, there is for any $\ve > 0$, some $\psi \in \ell^2(\Z^2)$ so that $\|\psi\|=1$ and $\|(-\Delta-q)\psi\|\leq \ve$. Without loss of generality, $\psi$ is compactly supported, and we let $L := 2\,\mathrm{diam}(\mathrm{supp }\,\psi)$, where the diameter is with respect to the $\ell^\infty$ distance on $\Z^2$. In particular, $\mathrm{supp }\,\psi$ is contained within a square $\Lambda \subset \Z^2$ of side length $L$.
	
	Outside a density zero subset $B\subset \Z^2$, one has $\P[V_n \leq \ve/2] \geq \delta$. In particular, there are infinitely many disjoint squares of side length $L$ which do not intersect $B$; were this not the case, $B$ would have lower density at least $1/L^2$. On all the squares of side length $L$ intersecting $B$ trivially, the probability $V_n \leq \ve$ for all sites in the square is at least $\delta^{L^2}$.
	
	By the converse of Borel-Cantelli for independent events, there is almost surely a square of side length $L$ so that $V_n \leq \ve /2$ on the entire square. We let $T$ be a translation sending $\mathrm{supp }\, \psi$ to such a square. Letting $\psi' = \psi \circ T^{-1}$, $\mathrm{supp }\,\psi'$ is contained in such a square, and
	\[ \|(H-q)\psi'\|\leq \ve\]
	Because $\ve$ is arbitrary, $q \in \sigma(H)$ almost surely by Weyl's criterion.
\end{proof}
We now discuss some concrete examples of potentials satisfying the hypotheses of \Cref{localizationalt} and \Cref{suffcond}, i.e. potentials for which our work shows localization at the bottom of the spectrum.
\begin{ex}
	If the $V_n$ are bounded, i.i.d., and $0$ is the bottom of the essential range, we are naturally in this setting; however, our innovations are not needed to treat this case.
\end{ex}
\begin{ex}
	We can treat certain ergodic but non-stationary random potentials, for example periodically random potentials where all the distributions have zero as the bottom of their essential ranges. A small piece of one explicit example is shown in Figure \ref{fig:ergpot}.
	
\begin{figure}[h]\caption{A non-stationary ergodic potential localized at the bottom of the spectrum}\label{fig:ergpot}
	\begin{tikzpicture}[scale=.7]
		\filldraw[black] (0,0) circle (3pt) node{};
		\filldraw[black] (2,0) circle (3pt) node{};
		\filldraw[black] (0,2) circle (3pt) node{};
		\filldraw[black] (2,2) circle (3pt) node{};
		\filldraw[black] (1,1) circle (3pt) node{};
		\filldraw[black](1,3) circle (3pt) node{};
		\filldraw[black] (3,1) circle (3pt) node{};
		\filldraw[black] (3,3) circle (3pt) node{};
		\draw[black] (0,1) circle (3pt) node{};
		\draw[black] (0,3) circle (3pt) node{};
		\draw[black] (1,0) circle (3pt) node{};
		\draw[black] (1,2) circle (3pt) node{};
		\draw[black] (2,1) circle (3pt) node{};
		\draw[black] (2,3) circle (3pt) node{};
		\draw[black] (3,0) circle (3pt) node{};
		\draw[black] (3,2) circle (3pt) node{};
		\draw[white] (3,4) circle (3pt) node{};
		
		\filldraw[black] (5,1) circle (3pt) node[anchor=west]{\quad $\frac{1}{2}(\delta_0+\delta_1)$};
		\draw[black] (5,2) circle (3pt) node[anchor=west]{\quad Uniform on $[0,1]$};
	\end{tikzpicture}
\end{figure}
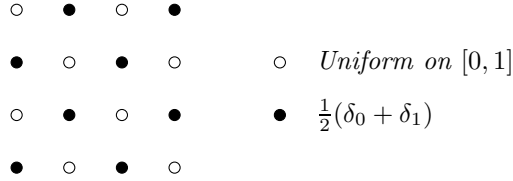
\end{ex}
\begin{ex}
	We can also treat non-ergodic models featuring e.g. an interface between two different types of noise. Again, we show a small piece of an example in Figure \ref{fig:nonergpot}.
	
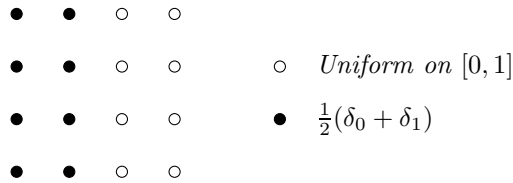
\begin{figure}[h]\caption{A non-ergodic potential localized at the bottom of the spectrum}\label{fig:nonergpot}
	\begin{tikzpicture}[scale = .7]
		\filldraw[black] (0,0) circle (3pt) node{};
		\filldraw[black] (0,1) circle (3pt) node{};
		\filldraw[black] (0,2) circle (3pt) node{};
		\filldraw[black] (0,3) circle (3pt) node{};
		\filldraw[black] (1,0) circle (3pt) node{};
		\filldraw[black](1,1) circle (3pt) node{};
		\filldraw[black] (1,2) circle (3pt) node{};
		\filldraw[black] (1,3) circle (3pt) node{};
		\draw[black] (2,0) circle (3pt) node{};
		\draw[black] (2,1) circle (3pt) node{};
		\draw[black] (2,2) circle (3pt) node{};
		\draw[black] (2,3) circle (3pt) node{};
		\draw[black] (3,0) circle (3pt) node{};
		\draw[black] (3,1) circle (3pt) node{};
		\draw[black] (3,2) circle (3pt) node{};
		\draw[black] (3,3) circle (3pt) node{};
		
		\draw[white] (3,4) circle (3pt) node{};
		
		\filldraw[black] (5,1) circle (3pt) node[anchor=west]{\quad $\frac{1}{2}(\delta_0+\delta_1)$};
		\draw[black] (5,2) circle (3pt) node[anchor=west]{\quad Uniform on $[0,1]$};
	\end{tikzpicture}
	\end{figure}
\end{ex}

\begin{ex}
    Generalizing the previous two examples, if one lets $\mu_1$, $\mu_2$, $\mu_3$, $\dots$, $\mu_N$ be distributions on $\R$ such that they are all compactly supported, all non-trivial (supported on at least two points), and all have $0 \in \mathrm{supp }\,\mu_n\,$, then any potential $\{V_n\}_{n\in\Z^2}$ with all the $V_n$ independent and having law among $\{\mu_1,\dots,\mu_N\}$ will satisfy the hypotheses of Theorems \ref{localizationalt} and \ref{suffcond}, yielding localization at the bottom of the spectrum.
\end{ex}	

We will at times need to condition on collections of variables and $\sigma$-algebras. For any variable $X$ with finite expectation on a space $(\Omega,\P,\mathcal{F})$ and some $\mathcal{F}' \subset \mathcal{F}$, we denote the conditional expectation of $X$ with respect to $\mathcal{F}'$ by $\E[X|\mathcal{F}']$. Given an event $\mathcal{E}$ in the same space, we denote its conditional probability by $\P[\mathcal{E}|\mathcal{F}']$; recall that by definition $\P[\mathcal{E}|\mathcal{F}'] = \E[1_{\mathcal{E}}|\mathcal{F}']$. Outside of \Cref{ucsec}, where we will need to use Azuma's inequality for submartingales, we will not use the $\sigma$-algebraic formalism too much.

By $\E[X|Y=y]$ we mean the expected value of $X$ given the assumption that $Y=y$; in most cases this amounts to considering $\E[X|\sigma(Y)]$ as a function of $y$, and this is how we will discuss conditional quantities for the most part. In all particular places where such expressions arise in this work, well-definition at all $y$ (i.e. regularity of the conditional distribution) is obvious.  $\P[\mathcal{E}|Y=y]$ is defined analogously. In \Cref{bigweg}, we will need to consider expressions of the form $\P[\mathcal{E}|Y=y]$, where $y$ lies outside the essential range of $Y$ for the purposes of an eigenvalue variation argument; as we explain there, there are no issues with well-definition.

We will sometimes make use of asymptotic notation like e.g. $O(n)$, $\Omega(\log n)$, $\Theta(\ve)$; we already have in the introduction. We will avoid using this notation during delicate arguments, instead using explicit constants of the form e.g. $C, C_k$ during these arguments. Nevertheless, when there is no risk of confusion we will freely use e.g. $O(n^2)$, $\Theta(\ve)$, $o(\ell(\Lambda)^{-2})$ with the usual meaning. E.g. $f(n) = O(g(n))$ as $n\rightarrow \infty$ if $\limsup_{n\rightarrow \infty} \frac{f(n)}{g(n)} < \infty$, $f(n) = \Omega(g(n))$ as $n \rightarrow \infty$ if $\liminf_{n \rightarrow \infty} \frac{f(n)}{g(n)} > 0$,  $f(n) = \Theta(g(n))$ if $f(n) = O(g(n))$ and $f(n) = \Omega(g(n))$, and finally $f(n) = o(g(n))$ if $\limsup_{n\rightarrow\infty} \frac{f(n)}{g(n)} = 0$, and the meaning is analogous when some parameter goes to zero.

When we are considering one square in isolation, either a regular square $[a,b]\times [c,d] \subset \Z^2$ or a tilted square (see definitions at the beginning of \Cref{ucsec}), we will denote its side length by $L$ for brevity. We will always use $\Lambda$ and expressions like $\Lambda'$, $\Lambda_n$ to denote standard squares, and $Q$, $Q'$, $Q_n$ to denote tilted squares. When considering multiple squares simultaneously, we let e.g. $\ell(\Lambda)$ denote the side length of $\Lambda$ to avoid ambiguity. For a standard square, side length is exactly the standard meaning; we explain what is meant by side length for tilted squares when we introduce the coordinate system.

Moreover, we often need to consider doublings, halvings, quadruplings, and so on of squares, denoting these by e.g. $Q$. These are the squares with the same center and having twice, or half, or quadruple the area. To be explicit, if $Q = [1,a] \times [1,a]$, $2Q = [1+ (\frac{1}{2}-\frac{\sqrt{2}}{2})a, 1+(\frac{1}{2}+\frac{\sqrt{2}}{2})a]$. Of course, this may yield non-integer coordinates; we simply consider the lattice points lying within the intervals so obtained. 

When we will do some combinatorial arguments regarding anti-chains (in the power set of some finite set ordered with respect to inclusion) we will, for any given finite set $X$ let $2^X$ denote its power set. Any antichains are with respect to the inclusion ordering.

Finally, much of the work concerns the relation between various called and asymptotic estimates as the ``scales'' get large. For example, in both \Cref{wegnersec} and Appendix \ref{msasec}, we will use ``for sufficiently large scales'' to mean for a sufficiently large starting scale; in these contexts the starting scale is denoted by $L_0$. In fact, in a technical sense, one should read ``for sufficiently large scales'' as ``for $L_0$ sufficiently large''; we keep this language of large scales to convey physical intuition.

\section{Unique continuation key lemma}\label{ucsec}

Throughout we use the tilted coordinates for $\Z^2$ given by $(s,t) = (x+y,x-y)$, following \cite{Ding-Smart}. In this coordinate system, we let $R_{[1,a],[1,b]}$ denote the tilted rectangle \[R_{[1,a],[1,b]}:=\{(x,y) \in \Z^2\,:\,(s,t) \in [1,a] \times [1,b]\}\] For the rectangle $R_{[a_1,a_2],[b_1,b_2]}$, we say its dimensions are $(a_2-a_1+1)\times (b_2 -b_1 + 1)$. 
In the standard coordinates, the equation $H\psi = E\psi$ can expressed locally as
\[ (4+V_{n_1,n_2}-E)\psi_{n_1,n_2} = \psi_{n_1-1,n_2}+\psi_{n_1+1,n_2}+\psi_{n_1,n_2-1}+\psi_{n_1,n_2+1}\]
In the tilted coordinates, it can be represented as:
\begin{equation}\label{localhamtilt}
	\psi_{s,t} = (4+V_{s-1,t-1}-E)\psi_{s-1,t-1} -\psi_{s-2,t-2} - \psi_{s-1,t+1} - \psi_{s+1,t-1}
\end{equation}

The following notions were introduced in \cite{Ding-Smart}:

\begin{deffo} We say $F$ is $\ve$-sparse in $R$ if for all $k \in \Z$, we have
\[ | F \cap R \cap \{(s,t) \in \Z^2 \,:\, s = k\}| \leq \ve |R \cap \{(s,t) \in \Z^2 \,:\, s = k\}|\]
and
\[ | F \cap R \cap \{(s,t) \in \Z^2 \,:\, t = k\}| \leq \ve |R \cap \{(s,t) \in \Z^2 \,:\, t = k\}|\]
\end{deffo}
This sparsity is precisely sparsity along diagonals. For tilted squares (tilted rectangles with equal side lengths) we have a notion of ``regularity''  at all scales formulated in terms of this sparsity.
\begin{deffo} We call $F$ $\ve$-regular in a tilted square $Q$ if any disjoint union of $Q_1, \dots, Q_n \subset Q$ with $F$ not $\ve$-sparse in all of the $Q_i$ forces
\[ |\cup Q_i | \leq \ve |Q|\]
\end{deffo}
Roughly, $F$ is $\ve$-regular if it is impossible to cover a significant proportion of $Q$ with disjoint squares where $F$ covers a significant portion of some diagonal.  With these notions defined, we can state our unique continuation theorem; this theorem is a generalization of \cite[Theorem 3.5]{Ding-Smart} for the most part, but it differs in that it considers a wider class of potentials, namely bounded and uniformly anti-concentrated potentials rather than only i.i.d. potentials.
\begin{thm}\label{ucucthmfinal}
	Let $V$ be a random potential defined by a uniformly bounded and anti-concentrated family of independent random variables $\xi_n$, i.e. satisfying \ref{aprime}, \ref{bprime} and \ref{cprime}. We let $H$ be the random operator $H= -\Delta + V$, and moreover we fix $\overline{E} \in [0,8+M]$, $\Lambda\subset \Z^2$ be a square of side length $L$, and $\mathcal{E}_{\text{uc}}(\Lambda,F, \alpha,\ve)$ denote the event that
	\begin{equation}
	\begin{cases}\label{doublesmallnesshyp}
		|E - \overline{E}| \leq e^{-\alpha (L \log L)^{1/2}}\\
		H\psi = E \psi \text{ on } \Lambda\\
		|\psi| \leq 1 \text{ on a } 1-\ve (L \log L)^{-1/2} \text{ fraction of } \Lambda\setminus F
	\end{cases}
	\end{equation}
	implies $ |\psi| \leq e^{\alpha L \log L}$ in $\frac{1}{2}\Lambda$.\\
	
\noindent Then for any $\ve > 0$ sufficiently small, there exists $\alpha = \alpha(\ve)$ such that if
	\begin{enumerate}[label=(\Alph*)]
		\item $L \geq \alpha$
		\item $F$ is $\ve$-regular in $\Lambda$
	\end{enumerate}
	we have $\P[\mathcal{E}_{\text{uc}}(\Lambda,F,\alpha,\ve)|V_F = v] \leq 1- e^{-\ve L^{1/4}}$
	for any $v:F\rightarrow [0,M]$. The requisite smallness of $\ve$ and largeness of $\alpha$ depend on the particular values of $\rho$, $\gamma$ and $M$.
\end{thm}

\begin{remm}
	The set $F$ is a collection of ``frozen'' sites. As the MSA is carried out, it is necessary to sample certain sites to find bounds on the probability of certain events at each stage of the inductive procedure. At each stage, we presume that the requisite bounds hold at previous stages. In particular, it is crucial that we not sample the same sites more than once. Hence a site, once sampled, becomes ``frozen'', and our estimates must hold regardless of the value of the potential at the frozen sites. (We are allowed to assume that they obey the almost sure bound, but otherwise cannot use information regarding these sites.) This idea goes back to \cite{Bourgain2005}, and it introduces non-trivial difficulties to the proof of unique continuation.
\end{remm}
This event depends on the parameters $\alpha$ and $\ve$, and the same is true for many events under consideration. We will generally make this dependence explicit in theorem statements, but suppress this in the notation during the course of proofs. We rarely need to vary $\alpha$ and $\ve$ beyond enlarging the former and shrinking the latter to ensure the validity of certain estimates.

The proof more or less follows that in \cite{Ding-Smart} where the $V_n$ are i.i.d. Bernoulli-$\frac{1}{2}$ variables. A few results are deterministic and so their proofs require no modification. These facts require notions of boundary and interior for tilted rectangles; the notion of the western boundary, at least implicitly, was used in \cite{Buhovsky22}, where the question of unique continuation was studied under the assumption $V\equiv 0$.
\begin{deffo}
	The west boundary of $R = R_{[1,a],[1,b]}$ is given by
	\[ \partial^{w}R = R_{[1,a],[1,2]} \cup R_{[1,2],[1,b]}\]
\end{deffo}
Also important, and used extensively in \cite{Ding-Smart}, \cite{Buhovsky22}, is what we will call the interior.
\begin{deffo}
	The interior of $R$ is given by
	\[ R^{\circ} = R_{[2,a-1],[2,b-1]}\]
\end{deffo}
We now state three results about how $\psi$, an eigenfunction of $H_R$ associated to some $E$ ``propagates'' from the western boundary. Throughout, $c_k$ (respectively $C_k)$ are small (respectively large) universal constants which do not meaningfully depend on any of the parameters involved in the inequalities they appear in. They will be allowed to depend on certain coarse features of the potential $V$, i.e. on $M$, $\rho$, and $\gamma$, but they are completely universal for fixed values of these quantities.
\begin{lem}[\cite{Ding-Smart}]\label{unex}
	 For fixed $E \in [0,8+M]$ and $R = R_{[1,a],[1,b]}$, given any initial data $\psi : \partial^{w}R \rightarrow \R$, there is a unique extension $\psi: R \rightarrow \R$ satisfying $H \psi = E \psi$ on $R^\circ$. Moreover, the mapping from the initial data to the extension is linear.
\end{lem}

Not only is there a unique extension, but its magnitude can be controlled in terms of the dimensions of $R$ and the magnitude of the initial data.
\begin{lem}[\cite{Ding-Smart}]\label{unexbound}For $E \in [0,8+M]$, $R= R_{[1,a],[1,b]}$, and $\psi$ solving $H\psi = E \psi$ in $R^\circ$, we have
	\[\|\psi\|_{\ell^\infty(R)} \leq e^{C_1 b\log a} \|\psi\|_{\ell^\infty(\partial^{w}R) }\]
\end{lem}

Of course, the value of $C$ will crucially depend on $M$, the almost sure bound on the potential. We can  take $C_1 = C\max\{1,\log M\}$, for some $C$ universal. Because $M$ is more or less ``universal'' for any fixed potential, this does not cause any serious issues or even require particularly attentive bookkeeping. Because for any fixed $E$ and $R$ the mapping from initial data to extension is linear, a straightforward consequence of this is that if a perturbation of our initial data $\psi|_{\partial^{w}R}$ corresponds to a perturbation of $\psi$ on all of $R$ by at most $e^{Cb\log a}$ times the norm of the perturbation carried out on our initial data. We have a similar result for making a small perturbation in energy instead of initial data.
\begin{lem}[\cite{Ding-Smart}]\label{envarbound} If $E_0, E_1 \in [-0,8+M]$, and $\psi_0$, $\psi_1$ are solutions to $H\psi_i = E_i \psi_i$ on $R^\circ$ with common initial data $\psi:\partial^{w}\rightarrow \R$, then we have
	\[ \|\psi_0 - \psi_1\|_{\ell^\infty(R)} \leq e^{C_1 b\log a} |E_1-E_0| \|\psi\|_{\ell^\infty(\partial^{w}R)}\]
\end{lem}

(One proves this lemma using the former and obtains a worse constant; we will just take the worse of the two and proceed with said constant in order to reduce the number of constants in play.) These results are instrumental in proving the key lemma, where virtually all the differences from the stationary case will arise insofar as the unique continuation result is concerned. (Note that these differences are more or less technical details and the ideas are the same; nevertheless the proof of the key lemma is  where the technical details appear.) Our key lemma, which plays the same role as \cite[Lemma 3.12]{Ding-Smart}, is:
\begin{lem}\label{key}
	Let $H$ be as in the hypothesis of \Cref{ucucthmfinal}. We fix $\overline{E} \in [0,8+M]$, and let $\mathcal{E}_{\text{ni}}(R_{[1,a],[1,b]}, \alpha,\ve,F)$ denote the event that
	\[ \begin{cases}
		|E-\overline{E}| \leq e^{-\alpha b \log a}\\
		H\psi = E \psi \text{ in } R^\circ \\
		|\psi| \leq 1 \text{ in } R_{[1,a],[1,2]}\\
		|\psi| \leq 1 \text{ in a } 1-\ve \text{ fraction of } R_{[1,a],[b-1,b]}\setminus F
	\end{cases}\]
	implies $|\psi| \leq e^{\alpha b \log a}$. Then there are $\alpha > 1 > \ve$ such that if
	\begin{enumerate}[label=(\Alph*)]
		\item $ a \geq \alpha b^2 \log a \geq \alpha^2$
		\item $F \subset \Z^2$ is $\ve$-sparse in $R_{[1,a],[1,b]}$
	\end{enumerate}
	then $\P[\mathcal{E}_{\text{ni}}(R_{[1,a],[1,b]},\alpha,\ve,F)|V_F = v] \geq 1 - e^{-\ve a}$ for any $v:F\rightarrow [0,M]$.
	The values of $\alpha$ and $\ve$ will depend on the constants $M$, $\rho$ and $\gamma$.
\end{lem}
\begin{proof}
We begin by using our results regarding propagation from the western boundary to make a reduction. Specifically, we define the event $\mathcal{E}_{\text{ni}'}(R,\alpha,\ve,F)$ as the event (set of operators) for which
\[\begin{cases}
	H\psi'' = \overline{E} \psi \text{ in } R^\circ\\
	\psi'' = 0 \text{ in } R_{[1,a],[1,2]}\\
	\max_{R_{[1,2],[1,b]}}|\psi''| \geq 1
\end{cases}\]
implies $|\psi| \geq e^{-\frac{1}{2}\alpha b\log a}$ in a $2\ve$ fraction of $R_{[1,a],[b-1,b]}\setminus F$. We will show that for sufficiently large $\alpha$, we in fact have $\mathcal{E}_{\text{ni}'}(R) \subset \mathcal{E}_{\text{ni}}(R)$, and hence that it suffices to bound the former below.

To show this inclusion, we suppose $|E - \overline{E}| \leq e^{-\alpha b \log a}$, $H\psi = E \psi$, $|\psi| \leq 1$ on $R_{[1,a],[1,2]}$ and $|\psi| \leq 1$ on a $1-\ve$ fraction of $R_{[1,a],[b-1,b]}\setminus F$. (Note that $\mathcal{E}_{\text{ni}}$ says precisely that under these conditions, $\psi$ obeys a specific bound on $R$. Hence $\mathcal{E}_{\text{ni}'} \subset \mathcal{E}_{\text{ni}}$ follows if these assumptions together with $H \in \mathcal{E}_{\text{ni}'}$ imply the bound.) 

We define $\psi'$ and $\psi''$ as follows (uniqueness and existence are given by \Cref{unex}):
\[
\begin{cases}
	H\psi' = E \psi' \text{ in } R^\circ\\
	\psi' = 0 \text{ in } R_{[1,a],[1,2]}\\
	\psi' = \psi \text{ in } R_{[1,2],[3,b]}
\end{cases}\]
and
\[ \begin{cases}
	H\psi'' = \overline{E}\psi''\\
	\psi'' = 0 \text{ in } R_{[1,a],[1,2]}\\
	\psi'' = \psi \text{ in } R_{[1,2],[3,b]}
\end{cases}\]
If $H \in \mathcal{E}_{\text{ni}'}$, we necessarily have $|\psi''| \geq e^{-\frac{1}{2}\alpha b \log a} \max_{R_{[1,2],[3,b]}}|\psi|$ in a 2$\ve$ fraction of $R_{[1,a],[b-1,b]}$. The energy variation bound in \Cref{envarbound} gives us
\[ \max_R |\psi' - \psi''| \leq e^{(C_1- \alpha )b \log a} \max_{R_{[1,2],[3,b]}}|\psi|\]
whence
\begin{align*}
	|\psi'| &\geq |\psi''| - |\psi' - \psi''|\\
	& \geq \left(e^{-\frac{1}{2}\alpha b \log a} - e^{(C_1-\alpha)b \log a}\right) \max_{R_{[1,2],[3,b]}}|\psi|\\
	&\geq \frac{1}{2} e^{-\frac{1}{2} \alpha b \log a} \max_{R_{[1,2],[3,b]}}|\psi|
\end{align*}
on a 2$\ve$ fraction of $R_{[1,a],[b-1,b]}\setminus F$, so long as $\alpha$ is sufficiently large, e.g. $\alpha > 2C_1$. We also have, by our unique extension bound in \Cref{unexbound}, that:
\begin{align*}|\psi| &\geq \Big| |\psi'|-|\psi-\psi'|\Big|\\
&\geq \max\{|\psi'| - e^{C_1b\log a},0\} \\
&\geq |\psi'|-e^{C_1b\log a} \end{align*}
holds on all of $R$, and in particular on $R_{[1,a],[b-1,b]}$.

By combining the two estimates we get on a $\ve$ proportion of sites on $R_{[1,a],[b-1,b]}\setminus F$ (and so in particular on at least one site) the following:
\[
	1 \geq |\psi| \geq \frac{1}{2}e^{-\frac{1}{2}\alpha b \log a}\max_{R_{[1,2],[3,b]}}|\psi|  -e^{C_1b\log a}
\]
which immediately yields $\max_{R_{[1,2],[3,b]}}|\psi| \leq 2e^{(C_1+\frac{1}{2}\alpha)b\log a}$. (Presuming $C_1b\log a > \log 2$.) 
Using \Cref{unexbound} again (recall that we have by hypothesis a bound on $\psi$ restricted to $R_{[1,a],[1,2]}$) gives:
	\[\max_R |\psi| \leq e^{(2C_1+\frac{1}{2}\alpha )b\log a} \]
	
Hence $\mathcal{E}_{\text{ni}'} \subset \mathcal{E}_{\text{ni}}$ so long as $\alpha$ is chosen sufficiently large; at this point $\alpha \geq 4C_1$ suffices. The rest of the proof is concerned with estimating $\mathcal{E}_{\text{ni}'}$. Towards accomplishing this, we consider certain random linear mappings determined by the potential, distinct from our $H$. The event $\mathcal{E}_{\text{ni}'}$ is entirely concerned with $\psi$ which are zero on $R_{[1,a],[1,2]}$. Such $\psi$ are entirely determined by their values on $R_{[1,2],[3,b]}$, what remains of $\partial^{w}R$.

Henceforth we call this data $\psi^0$, and so $\psi^0$ are maps $R_{[1,2],[b-1,b]} \rightarrow \R$. We are also concerned with the amount of large support in $R_{[1,a],[3,b]}$, where we hope to find at least a $2\ve$ fraction of sites in this region with $|\psi| \geq 1$. Henceforth we refer to $\psi$ restricted to this region by $\psi^1$.

Under the assumption that $\psi = 0$ on $R_{[1,a],[1,2]}$, $\psi^0$ totally determines $\psi^1$ for any fixed potential. So $\psi^0 \mapsto \psi^1$ is a random linear map, with its randomness coming from the potential. In these terms, $V \in \mathcal{E}_{\text{ni}'}$ can be reformulated as
\[ \inf_{\|\psi^0\|_\infty = 1} |\{|\psi^1| \geq e^{-\frac{1}{2}\alpha b \log a}\}| \geq 2\ve |R_{[1,a],[b-1,b]}|\]
with the implicit dependence on $V$ of this expression coming from the random map sending $\psi^0$ to $\psi^1$. Towards showing that this bound holds with high probability, we first show that the bound below
\[ |\{|\psi^1| \geq e^{-\frac{1}{2}\alpha b \log a}\}| \geq 2\ve |R_{[1,a],[2,b]}|\]
holds with high probability for any fixed $\psi^0$ with $\|\psi^0\|_{\ell^\infty} = 1$.

We make use of an entirely deterministic result from \cite{Ding-Smart}.
\begin{claim}\label{bigrow}
    There is a positive constant $C_2$ (depending on $M, \rho, \gamma$) so that for any $\psi^0: R_{[1,2],[3,b]} \rightarrow \R$, there is $(s_0,t_0) \in R_{[1,2],[3,b]}$ such that
	\[ |\psi| \geq e^{-C_2b\log a} \|\psi^0\|_\infty \text{ in } R_{[1,a],[s_0,t_0]}\]
\end{claim}

(The precise result is \cite[Claim 3.15]{Ding-Smart}.) More concisely, there is a ``tilted'' row where $\psi$ is not too small regardless of the choice of potential. From here we propagate this largeness down to $R_{[1,a],[b-1,b]}$ via the structure of $H$. This claim in particular is where the technical differences between our work and that of \cite{Ding-Smart} arise.
\begin{claim}
    There are positive constants $C_3$ and $c_3$ (allowed to depend on $M$, $\rho$ and $\gamma$ such that for any $\psi^0: R_{[1,2],[3,b]}$  with $\|\psi^0\|_\infty =1$  such that
	\begin{equation} \P[|\{|\psi^1| \geq e^{-C_3b\log a} \}| \geq \ve a\,|\,V_F = v] \geq 1-e^{-c_3a}
	\end{equation}
\end{claim}
We let $(s_0,t_0) \in R_{[1,2],[3,b]}$ be a point satisfying the conclusion of \Cref{bigrow}. Suppose that $(s',t')$ is some point in $R_{[3,a],[b-1,b]}$ with $(s'-1,t_0) \in R_{[2,a],[t_0,t_0]} \setminus F$. As a consequence of $H\psi = \overline{E}\psi$ and $\psi= 0$ on $R_{[1,a],[1,2]}$, we obtain
\begin{equation}\label{thing} \psi_{s',t'} = - \psi_{s'-2,t'} + \sum_{0 \leq k \leq \frac{t'-1}{2}} (-1)^k (4-\overline{E} + V_{s'-1,t'-1-2k}) \psi_{s'-1,t'-1-2k}\end{equation}

In particular, this follows by induction, (\ref{localhamtilt}) and our assumption that $\psi \equiv 0$ on $R_{[1,a],[1,2]}$. Because of the way in which solutions propagate, in general $\psi_{s,t}$ is determined entirely by the initial data $\psi|_{\partial^wR}$ together with the potential on $R_{[1,s-1],[1,b]} \cup R_{[1,a],[1,t-1]}$. In this particular case where the boundary data on $R_{[1,a],[1,2]}$ is zero however, we in fact have that $\psi_{s,t}$ is entirely determined by $\psi^0$ together with the potential on $R_{[1,s-1],[1,b]}$. In particular, for fixed $\psi^0$ the random functions $\psi_{s'-1, t'-1-2k}$ are measurable with respect to the $\sigma$-algebra generated by the potential on $R_{[1,s'-2],[1,b]}$. The first important consequence of this is that the dependence of $\psi_{s',t'}$ on $V_{s'-1,t_0}$ is entirely via the term $(-1)^{t'-t_0-1}(4-\overline{E}+V_{s'-1,t_0})\psi_{s'-1,t_0}$.

We rewrite (\ref{thing}) as
\[ \psi_{s',t'} = K \pm V_{s'-1,t_0}\psi_{s'-1,t_0}\]
collapsing all the terms not depending on $V_{s'-1,t_0}$ into $K$. We then have $|\psi_{s',t'}| \leq e^{-C_3b\log a}$ equivalent to
\[ |V_{s'-1,t_0}\psi_{s'-1,t_0} \pm K| \leq 2e^{-C_3b\log a}\]
or
\[ V_{s'-1,t_0} \in \left[\frac{\pm K-e^{-C_3b\log a}}{|\psi_{s'-1,t_0}|},\frac{\pm K+e^{-C_3b\log a}}{|\psi_{s'-1,t_0}|}\right]\] 
Hence, in particular, smallness of $\psi_{s',t'}$ is contingent on $V_{s'-1,t_0}$ falling in some small interval of length $\frac{2e^{-C_3b\log a}}{|\psi_{s'-1,t_0}|}$. However, by our lower bound on $|\psi_{s'-1,t_0}|$, we in fact have that the interval has size at most $2e^{(C_2-C_3)b\log a}$. As long we have:
\[ C_3 \geq C_2 - \frac{\log(\gamma)-\log 2}{\alpha^2}\]
we obtain $2e^{-(C_3-C)b\log a} \leq \gamma$ for $b\log a \geq \alpha ^2$. (Recall that $\gamma$ is one of the parameters witnessing the uniform anti-concentration of the variables $V_{s,t}$, see (\ref{unifnoncon}).) Then in particular, for $b\log a$ sufficiently large and $C_3$ so chosen
\[\P[|\psi_{s',t'}| \geq e^{-C_3b\log a}] \geq \rho\]
where $\rho$ is the other parameter witnessing the uniform anti-concentration. Moreover, because this analysis relied entirely on the study of $V_{s'-1,t_0}$, if we let $\mathcal{F}'$ denote the $\sigma$-algebra generated by the potential on $R_{[1,s'-2],[1,b]} \cup F$, we have
\[ \P[|\psi_{s',t'}| \geq e^{-C_3 b\log a}\,|\, \mathcal{F}'] \geq \rho\]
Now we let $s_1,\dots s_K$ be an increasing enumeration of the $(s,t) \in R_{[3,a],[b-1,b]}$ with $(s-1,t_0) \notin F$. Note that there is only one possible choice of $t$ because of the structure of the tilted coordinates, so that either all these points are $(s_k,b-1)$ or all of them are $(s_k,b)$, and so we call this common value $t_1$. If $\mathcal{F}_k$ is defined to be the $\sigma$-algebra generated by the potential on $R_{[1,s_k-2],[1,b]}$, then by the earlier discussion we have
\[ \P[|\psi_{s_k,t_1}| \geq e^{-C_3b\log a} | \mathcal{F}_k] \geq \rho\]

We now let $I_k$ be the indicator events
\[ I_k = \begin{cases} 0 \text{ if } |\psi_{s_k,t_1}| < e^{-C_3b\log a} \\ 1 \text{ if } |\psi_{s_k,t_1}| \geq e^{-C_3b\log a}\end{cases}\]
and $S_k = \sum_{j=1}^k (I_j - \rho)$. Clearly $S_k$ is a submartingale with respect to the filtration $\mathcal{F}_k$. In fact, if we let $\mathcal{F}$ denote the $\sigma$-algebra generated by the potential on $F$, then the same thing is true if we consider $\tilde{S}_k = \E[S_k|\mathcal{F}]$ and the filtration of $\tilde{\mathcal{F}}_k$ with $\tilde{\mathcal{F}}_k = \sigma(\mathcal{F}\cup \mathcal{F}_k)$. Henceforth we consider $\tilde{S}_K$, noting that
\[ \tilde{S}_K + K\rho \leq |\{|\psi^1| \geq e^{-C_3 b \log a}\}|\]

\noindent Using the Azuma inequality, we get for any $\ve' > 0$
\[ \P[ \tilde{S}_K + K \rho \leq  K (\rho-\ve') ] \leq e^{-\frac{K(\ve')^2}{2}}\]
By the sparsity of $F$, we have $K \geq \frac{1}{3}(1-\ve) a$, where $\ve$ is the sparsity parameter. Assuming e.g. $\ve < \frac{1}{2}$, we obtain $\frac{1}{6}a \leq K \leq a$ and hence
\[ \P[\{|\psi^1| \geq e^{-C_3b \log a} \}| \leq a(\rho - \ve')] \leq e^{-\frac{a(\ve')^2}{12}}\]
\noindent Letting $\ve'' = (\rho - \ve')$, we obtain
\[ \P[\{|\psi^1| \geq e^{-C_3b \log a} \}| \leq a\ve''] \leq e^{-\frac{1}{12}a(\rho-\ve'')^2}\]
For all $\ve''< \frac{\rho}{2}$ we have
\[\frac{1}{12}a(\rho - \ve'') \geq \frac{1}{48}a\rho^2\]
Hence for small $\ve''$ we have
\[ \P[\{|\psi^1| \geq e^{-C_3b \log a} \}| \leq ac\ve''] \leq e^{-c_3a}\]
with $c_3 = \frac{\rho^2}{48}$.
This is precisely what was sought; $|R_{[1,a],[b-1,b]}| \geq \frac{1}{3}a$. Hence if we set $\ve'' = 6\ve$, we have
\[ \P[\{|\psi^1| \geq e^{-C_3b \log a} \}| \leq 2\ve |R_{[1,a],[b-1,b]}] \leq e^{-c_3a}\]
and moreover because everything can be done using the conditioned variables $\tilde{S}_k$, we obtain
\[ \P[\{|\psi^1| \geq e^{-C_3b \log a} \}| \geq 2\ve |R_{[1,a],[b-1,b]} |\,|\, V_F = v] \geq 1- e^{-ac}\]
as was sought. (Because we assume $\ve''\leq \rho/2$, this imposes the condition $\ve \leq \frac{\rho}{12}$.)

To finish the proof we establish the existence of $\tilde{X}$ a large finite subset of the $\ell^\infty(R_{[1,2],[3,b]})$ unit ball so that for any $\psi^0$ with $\|\psi^0\|_\infty = 1$, we have the existence of some $\tilde{\psi}^0 \in \tilde{X}$ with $\|\psi^0 - \tilde{\psi}^0\|_\infty \leq e^{-Cb\log a}$. Specifically, we have the following:
\begin{claim}
	For any $\beta > 0$, there is $\tilde{X} \subset \{ \|\psi^0\|_\infty =1 \}$ with $|\tilde{X}| \leq e^{2\beta b^2\log a}$ and \[\sup_{\|\psi^0\|_\infty = 1}\inf_{\tilde{\psi}^0 \in \tilde{X}}  \|\psi^0 - \tilde{\psi}^0\|_\infty \leq e^{-2\beta b\log a} \]
    so long as $b$ is sufficiently large; the requisite largeness of $b$ only depends on $\beta$, and not on $M$, $\rho$ or $\gamma$.
\end{claim}

We briefly let $\eta= e^{-2\beta b \log a}$, and recall some basic notions from metric geometry.

\begin{deffo}
    A (non-empty) subset $\tilde{X}$ of a metric space $(x,d)$ is called an $\eta'$-net (for $\eta' > 0$) if
    \[ \sup_{x \in X} \inf_{x' \in X'} d(x,x') \leq \eta' \]
\end{deffo}

\begin{deffo}
    A subset $\tilde{X}$ of a metric space $(X,d)$ is called an $\eta'$-packing if
    \[ \inf_{\substack{x,x' \in \tilde{X}\\ x \neq x'}} d(x,x') \geq \eta'\]
\end{deffo}

We will also use the following well-known and elementary fact:
\begin{prop}\label{packing2net}
    Any maximal $\eta'$ packing (with respect to inclusion) is an $\eta'$-net. 
\end{prop}
Because of its brevity we include the proof: presume an $\eta'$-packing is not an $\eta'$-net. Then there is a point $x$ at a distance $\eta'$ from all points in our packing, and so in particular one can add this point and preserve the packing property. Hence if a packing is not a net, it is not maximal. From this proposition, we see that to demonstrate the existence of small $\eta$-nets (which is what the claim  seeks to do) it suffices to show a bound on the size of $\eta$-packings.

The claim essentially follows from the fact the in a $d$-dimensional space, the surface area of a unit ball of radius $\eta$ in in the unit sphere is on the order of $\eta^{d-1}$; however, since we are dealing with the $\ell^\infty$ unit ball there is no natural notion of surface area. Moreover, we are not just taking $\eta' \to 0$ with the dimension fixed; we are concerned with the case where $\eta'$ is exponentially small in the dimension, as $d \rightarrow \infty$. We will use a bound which takes both of these into account. 

We let $\Phi$ be the map $x \mapsto x/\|x\|_\infty$; this gives a homeomorphism from the $\ell^2$ unit sphere to the $\ell^\infty$ unit sphere. We let $m$ be the pushforward of surface area via $\Phi$. Then the claim follows more or less from Proposition \ref{packing2net} and the following bound, valid for any $0 < \eta' < 1/2$ with some completely universal $c$:

\begin{equation}\label{packingbd}
    \frac{m(\{\,\|x\|_\infty =1\,\text{ and } \|x-y\|_\infty \leq \eta'\})} {m(\{\|x\|_\infty = 1\})} \geq c (\eta')^{d-1} d^{\frac{d-1}{4}}
\end{equation}

We do not provide a comprehensive proof of this bound; it is tedious but straightforward. We briefly discuss the main ingredients. One of the key steps is to transport the problem from $\ell^\infty$ to $\ell^2$, via the map $\Phi$. The following basic bound for vectors in $\R^d$ is necessary:
\[ \|x\|_\infty \leq \|x\|_2 \leq \sqrt{d} \|x\|_\infty \]
This means that any $\eta'$ ball in $\ell^\infty$ has preimage under $\Phi$ containing an $\eta'/\sqrt{d}$ ball. Up to introducing additional polynomial (and so for our purposes easily controlled) factors, the problem is now reduced to $\ell^2$. There are formula for the areas of spheres in dimension $d$ and spherical caps of a given angle in terms of special functions; we use formulae for the latter from \cite{li2010concise}. Combined with inequalities for the relevant special functions from \cite{cerone2007special}, the bound follows.

Note that for any $\eta'$-packing, the $\eta'$ neighborhoods are disjoint. (Our particular definitions do allow some overlap, but of positive codimension, which can be ignored measure theoretically.) In particular, a lower bound on the mass of $\eta$-balls in proportion to the mass of the whole unit sphere yields an upper bound on the size of any $\eta$-packing. Plugging in $\eta$ and $d = |R_{[3,a],[1,b]}| \in [b/2, b]$ to \ref{packingbd} , we get an upper bound on the size of any $\eta$-packing: the largest any such packing can be is $Ce^{\beta b^2 \log a}b^{b/4}$, with $C$ absolute. Taking $b$ sufficiently large (depending only on $\beta$) one can absorb the terms $Cb^{b/4}$ by increasing the constant in the exponent of the other term, giving the upper bound $e^{2\beta b^2 \log a}$ for all $b$ sufficiently large, proving the claim.

Combining this claim with our bound from \Cref{unexbound} yields
\[ \inf_{\|\psi^0\|_\infty = 1} |\{|\psi^1| \geq e^{-\beta b \log a} - e^{(C_1-2\beta)b\log a}\}| \geq \inf_{\tilde{\psi}^0 \in \tilde{X}} |\{|\tilde{\psi}^1| \geq e^{-\beta b \log a}\}| \]

Hence as long as $\beta$ is sufficiently large, we have
\begin{align*} &\P\left[ \inf_{\|\psi^0\|_\infty = 1} |\{|\psi^1| \geq e^{-\frac{1}{2} \beta \log a}\}| \geq 2\ve |R_{[1,a],[b-1,b]}|\right]  \geq \\
& \P\left[ \inf_{\tilde{\psi}^0 \in \tilde{X}} |\{|\tilde{\psi}^1| \geq e^{-2 \beta \log a}\}| \geq 2\ve |R_{[1,a],[b-1,b]}|\right]
\end{align*}
whence (possibly assuming $\beta$ even larger) we can bound the right hand side from below by $1 - e^{C\beta b^2\log a-ca}$ by a union bound. Taking $\alpha$ large enough compared to $C\beta b^2 \log a$, we obtain a bound of $1-e^{-\ve a}$ for $a > \alpha$. Moreover, while we omitted this from the notation for the sake of space, this all holds true on the level of conditioning on $V_F =v$, proving the key lemma.\end{proof}

We have thus far followed \cite{Ding-Smart} somewhat closely, but we have made some changes. Having proven \Cref{key}, \Cref{ucucthmfinal} follows using more or less the exact strategy used by Ding and Smart in the original paper \cite{Ding-Smart} to prove their Theorem 3.5 from their Lemma 3.12. The rest of the details will be carried out in Appendix \ref{key2uc}.

\section{Anti-concentration estimates}\label{antisec}
We recall in this section certain notions introduced in \cite{Ding-Smart}, and prove probabilistic results necessary for the proof of the Wegner lemma in \Cref{wegnersec}. We will be working always with certain measures on $2^{\{1,\dots,N\}}$ (the powerset of $\{1,\dots,N\}$) which have a product structure. To clarify things, it is useful to identify subsets $\boldsymbol{\xi} \subset \{1,\dots,N\}$  with elements of the discrete hypercube $\{0,1\}^N$ in the natural way: we abuse notation and use $\boldsymbol{\xi}$ to also denote the vector $(\xi_1,\dots,\xi_N) \in \{0,1\}^N$ with $\xi_k =1$ if and only if $k \in \boldsymbol{\xi}$. We will generally use $\boldsymbol{\xi}$ for a random subset/element of the discrete hypercube, and $\boldsymbol{\eta}$ for a deterministic one in this section. We will also use $\boldsymbol{\psi}_k$ also to denote certain random sets/points in the discrete hypercube, but only those forming part of a random maximal chain. We denote by e.g. $\boldsymbol{\xi}^c$ the complement of $\boldsymbol{\xi}$.

\begin{deffo}
	A probability measure $\P$ on $2^{\{1,\dots,N\}}$ (equivalently $\{0,1\}^N$) has product structure if the component variables $\xi_k$ are all independent. We say moreover it has non-trivial product structure if $p_k := \P[\xi_k = 1] \in (0,1)$. We may also call distributions with these properties product distributions and non-trivial product distributions.
\end{deffo}
Note that the $\xi_k$ are automatically Bernoulli variables, taking values in $\{0,1\}$, so that product structure means they are independent Bernoulli variables. We are concerned with showing that given certain $\mathcal{A} \subset 2^{\{1,\cdots,N\}}$ with combinatorial structure, $\P[\boldsymbol{\xi} \in \mathcal{A}]$ is small for $\P$ with non-trivial product structure. In the second half of this section we will develop bounds for events depending on $N$ Bernoulli variables; in the first half we prove that such bounds can tell us about events depending on $N$ general variables which are uniformly anti-concentrated.
\subsection{Bernoulli decompositions}
The fact that we can essentially reduce to the case of independent Bernoulli variables will be shown using the theory of Bernoulli decompositions, representations of random variables in terms of uniform noise together with a Bernoulli component. Specifically:
\begin{deffo}
	We say a pair of measurable real valued functions $(Y(t),Z(t))$ with domain $(0,1)$ are a $p$-Bernoulli decomposition of a real valued random variable $X$ if, for $t$ uniformly distributed on $[0,1]$, $\xi$ a Bernoulli random variable with probability $p$ of taking value $1$, and $\xi$ and $t$ independent we have
	\[ X \overset{\mathcal{D}}{=} Y(t) + Z(t)\xi \]
	where the equality is in the sense of distributions.
\end{deffo}
Aizenman, Germinet, Klein and Warzel showed the existence of Bernoulli distributions optimizing certain quantities related to $Z(t)$, which should be understood as the strength or influence of the Bernoulli noise, in \cite[Theorems 2.1,2.2]{Aizenman2009}. The original paper also detailed various applications, including some to the theory of random Schr\"odinger operators. These methods were also used in \cite{GK2012} to push the methods of \cite{Bourgain2005}; this allowed them to consider general singular noise. The proofs for results on general i.i.d. noise in fact required results for certain non-stationary operators, so called ``generalized Anderson Hamiltonians''.

The absence of stationarity presents some difficulties; the potential at each site admits a decomposition, even a non-trivial one, but these may vary wildly. (By non-trivial, we mean that $\inf_t Z(t)$, the strength of the Bernoulli noise, and $p(1-p)$, the variance of the Bernoulli part, are both positive.) Prima facie, there is also no reason that the variance or strength of the Bernoulli parts can't get arbitrarily small, even if they are guaranteed to be positive everywhere by \cite[Remark 2.1]{Aizenman2009}. Crucially for our whole argument, we show that we can find decompositions so that this is not the case; the variance and strength of the Bernoulli noise are uniformly bounded below.
\begin{thm}\label{decompgoodalt}
	Let $0<\gamma$, $0<\rho<1$ and $M>0$. Then there are $0 < p_-(\gamma,\rho,M) < p_+(\gamma,\rho,M) < 1$ and $\iota(\gamma,\rho,M) > 0$ such that any $X$ which satisfies $0 \leq X \leq M$ almost surely and is anti-concentrated with respect to the parameters $\gamma$ and $\rho$ admits a $p$-Bernoulli decomposition $X \overset{\mathcal{D}}{=} Y(t) + Z(t)\xi$ where
	\begin{enumerate}[label=(\Alph*)]
		\item $p_- \leq p \leq p_+$
		\item $Z(t) \geq \iota > 0$ for all $t \in (0,1)$
	\end{enumerate}
\end{thm}

\begin{proof}
	Our Bernoulli decompositions will be defined in terms of the ``inverse'' cumulative distribution function $G_X(t):= \inf\{u \in \R\,:\, \P[X \leq u]\geq t\}$. Though the following reformulation is obvious, it will clarify certain aspects of the proof:
	\begin{equation}\label{invcdfiff} G_X(t) \leq u \Leftrightarrow \P[X\leq u]\geq t\end{equation}
	We will henceforth drop the explicit dependence on $X$ in our notation, as we will consider one variable at a time. It is verified in e.g. \cite{Aizenman2009} that
	\begin{equation}\label{berform}
		Y(t)=G((1-p)t)\qquad Z(t) = G(1-p+pt)-G((1-p)t)
	\end{equation}
	is a $p$-Bernoulli decomposition of $X$. Aizenman et al. showed that for any non-trivial $X$, there is some $p\in (0,1)$ such that $\inf_t Z(t) > 0$. Our result then fundamentally says this can be uniformized over all $X$ obeying a uniform bound and uniformly anti-concentrated with respect to parameters $\gamma$ and $\rho$.
	
	Na\"ively, we may hope that $\inf_t G(1-p+pt) > \sup_t G((1-p)t)$; this can occur, but requires some kind of ``gap'', e.g. if $X\sim \text{Unif}(0,1)$, then for any $p$ we have
	\[ \lim_{t\rightarrow 1} G((1-p)t) = \lim_{t\rightarrow 0} G(1-p+pt)\]
	$G$ is monotone, so these limits give the supremum and infimum. However, finding a gap where a very small amount of mass is concentrated is still the key to our proof. Specifically:
	\begin{claim}\label{abc}
		If $a<b$ are such that
		$\P[X \in [a,b]] < \frac{1}{2}\min\{\P[X\leq \frac{b+a}{2}],\P[X \geq \frac{b+a}{2}]\}$
		then
		\[ \inf_t [G(1-p+pt)-G((1-p)t)] \geq \frac{b-a}{2}\]
		for $1-p = \P[X\leq \frac{b+a}{2}]$.
	\end{claim}
	The claim follows immediately from the bounds
	\begin{align*} G((1-p)t) &\leq \begin{cases} a &\text{ for }t < \frac{1}{2}\\ \frac{b+a}{2}&\text{ for }t \geq \frac{1}{2}\end{cases}\\
		G(1-p+pt) &\geq \begin{cases} \frac{b+a}{2} &\text{ for } t < \frac{1}{2}\\ 
			b&\text{ for } t \geq \frac{1}{2} \end{cases}\end{align*}
	by applying (\ref{invcdfiff}). Hence it remains to show that appropriate $a$ and $b$ can be found.
	
    We now assume without loss of generality that $\gamma \leq M/4$; anti-concentration is a monotone condition so we can take $\gamma$ smaller without issue. (Our bounds will also be worse with smaller $\gamma$, so uniformity is not violated.) Clearly (as one can see by covering $[0,M]$ with intevals of size $\gamma/2$) we have $\sup_{x \in [0,M]} \P\left[X \in \left[x-\frac{\gamma}{4},x+\frac{\gamma}{4}\right]\right] \geq \frac{\gamma}{2M}$. In particular, there is some $x$ with $\P[ X \in \left[ x-\frac{\gamma}{4}, x + \frac{\gamma}{4}\right] \geq \gamma/2M$. Simultaneously, clearly for any $x$ either $\P[X > x+\frac{\gamma}{2}] \geq \rho/2$ or $\P[X < x-\frac{\gamma}{2}] \geq \rho/2$. We will treat the former case explicitly, and from now on take $x$ fixed so that
	\[ \P\left[X \in \left[x-\frac{\gamma}{4},x+\frac{\gamma}{4}\right]\right] \geq \frac{\gamma}{2M}\]
	and
	\[ \P\left[X > x+\frac{\gamma}{2}\right] \geq \frac{\rho}{2}\]
    The rough idea is to find a subinterval $I \subset \left[ x+\frac{\gamma}{4}, x + \frac{\gamma}{2}\right]$ with $\P[I]$ suitable for using \Cref{abc}; towards this end we partition $\left[ x+\frac{\gamma}{4}, x + \frac{\gamma}{2}\right]$ into $K$ many intervals for \begin{equation*} K := 8\max\left\lceil\left\{\frac{2}{\rho},\frac{2M}{\gamma}\right\}\right\rceil + 1\end{equation*} 
    
    Indeed, we let\begin{equation*}
        I_k := \left[ x+\frac{\gamma}{4} + \frac{k}{K}\frac{\gamma}{4}, x + \frac{\gamma}{4} + \frac{k+1}{K}\frac{\gamma}{4}\right]
    \end{equation*}
    for $k = 0,\cdots, K-1$.
    \begin{claim}\label{kadjclaim}
        There is $k \in \{0,\cdots, K-2\}$ such that 
        \begin{equation}\label{adjk}
            \max\{\P[X \in I_k],\P[X \in I_{k+1}]\} \leq \frac{1}{2} \min\left\{\frac{\rho}{2}, \frac{\gamma}{2M}\right\}\
        \end{equation}
    \end{claim}
    Indeed, we have the following elementary bound
	\begin{equation}\label{elementary}\sum_{k=0}^{K-1} \P[X \in I_k] \leq 2\P[X \in \cup_{k=0}^{K-1} I_k] \leq 2\end{equation}
    In particular, there can be at most $4\max\left\{\frac{2}{\rho},\frac{2M}{\gamma}\right\}$ intervals  $I_k$ such that $\P[X \in I_k] > \frac{1}{2}\min\left\{\frac{\rho}{2},\frac{\gamma}{2M}\right\}$. By our choice of $K$, it is impossible that there is no $k$ satisfying (\ref{adjk}).
	
    Fixing $k$ satisfying the conclusion of \Cref{kadjclaim}, we set $1-p = \P[X \leq x + \frac{\gamma}{4} + \frac{k+1}{K}\frac{\gamma}{4}]$. Note that this is precisely the midpoint of the interval $I:=I_k \cup I_{k+1}$. We obtain
	\[\inf_t[G(1-p+pt)-G((1-p)t)] \geq \frac{|I|}{2} = \frac{\gamma}{4K}\]
	by \Cref{abc}. By construction $\frac{\gamma}{2M}\leq 1-p \leq 1-\frac{\rho}{2}$. This concludes the case where at least half the mass outside $[x-\frac{\gamma}{2},x+\frac{\gamma}{2}]$ is above the interval; accounting for the symmetric case gives the theorem with
	\begin{align*}
		p_- &= \min\left\{\frac{\rho}{2},\frac{\gamma}{2M}\right\}\\
		p_+ &= 1-\min\left\{\frac{\rho}{2},\frac{\gamma}{2M}\right\}\\
		\iota &= \frac{\gamma}{4}\frac{1}{8\max\left\{\frac{2}{\rho},\frac{2M}{\gamma}\right\} + 1}
	\end{align*}

\end{proof}

We mention briefly that this readily yields a version for almost surely bounded variables with positive variance by using \Cref{pzcor}.
\begin{thm}
	For any positive $M$ and $\sigma^2$,  there are $p_-(M,\sigma^2) > 0$, $p_+(M,\sigma^2)<1$ and $\iota(M,\sigma^2)>0$ such that any random variable $X$ with $0 \leq X \leq M$ almost surely and $\var X \geq \sigma^2$ has a $p$-Bernoulli decomposition $X \overset{\mathcal{D}}{=} Y(t) + \xi Z(t)$ with
	\begin{enumerate}
		\item $p_-\leq p \leq p_+$
		\item $Z(t) \geq \iota > 0$ for all $t \in (0,1)$
	\end{enumerate}
\end{thm}

\subsection{The $\kappa$-Sperner property and related probabilistic bounds}

Our work in showing uniform Bernoulli decompositions exist mean that in some sense it suffices to understand the case where the potential variables are independent Bernoulli variables which are not necessarily identically distributed. This will be made precise in the proof of the Wegner estimate; in that same proof it will also be shown that under certain circumstances, configurations of the potential for which a certain resolvent is inconveniently large have a certain combinatorial property introduced in \cite{Ding-Smart}:

\begin{deffo}
	We say $\mathcal{A} \subset 2^{\{1,\dots,N\}}$ is a $\kappa$-Sperner family for $\kappa \in (0,1]$ if for every $\boldsymbol{\eta} \in \mathcal{A}$ there is $B(\boldsymbol{\eta}) \subset \boldsymbol{\eta}^c$ satisfying $|B(\boldsymbol{\eta})| \geq \kappa|\boldsymbol{\eta}^c|$ such that for any $\boldsymbol{\eta}' \in \mathcal{A}$, $\boldsymbol{\eta}' \supset \boldsymbol{\eta}$ implies $\boldsymbol{\eta}' \cap B(\boldsymbol{\eta}) =\varnothing$.
\end{deffo}
The $\kappa =1$ case corresponds to the classical notion of a Sperner family, a collection of subsets which is an antichain with respect to inclusion. The bound proved in \cite{Ding-Smart} is the following:
\begin{thm}[\cite{Ding-Smart}]\label{rhosperner}
	If $\mathcal{A}\subset \{1,\dots,N\}$ is $\kappa$-Sperner for $\kappa \in (0,1]$, then
	\[ |\mathcal{A}| \leq \frac{2^N}{\kappa \sqrt{N}}\]
\end{thm}
If $\boldsymbol{\xi}$ is uniformly distributed, then this result is equivalent to:
\[
	\P[\boldsymbol{\xi} \in \mathcal{A}] \leq \frac{1}{\kappa\sqrt{N}}\]
for any $\mathcal{A} \subset \{1,\dots,N\}$ which is $\kappa$-Sperner.
We will prove an analogous result for general non-trivial product distributions.

\begin{lem}\label{combbound}
	For any non-trivial product distribution with $\beta \leq \min\{p_k,1-p_k\}$ for all $k$, and $\mathcal{A}$ a $\kappa$-Sperner family in $2^{\{1,\dots,N\}}$, we have:
	\begin{equation}\label{probboundouch}\P[\boldsymbol{\xi} \in \mathcal{A}] \leq \frac{C_\beta}{\kappa \sqrt{N}}\end{equation}
	where $C_\beta$ is a constant depending only on $\beta$.
\end{lem}
\Cref{rhosperner} is then the special case $\beta = 1/2$ of this result. Under assumptions weaker than $\beta = 1/2$ but stronger than those of our theorem, the result easily follows from work in \cite{Aizenman2009}. Specifically, the case $\kappa=1$ is proven in \cite[Lemma 3.2]{Aizenman2009}. If $\kappa < 1$ but the distribution of $\boldsymbol{\xi}$ is uniform in the sense that $p_k \equiv p \in (0,1)$, there is technically no proof in the literature, but, making use of (\ref{cardlym}) below, it is trivial to adapt the proof of \cite[Lemma 3.1]{Aizenman2009} to obtain a proof for said case.

We were not able to modify these proofs to work for our purposes, at least not on a technical level. At a high level, the strategy used in both \Cref{rhosperner} and \cite[Lemma 3.1]{Aizenman2009} is the one we end up using; one exploits some variant of the Lubell-Yamamoto-Meshalkin (LYM) inequality. For the former result specifically, one uses the following:
\begin{lem}(\cite{Ding-Smart})\label{dslym}
	If $\mathcal{A} \subset 2^{\{1,\dots,N\}}$ is a $\kappa$-Sperner family,
	\begin{equation}\label{cardlym} \sum_{\boldsymbol{\eta} \in \mathcal{A}} \binom{N}{|\boldsymbol{\eta}|}^{-1}\leq \frac{1}{\kappa}\end{equation}
\end{lem}
The original LYM inequality is the special case $\kappa=1$. For our purposes this version of the LYM inequality does not suffice; the estimates relating $|\mathcal{A}|$ to  $\P[\boldsymbol{\xi}\in \mathcal{A}]$ are very weak, even under reasonable assumptions on the size (e.g. $\P[\boldsymbol{\xi} \in \mathcal{A}] = \Theta(1)$, $|\mathcal{A}| \geq 2^{cN}$) or structure (e.g. $|\boldsymbol{\eta}| =k$ for all $\boldsymbol{\eta} \in \mathcal{A}$) of $\mathcal{A}$. However, the following generalization of \cite[Theorem 3]{yehuda2021slicing} suffices:
\begin{lem}\label{yylym}
	Let $\boldsymbol{\xi}$ be distributed with respect to some non-trivial product distribution, with $0<\beta \leq \min\{p_k,1-p_k\}$ for all $k$. There is a constant $C_\beta$ only depending on $\beta$ such that for any $\kappa$-Sperner family $\mathcal{A} \subset 2^{\{1,\dots,N\}}$:
	\begin{equation}
		\sum_{0\leq k\leq N} \P[\boldsymbol{\xi} \in \mathcal{A}\,|\, |\boldsymbol{\xi}| = k] \leq \frac{(1-\beta)^2}{\beta^2\kappa}
	\end{equation}
\end{lem}
Specifically, Yehuda and Yehudayoff proved, for any Sperner family $\mathcal{A}\subset \{1,\dots,N\}$, that:
\[ \sum_{0\leq k\leq N} \P[\boldsymbol{\xi} \in \mathcal{A}\,|\, |\boldsymbol{\xi}| = k] \leq 1\]
To prove \Cref{yylym}, we need to make use of a technical lemma from the same paper, \cite[Lemma 4]{yehuda2021slicing}.
\begin{lem}[\cite{yehuda2021slicing}]\label{chaindist}
	Let $\xi$ be distributed with respect to some probability measure with non-trivial product structure on $\{1,\dots,N\}$. There is a probability distribution on maximal chains
	\[ \varnothing = \boldsymbol{\psi}_0 \subset \boldsymbol{\psi}_1 \subset \dots \subset \boldsymbol{\psi}_N = \{1,\dots,N\}\]
	such that for all $k$, $\boldsymbol{\psi}_k$ is distributed with respect to the distribution of $\boldsymbol{\xi}$ conditioned on $|\boldsymbol{\xi}| = k$. To be precise, for any deterministic $\boldsymbol{\eta}$ with $|\boldsymbol{\eta}| = k$,
	\[ \P[\boldsymbol{\psi}_k = \boldsymbol{\eta}] = \frac{\P[\boldsymbol{\xi} = \boldsymbol{\eta}]}{\P[|\boldsymbol{\xi}| = k]}\]
\end{lem}
We will not reproduce the proof of this lemma, though we will use certain facts from the proof in proving \Cref{yylym}.
\begin{proof}[Proof of \Cref{yylym}]We let $\boldsymbol{\Psi} = \{\boldsymbol{\psi}_0,\dots,\boldsymbol{\psi}_N\}$ be a maximal chain distributed with respect to the distribution defined in \Cref{chaindist}. For $\mathcal{A}$ which is $\kappa$-Sperner, we let $L$ denote the count of $\boldsymbol{\psi}_k \in \mathcal{A}$. Clearly
	\begin{align*}
		\E L &=  \sum_{0 \leq k \leq N} \P[\boldsymbol{\psi}_k \in \mathcal{A}]\\
		&=\sum_{0\leq k\leq N} \P[\boldsymbol{\xi}\in \mathcal{A}\,|\,|\boldsymbol{\xi}| =k ]
	\end{align*} 
	so that in particular, the result follows as soon as we demonstrate
	\[ \E L \leq \frac{C_\beta}{\kappa}\]
	
	At the same time, $\P[L \geq m+1\,|\,L\geq m] \leq \P[\boldsymbol{\psi}_{k_m+1}\setminus \boldsymbol{\psi}_{k_m} \notin B(\boldsymbol{\psi}_{k_m})]$, where $\boldsymbol{\psi}_{k_m}$ is the $m$-th element of the random chain to intersect $\mathcal{A}$. Here we abuse notation, identifying the singleton set $\boldsymbol{\psi}_{k_m+1}\setminus \boldsymbol{\psi}_{k_m}$ with its lone element. It thus suffices to show (for any $k$) \begin{equation}\P[\boldsymbol{\psi}_{k+1}\setminus \boldsymbol{\psi}_k \notin B(\boldsymbol{\psi}_k)] \leq 1-c_\beta\kappa \end{equation} Indeed, from it one immediately obtains:
	\begin{align*}
		\E L &= \sum_{k=1}^N \P[L\geq k+1\,|\, L\geq k]\\
		&\leq \sum_{k=0}^\infty (1-c_\beta\kappa)^k \\
		&\leq \frac{1}{c_\beta\kappa}
	\end{align*} 
	This is more or less the strategy pursued in the proof of \cite[Theorem 4.2]{Ding-Smart}, with the complication that in this setting not all $j \in \boldsymbol{\psi}_k^c$ are equally likely to be chosen. However, \cite{yehuda2021slicing} provides an explicit formula for the distribution of $\boldsymbol{\psi}_{k+1}\setminus \boldsymbol{\psi}_k$:
	\begin{equation}\label{jchance} \P[\boldsymbol{\psi}_{k+1} = \boldsymbol{\psi}_k \cup \{j\}] = \frac{q_j h_{\boldsymbol{\psi}_k,j}}{g_{k+1}(\boldsymbol{q})}\end{equation}
	where
	\begin{align*}
		q_j &= \frac{p_j}{1-p_j}\\
		\boldsymbol{q} &= (q_1,\dots,q_N)\\
		h_{\boldsymbol{\psi}_k,j} &= \sum_{\substack{|\boldsymbol{\eta}|=k\\\boldsymbol{\eta}\not\ni j}} \frac{1}{|\boldsymbol{\psi}_k \setminus \boldsymbol{\eta}|+1} \prod_{\ell \in \boldsymbol{\eta}} q_\ell &&= \E\left[\frac{1}{|\boldsymbol{\psi}_k\setminus \boldsymbol{\xi}|+1}\,\Big|\, |\boldsymbol{\xi}| = k+1,\; j\not\in \boldsymbol{\xi}\right]\\
		g_{k+1}(\boldsymbol{q}) &= \sum_{|\boldsymbol{\eta}|=k}\prod_{\ell \in \boldsymbol{\eta}} q_\ell &&= \P[|\boldsymbol{\xi}| = k+1]
	\end{align*}
	(Where we give two expressions, the second is given to clarify; we will only use the first definition in our arguments.) We note that of course (\ref{jchance}) is only valid for $j \notin \boldsymbol{\psi}_k$, and that $\sum_{j\notin \boldsymbol{\psi}_k} \P[\boldsymbol{\psi}_{k+1}\setminus \boldsymbol{\psi}_k =j] =1$. In particular, because $|B(\boldsymbol{\psi}_k)| \geq \kappa |\boldsymbol{\psi}_k^C|$, it suffices to show the existence of positive $c_\beta$ such that
	\begin{equation}\label{localjbd}
		q_jh_{\boldsymbol{\psi}_k,j} \geq c_\beta q_{j'}h_{\boldsymbol{\psi}_k,j'}
	\end{equation}for any $j, j' \in \boldsymbol{\psi}_k^C$. 
	Indeed; this implies $\P[\boldsymbol{\psi}_{k+1}\setminus \boldsymbol{\psi}_k =j]$ is, for all $j$, at least $c_\beta/|\boldsymbol{\psi}_k^C|$. By symmetry, there must also be a corresponding upper bound, and in fact we will prove an upper and lower bound. The main work is in the analysis of the ratio $h_{\boldsymbol{\psi}_k,j}/ h_{\boldsymbol{\psi}_k,j'}$, as it is easy to deal with $q_j/q_{j'}$ at the end. We fix $j, j'$, and write
	\[ \tilde{h}_{\boldsymbol{\psi}_k, j} = \sum_{\substack{|\boldsymbol{\eta}|=k\\\boldsymbol{\eta}\not\ni j\\\boldsymbol{\eta}\ni j'}} \frac{1}{|\boldsymbol{\psi}_k \setminus \boldsymbol{\eta}|+1} \prod_{\ell \in \boldsymbol{\eta}} q_\ell\]
	and $\tilde{h}_{\boldsymbol{\psi}_k,j'}$ symmetrically. We also define
	\[h_{\boldsymbol{\psi}_k,j,j'} = \sum_{\substack{|\boldsymbol{\eta}|=k\\ \boldsymbol{\eta}\not\ni j,j'}} \frac{1}{|\boldsymbol{\psi}_k\setminus \boldsymbol{\eta}|+1} \prod_{\ell \in \boldsymbol{\eta}} q_\ell \]
	Because \[ h_{\boldsymbol{\psi}_k,j,j'}= h_{\boldsymbol{\psi}_k,j} - \tilde{h}_{\boldsymbol{\psi}_k,j} = h_{\boldsymbol{\psi}_k,j'} - \tilde{h}_{\boldsymbol{\psi}_k,j'} \geq 0 \]
	we have the following implications:
	\begin{align*}
		\tilde{h}_{\boldsymbol{\psi}_k,j} \leq \tilde{h}_{\boldsymbol{\psi}_k,j'} &\Rightarrow \frac{\tilde{h}_{\boldsymbol{\psi}_k,j}}{\tilde{h}_{\boldsymbol{\psi}_k,j'}} \leq \frac{h_{\boldsymbol{\psi}_k,j}}{h_{\boldsymbol{\psi}_k,j'}} \leq 1\\
		\\
		\tilde{h}_{\boldsymbol{\psi}_k,j} \geq \tilde{h}_{\boldsymbol{\psi}_k,j'} &\Rightarrow \frac{\tilde{h}_{\boldsymbol{\psi}_k,j}}{\tilde{h}_{\boldsymbol{\psi}_k,j'}} \geq \frac{h_{\boldsymbol{\psi}_k,j}}{h_{\boldsymbol{\psi}_k,j'}} \geq 1
	\end{align*}
	(This follows by e.g. differentiating $\frac{a+x}{b+x}$.) It thus remains to estimate $\tilde{h}_{\boldsymbol{\psi}_k,j}/\tilde{h}_{\boldsymbol{\psi}_k,j'}$. But this ratio is exactly $q_j/q_{j'}$. Indeed, let $f$ be the map for the sets $\boldsymbol{\eta}$ with size $k$ and containing $j'$ but not $j$ by $f(\boldsymbol{\eta}) = \boldsymbol{\eta} \cup \{j\} \setminus \{j'\}$. Clearly
	\[ \tilde{h}_{\boldsymbol{\psi}_k,j'} = \sum_{\substack{|\boldsymbol{\eta}|=k\\\boldsymbol{\eta}\not\ni j\\\boldsymbol{\eta}\ni j'}} \frac{1}{|\boldsymbol{\psi}_k\setminus f(\boldsymbol{\eta})|+1}\prod_{\ell \in f(\boldsymbol{\eta})} q_\ell \]
	and moreover
	\[ \frac{1}{|\boldsymbol{\psi}_k\setminus f(\boldsymbol{\eta})| + 1} \prod_{\ell \in f(\boldsymbol{\eta})}q_\ell  = \frac{q_j}{q_{j'}}\frac{1}{|\boldsymbol{\psi}_k\setminus \boldsymbol{\eta}| + 1} \prod_{\ell \in \boldsymbol{\eta}} q_\ell\]
	whence we obtain another pair of bounds:
	\begin{align*}
		1 \leq \frac{h_{\boldsymbol{\psi}_k,j}}{h_{\boldsymbol{\psi}_k,j'}} \leq \frac{q_{j'}}{q_j} &\quad \text{ if } q_j \leq q_{j'}\\
		\frac{q_{j'}}{q_j} \leq \frac{h_{\boldsymbol{\psi}_k,j}}{h_{\boldsymbol{\psi}_k,j'}} \leq 1  &\quad \text{ if } q_j \geq q_{j'}
	\end{align*}
	Finally, these bounds immediately imply:
	\begin{align*}
		\frac{q_j}{q_{j'}} \leq \frac{q_jh_{\boldsymbol{\psi}_k,j}}{q_{j'}h_{\boldsymbol{\psi}_k,j'}} \leq 1 &\quad \text{ if } q_j \leq q_{j'}\\
		1 \leq \frac{q_jh_{\boldsymbol{\psi}_k,j}}{q_{j'}h_{\boldsymbol{\psi}_k,j'}} \leq \frac{q_j}{q_{j'}}  &\quad \text{ if } q_j \geq q_{j'}
	\end{align*}
	Obviously 
	\[\frac{\beta^2}{(1-\beta)^2} \leq \frac{q_j}{q_{j'}} \leq \frac{(1-\beta)^2}{\beta^2}\]
	so that the theorem follows.
\end{proof}
We now prove \Cref{combbound} largely along the lines of the proof of \cite[Theorem 4.2]{Ding-Smart} and \cite[Lemma 3.1]{Aizenman2009}. 
\begin{proof}[Proof of \Cref{combbound}]
	Recall that $\mathcal{A}$ denotes a $\kappa$-Sperner family, and $\boldsymbol{\xi}$ is distributed with respect to some product distribution with $0<\beta \leq \min\{p_k,1-p_k\}$ for all $k$. We let $\mathcal{A}_k$ denote the sub-family of $\mathcal{A}$ consisting solely of size $k$ subsets. We can then reformulate \Cref{yylym} as
	\[ \sum_{k=0}^N \frac{\P[\boldsymbol{\xi} \in \mathcal{A}_k]}{\P[|\boldsymbol{\xi}| =k]} \leq \frac{(1-\beta)^2}{\beta^2\kappa}\] 
	We immediately get
	\begin{align*} \P[\boldsymbol{\xi}\in \mathcal{A}] &= \sum_{k=0}^N \P[\boldsymbol{\xi} \in \mathcal{A}_k]\\
		&\leq \max_{0\leq k\leq N} \P[|\boldsymbol{\xi}|=k] \sum_{k=0}^N \frac{\P[\boldsymbol{\xi} \in \mathcal{A}_k]}{\P[|\boldsymbol{\xi}| =k]}\\
		&\leq \frac{(1-\beta)^2}{\beta^2}\max_{0\leq k\leq N}\P[|\boldsymbol{\xi}|=k]
	\end{align*} 
	It follows from \cite[Theorem 5]{yehuda2021slicing} that \[\max_{0\leq k\leq N}\P[|\boldsymbol{\xi}|=k] \leq \frac{C}{\sqrt{\beta(1-\beta)N}}\] with $C$ a universal constant. Hence we obtain the lemma with $C_\beta = C\beta^{-5/2}(1-\beta)^{3/2}$.
\end{proof}
\begin{remm}
	As we have already discussed to some degree, under stronger assumptions on either the combinatorial structure ($\kappa =1$) or on the distribution of $\boldsymbol{\xi}$ (constant $p_k$), this was either already known or very easily extrapolated from what was known. We mention that in these cases, the dependence of the constant on $\beta$ is also better. In the former case, one can take $C_\beta = C\beta^{-1}$ with $C$ universal, and in the latter $C_\beta =C\beta^{-1/2}(1-\beta)^{-1/2}$, with $C$ universal. We believe our $C_\beta$ to be suboptimal in its dependence on $\beta$; we did not seriously try to optimize the constant in \Cref{yylym}, and it is possible gains could be made there.
\end{remm}

\section{Wegner estimate}\label{wegnersec}
In this section, we obtain our analogue of the Wegner lemma by more or less the method of \cite{Ding-Smart} combined with \Cref{combbound} and \Cref{decompgoodalt}. The main novelty here is the use of the results on Bernoulli decompositions; however these decompositions are used within the a key step of the lemma, rather than before or after the main arguments, so we go through the entirety of the proof. The main novelty is contained in Claims \ref{layerclaim} and \ref{smallcountclaim}. Precisely, we prove:
\begin{lem}\label{bigweg}
	If $\overline{E}$ is fixed and the potential $V$ satisfies the assumptions \ref{aprime}, \ref{bprime} and \ref{cprime} (so that \Cref{ucucthmfinal} is applicable), and moreover the following hold:
	\begin{enumerate}[label=(\Roman*)]
		\item$\eta > \ve > \delta> 0$ are sufficiently small
		\item $K \geq 1$ is an integer
		\item $L_0 \geq \cdots L_5 \geq C_{\eta,\ve,\delta,K}$ are dyadic scales satisfying $L_k^{1-2\delta} \geq L_{k+1} \geq L_k^{1-\ve}$
		\item $\Lambda \in \Z^2$ with $\ell(\Lambda) = L_0$
		\item $\Lambda_1',\dots,\Lambda_j' \subset \Lambda$ with $\ell(\Lambda_j') = L_3$
		\item $G \subset \cup_k \Lambda_j'$ with $0 < |G| \leq L_0^\delta$
		\item $F \subset \Z^2$ is $\eta$-regular in every $\Lambda' \subset \Lambda \setminus \cup_j \Lambda_j'$ with $\ell(\Lambda')=L_3$
		\item \label{glocal} $V:\Lambda \rightarrow [0,M]$, $V_F = v$, $|E - \overline{E}| \leq e^{-L_5}$, and $H_\Lambda\psi = E\psi$ implies
		\[ e^{L_4}\|\psi\|_{\ell^2(\Lambda\setminus \cup_k \Lambda_j')} \leq \|\psi\|_{\ell^2(\Lambda)} \leq  (1+L_0^{-\delta})\|\psi\|_{\ell^2(G)}\]
	\end{enumerate}
	then
	\[ \P[\|(H_\Lambda-\overline{E})^{-1}\| \leq e^{L_1}\,|\, V_F = v] \geq 1-L_0^{C\ve-1/2}\]
	Here $(H_\Lambda-\overline{E})^{-1}$ is the inverse of $H_\Lambda -\overline{E}$ in $\ell^2(\Lambda)$, not in $\ell^2(\Z^2)$, and the smallness required of $\eta$, $\ve$, $\delta$ and the value of $C_{\eta,\ve,\delta,K}$ depend on $M$, $\rho$ and $\gamma$.
\end{lem}
However, before we can prove this result, we need to recall various technical results, all of which are deterministic in nature and hence require no changes for the non-stationary context. The first of these is \cite[Lemma 5.2]{Ding-Smart}:
\begin{lem}[\cite{Ding-Smart}]\label{eigenvar}
	Let $A$ be an $n\times n$ real symmetric matrix, with eigenvalues $E_1 \geq \cdots \geq E_n$ and associated orthonormal eigenbasis $v_1,\dots v_n$. We also let $e_1, \dots, e_n$ denote the standard basis for $\R^n$ and $e_k \otimes e_k$ the rank one matrix acting by $v \mapsto \langle e_k, v\rangle e_k$. If
	\begin{enumerate}[label=(\Alph*)]
		\item $0 < r_1 < r_2 < r_3 < r_4 < r_5 < 1$\label{basicchain}
		\item $r_1 \leq c \min\{r_3r_5,r_2r_3/r_4\}$\label{tech1}
		\item $0 < E_j \leq E_i < r_1 < r_2 < E_{i-1}$\label{gap}
		\item $|\langle v_j, e_k \rangle|^2 \geq r_3$\label{highcon}
		\item $\sum_{r_2 < E_\ell < r_5} |\langle v_\ell, e_k\rangle|^2 \leq r_4$\label{lowcon}
	\end{enumerate}
	then $\tr 1_{[r_1,\infty)}(A) < \tr 1_{[r_1,\infty)}(A + e_k\otimes e_k)$.
\end{lem}
This result says roughly that under certain technical conditions, we can move an eigenvalue $E_j$ across a threshold $r_1$ by adding $e_k \otimes e_k$ to the matrix if the associated eigenvector $v_k$ has enough mass at coordinate $k$ and other eigenvectors ($v_\ell$) corresponding to nearby energies ($r_2 < E_\ell < r_5$) don't have much mass there.

We also need a bound on the possible number of ``almost orthogonal'' eigenvectors in $\R^n$, a weaker form of the Kabatyanskii-Levenshtein bounds \cite{klcode} on spherical codes, for which a proof is given by Tao in \cite{taobound}:
\begin{lem}[\cite{taobound}]\label{taobound}
	If $v_1,\dots, v_m \in \R^n$ satisfy $|\langle v_i,v_j \rangle - \delta_{ij}| < \frac{1}{\sqrt{5n}}$, then $m \leq (5-\sqrt{5})n/2$.
\end{lem}
The last technical result we need is the following, which guarantees the existence of small scale boxes away from ``bad'' sets where eigenfunctions have non-negligible magnitude compared to their magnitude on a large scale box. This is \cite[Lemma 5.3]{Ding-Smart}
\begin{lem}[\cite{Ding-Smart}]\label{needsq}
	For every integer $K \geq 1$, if
	\begin{enumerate}[label=(\Alph*)]
		\item $L \geq C_K L' \geq L' \geq C_K$
		\item $\Lambda \subset \Z^2$ with $\ell(\Lambda) = L$
		\item $\Lambda'_k \subset \Z^2$ for $k=1,\dots,K$ with $\ell(\Lambda'_k) = L'$
		\item $H_\Lambda \psi = E \psi$
	\end{enumerate}
	then there is some $\Lambda'$ satisfying
	\begin{enumerate}[label=(\roman*)]
		\item $\ell(\Lambda') = L'$
		\item $2\Lambda' \subset \Lambda \setminus \cup_{k=1}^K \Lambda'_k$
		\item $\|\psi\|_{\ell^\infty(\Lambda')} \geq e^{-C_KL'}\|\psi\|_{\ell^\infty(\Lambda)}$
	\end{enumerate}
\end{lem}
With this, we are ready to prove \Cref{bigweg}.
\begin{proof}[Proof of \Cref{bigweg}]
	Throughout, $C > 1 > c > 0$ will be constants allowed to depend on $\ve, \delta, \eta$ and $K$, but crucially not on the scale. (As with most of our results, there is also an implicit dependence on parameters of the potential $M$, $\gamma$, and $\rho$.) We let $E_1(H_\Lambda) \geq \cdots \geq E_{L_0^2}(H_\Lambda)$ be eigenvalues, and further let $\psi_k(H_\Lambda)$ be associated normalized eigenfunctions, treating these as deterministic functions of the potential living in $[0,M]^\Lambda$.
	\begin{claim}
		Without loss of generality we can assume $\cup_k \Lambda'_k \subset F$.
	\end{claim}
	Let $F' = \cup_k \Lambda'_k \setminus F$. Clearly $F \cup F'$ still satisfies the sparsity assumptions in the statement of \Cref{bigweg}. Further, we have for any event $\mathcal{E}$
	\[\P[\mathcal{E}|V_F = v] = \E[\P[\mathcal{E}|V_{F\cup F'} = v\cup v']|V_F=v]\]
	so that in particular, were we able to show the requisite bounds conditioned on $V_{F\cup F'} = v\cup v'$, the requisite bounds conditioned only on $V_F = v$ follow by averaging over all the realizations $V_{F'} = v'$.
	\begin{claim}We let $\mathcal{E}_{uc}$ denote the event that for energies $E$ with $|E - \overline{E}|< e^{-L_5}$, and associated eigenfunctions $H_\Lambda \psi = E \psi$, we have
		\[ |\{|\psi|\geq e^{-L_2}\|\psi\|_{\ell^\infty(\Lambda)}\}\setminus F|\geq L_4^{3/2}\]
	Then 
	$\P[\mathcal{E}_{uc}|V_F=v] \geq 1-e^{-L_0^\ve}$.
	\end{claim}
	We let $\alpha' > 1 > \ve'$ be constants such that \Cref{ucucthmfinal} holds, and without loss of generality assume $\ve' > \eta$. Then, letting $\mathcal{E}_{uc}(\Lambda',F)$ denote the unique continuation event from \Cref{ucucthmfinal} and $\mathcal{E}_{uc}'$ denote the event
	\[ \mathcal{E}_{uc}' := \bigcap_{\substack{\Lambda' \subset \Lambda\setminus\cup_k \Lambda_j'\\\ell(\Lambda')=L_3}}\mathcal{E}_{uc}(\Lambda',F)\]
	we obtain\[ \P[\mathcal{E}_{uc}'|V_F=v] \geq 1-e^{-\ve'L^{1/4}_3+C\log L_0} \geq 1-e^{-L_0^\ve} \]
	with the last inequality holding for sufficiently large scales. So in particular it suffices to show $\mathcal{E}_{uc} \supset \mathcal{E}_{uc}'$ (again at sufficiently large scales). This follows straightforwardly; using \Cref{needsq}, there is some $\Lambda'$ such that $\Lambda' \subset \Lambda\setminus\cup_k \Lambda'_k$ and
	\[ \|\psi\|_{\ell^\infty(\frac{1}{2}\Lambda')} \geq e^{-CL_3}\|\psi\|_{\ell^\infty(\Lambda)}\]
	so that in particular, for relevant $E$ and $\psi$, $\mathcal{E}_{uc}'$ implies
	\[ |\{|\psi|\geq e^{-\alpha'L_3\log L_3}\|\psi\|_{\ell^\infty(\frac{1}{2}\Lambda')}\}\cap \Lambda'\setminus F| \geq \frac{L_3^{3/2}}{(\log L_3)^{1/2}}\]
	Small powers of the scales dominate constants and logarithmic terms, so we obtain:
	\[|\{|\psi|\geq e^{-L_2}\|\psi\|_{\ell^\infty(\Lambda)}\}\cap \Lambda \setminus F| \geq L_4^{3/2}\]
	which shows the claim; $\mathcal{E}_{uc} \supset \mathcal{E}_{uc}'$ for sufficiently large scales.
	We've shown in particular that $\mathcal{E}_{uc}$ failing is very unlikely; it remains to show that $\|(H-\overline{E})^{-1}\| > e^{L_1}$ is very unlikely if $\mathcal{E}_{uc}$ holds. This follows via eigenvariation on the large support.
	\begin{claim}\label{layerclaim}
		Letting $s_\ell = e^{-L_1+\ell(L_2-L_4+C)}$ and further letting $\mathcal{E}_{k_1,k_2,\ell}$ denote the event that
		\[ |E_{k_1}-\overline{E}|,|E_{k_2}-\overline{E}| < s_\ell\quad\text{and}\quad |E_{k_1-1}-\overline{E}|,|E_{k_2+1}-\overline{E}| \geq s_{\ell+1}\]
		we have the bounds
		\[\P[\mathcal{E}_{k_1,k_2,\ell}\cap \mathcal{E}_{uc}|V_F =v] \leq CL_0L_4^{-3/2}\]
		for $L_0$ sufficiently large.
	\end{claim}	
	To prove this, we use the Bernoulli decompositions guaranteed to exist by \Cref{decompgoodalt}. Specifically, there are
	\begin{enumerate}[label=(\Alph*)]
		\item $0<p_-\leq p+<1$
		\item $0<\iota$
		\item Measurable functions $(Y_n)_{n\in\Z^2}$ and $(Z_n)_{n\in\Z^2}$ from $(0,1)$ to $[0,M]$, with $Z_n(t) \geq \iota$ for all $n \in \Z^2$ and $t \in (0,1)$
	\end{enumerate}
	such that
	\[V_n \overset{D}{=} Y_n(t_n) + Z_n(t_n)\xi_n\]
	where the $t_n$ are uniformly distributed on $(0,1)$, the $\xi_n$ have law $(1-p_n)\rho + p_n\delta_1$ for $p_-\leq p_n \leq p_+$ and all the $t_n$ and $\xi_n$ are independent.
	
	Then our system is distributed the same as one where we replace $V_n$ with $Y_n(t_n) + Z_n(t_n)\xi_n$. In particular, we can consider $\mathcal{E}_{k_1,k_2,\ell}\cap \mathcal{E}_{uc}$ (conditioned on $V_F=v$) as an event depending entirely on $(t_n)_{n\in\Lambda\setminus F}$ and $(\xi_n)_{n\in\Lambda\setminus F}$.
	
	We use the same boldface notation to denote vectors as in \Cref{antisec}, so that $\boldsymbol{\xi} = (\xi_n)_{n\in\Lambda\setminus F}$ and $\mathbf{t} = (t_n)_{n\in\Lambda\setminus F}$. We condition $\mathcal{E}_{k_1,k_2,\ell}\cap \mathcal{E}_{uc}$ further, also conditioning on $\mathbf{t} = \mathbf{t}'$. In particular, it clearly suffices to show
	\[\P[\mathcal{E}_{k_1,k_2,\ell}\cap \mathcal{E}_{uc}|V_F =v\text{ and } \mathbf{t} = \mathbf{t}'] \leq CL_0L_4^{-3/2}\]
	for any choice of $\boldsymbol{t}'$.
	
	Clearly $\mathcal{E}_{k_1,k_2,\ell}\cap \mathcal{E}_{uc}$ conditioned on $V_F = v$ and $\mathbf{t}=\mathbf{t}'$ is entirely determined by $\boldsymbol{\xi}$; we identify realizations of $\boldsymbol{\xi}$ with subsets of $\Lambda\setminus F$ in the standard way where $\boldsymbol{\xi}$ as a set is the lattice points where $\xi_n = 1$, so that $\mathcal{E}_{k_1,k_2,\ell}\cap \mathcal{E}_{uc}$ can be identified with an element of $2^{\Lambda \setminus F}$.
	
	Using this identification, we can define events determined by the potential on $\Lambda \setminus F$ which cover $\mathcal{E}_{k_1,k_2,\ell}\cap \mathcal{E}_{uc}$: $\mathcal{E}_{k_1,k_2,\ell,+}$ and $\mathcal{E}_{k_1,k_2,\ell,-}$. We say $\boldsymbol{\xi} \in \mathcal{E}_{k_1,k_2,\ell,+}$ if and only if $\boldsymbol{\xi} \in \mathcal{E}_{k_1,k_2,\ell}$ and $|\boldsymbol{\xi} \cap \{|\psi_{k_1}|\geq e^{-L_2}\}| \geq \frac{1}{2}L_4^{3/2}$. Similarly, we say $\boldsymbol{\xi} \in \mathcal{E}_{k_1,k_2,\ell,-}$ if $\boldsymbol{\xi} \in \mathcal{E}_{k_1,k_2,\ell}$ and $|\boldsymbol{\xi}^C \cap \{|\psi_{k_1}|\geq e^{-L_2}\}| \geq \frac{1}{2}L_4^{3/2}$. Clearly $\mathcal{E}_{k_1,k_2,\ell} \cap \mathcal{E}_{uc} \subset \mathcal{E}_{k_1,k_2,\ell,+}\cup \mathcal{E}_{k_1,k_2,\ell,-}$ and so it suffices to estimate both of these. (Note that these events on the right hand side only make sense after fixing $\mathbf{t} = \mathbf{t}'$.)
	
	As a consequence of \Cref{combbound}, the claim follows immediately if we can show that $\mathcal{E}_{k_1,k_2,\ell,\pm}$ are $\kappa$-Sperner families in $\Lambda \setminus F$ for $\kappa = \frac{1}{2}L_0^{-2}L_4^{3/2}$. To show this, we use \Cref{eigenvar} and proceed more or less identically as in \cite[Claim 5.10]{Ding-Smart}. We consider the case where $\boldsymbol{\xi} \in \mathcal{E}_{k_1,k_2,\ell,-}$.  We let $\tilde{H} := \frac{1}{\iota}(H-\overline{E}-s_\ell)$ with eigenvalues $\tilde{E_k} = \frac{1}{\iota}(E_k-\overline{E}-s_\ell)$ and the same orthonormal eigenvectors $\psi_k$. We define the radii $r_1 = \frac{2}{\iota}s_\ell$, $r_2 = \frac{1}{\iota}s_{\ell+1}$, $r_3 = \frac{1}{\iota}e^{-L_2}$, $r_4 = \frac{1}{\iota}e^{cL_4}$, and $r_5 = \frac{1}{\iota}e^{-L_5}$, for use with \Cref{eigenvar}. 
	Obviously these radii satisfy conditions \ref{basicchain}, \ref{tech1}, and \ref{gap} from \ref{eigenvar} with respect to $E_i$ and $E_{i+1}$ for sufficiently large $L_0$. For $n \in \boldsymbol{\xi}^C \cap \{|\psi_{k_1}|\geq e^{-L_2}\}$, we also have condition \ref{highcon}, where in this case we consider the mass of $\psi_{k_1}$ at site $n$. Finally, under our assumption \ref{glocal} for \Cref{bigweg}, we obtain condition \ref{lowcon} for application of \Cref{eigenvar}. 
	
	In particular, increasing the potential by at least 1 for the operator at any of the points $\boldsymbol{\xi} \cap \{|\psi_{k_1}|\geq e^{-L_2}\}$ pushes the eigenvalue past $r_2$. But this means ``flipping'' $\boldsymbol{\xi}$ at any of these points forces us to exit $\mathcal{E}_{k_1,k_2,\ell,-}$; by the minimum bound $Z_n(t_n) \geq \iota$, a flip represents adding at least one to the potential at that site for the operator $\tilde{H}$. It follows that $\mathcal{E}_{k_1,k_2,\ell,-}$ is $\frac{1}{2}L_0^{-2}L_4^{3/2}$-Sperner by the lower bound on the size of $\boldsymbol{\xi}^C \cap \{|\psi_{k_1}| \geq e^{-L_2}\}$ in the definition of $\mathcal{E}_{k_1,k_2,\ell,-}$. Showing the same for $\mathcal{E}_{k_1,k_2,\ell,+}$ follows by a symmetric argument.
	\begin{claim}\label{smallcountclaim}
		There is a set $K = K(F,v,\boldsymbol{t})\subset {1,\dots,L_0^2}$ (recall that $v$ is the potential on $F$) such that $|K| \leq CL_0^\delta$ and when conditioning on $V_F = v$, we have the inclusion
		\[ \{\|(H_\Lambda - \overline{E})^{-1}\| \geq e^{L_1}\} \subset \bigcup_{\substack{k_1,k_2\in K\\0 \leq \ell \leq CL_0^\delta}} \mathcal{E}_{k_1,k_2,\ell}\]
	\end{claim}
	Because we condition on $V_F=v$, the eigenvalues $E_k$ and eigenfunctions $\psi_k$ are functions of the configuration of the potential on $\Lambda\setminus F$, which is identified with $(\boldsymbol{t},\boldsymbol{\xi})$. Fixing $\boldsymbol{t}=\boldsymbol{t}'$, we treat $E_k$ and $\psi_k$ as functions of $\boldsymbol{\xi}$. The result follows as soon as we can demonstrate that there are at most $CL_0^\delta$ many indices $k$ such that $|E_k-\overline{E}| \leq e^{-L_2}$ for some $\boldsymbol{\xi} \in \{0,1\}^{\Lambda\setminus F}$. Indeed, we enumerate such eigenvalues $E_1,\dots,E_m$ and consider the annuli
	\[I_\ell:=\begin{cases}
		[\overline{E}-s_{\ell+1},\overline{E}+s_{\ell+1}]\setminus [\overline{E}-s_\ell,\overline{E}+s_\ell]&\text{ for }0\leq \ell< m\\
		[\overline{E}-e^{-L_2},\overline{E}+e^{-L_2}]\setminus [\overline{E}-s_m,\overline{E}+s_m]&\text{ for }\ell=m
	\end{cases}\]
	Note that if $\delta$ is sufficiently small $-L_2 > -L_1 + CL_0^\delta(L_2-L_4+C)$ for large scales, and moreover that at least one of these annuli contains no eigenvalues. Hence, if there are any eigenvalues satisfying $|E_k-\overline{E}| < e^{-L_2}$ then there are some $k'$ and $\ell$ such that $\mathcal{E}_{k,k',\ell}$ (or $\mathcal{E}_{k',k,\ell})$ holds.
	
	Finally, by the quasi-localization condition \ref{glocal}, we will show that if there is some $\boldsymbol{\xi}'\in\{0,1\}^{\Lambda\setminus F}$ such that $|E_k(\boldsymbol{\xi'}) - \overline{E}| < e^{-L_2}$, then it turns out that for any $\boldsymbol{\xi}\in \{0,1\}^{\Lambda\setminus F}$ we have $|E_k(\boldsymbol{\xi}) - \overline{E}| < e^{-L_4}$. (Recall that we've presumed that the bad squares $\Lambda_j'$ are contained in $F$ so that this condition readily implies $\|\psi_k(\boldsymbol{\xi}')\|_{\ell^\infty(\Lambda\setminus F)} \leq e^{-L_4}$ so long as $|E_k(\boldsymbol{\xi}')-\overline{E}| \leq e^{-L_5}$.)
	
	We fix $\boldsymbol{\xi}$ and set $\boldsymbol{\xi}_s = \boldsymbol{\xi}'+s(\boldsymbol{\xi}-\boldsymbol{\xi}')$ for $s \in [0,1]$ to linearly interpolate. Note that $\boldsymbol{\xi}_t \in [0,1]^{\Lambda\setminus F}$ rather than $\{0,1\}^{\Lambda\setminus F}$ like previously considered; all the relevant quantities are still defined. Standard eigenvalue variation yields:
	\begin{align*}
		|E_k(\boldsymbol{\xi}_{s'}) - \overline{E}| &\leq |E_k(\boldsymbol{\xi}')-\overline{E}| + M\int_0^{s'} \|\psi_k(\boldsymbol{\xi}_{s})\|_{\ell^2(\Lambda\setminus F)}^2\,ds\\
		&\leq e^{-L_2} + M|\Lambda|\int_0^{s'} e^{-2L_4} + 1_{|E_k(\boldsymbol{\xi}_s)-\overline{E}|\geq e^{-L_5}}\,ds\\
		&\leq \begin{cases} e^{-L_4} &\text{ if } L_0 \text{ large and } s \leq s'\text{ implies }|E_k(\boldsymbol{\xi}_s)-\overline{E}| \leq e^{-L_5}\\
			1 &\text{ otherwise}
		\end{cases}
	\end{align*}
	This implies the necessary bound $|E_k(\boldsymbol{\xi})-\overline{E}| \leq e^{-L_4}$; indeed $E_k(\xi_s) \geq e^{-L_5}$ is impossible for $s' < \frac{e^{2L_4-L_5}}{ML_0^2}$, so that in particular estimates our estimates suffice for all $s' \in [0,1]$ with $L_0$ sufficiently large, requisite largeness depending on $M$. All the inequalities except the first one are standard. For the first inequality, recall that $M$ is the bound on the potential $V$; it is clearly an upper bound for $Z_n(t_n')$. If we let $H_s$ denote $H_\Lambda$ as we interpolate from the potential corresponding to $\boldsymbol{\xi}'$ at $s=0$ to the potential given by $\boldsymbol{\xi}$ at $s=1$, then it follows that
	\begin{equation*}\left|\frac{d}{ds} H_s \right| \leq  MP_{\Lambda\setminus F} 
	\end{equation*}
	Here the inequality is meant in the usual sense of self-adjoint operators, and $P_{\Lambda\setminus F}$ is the orthogonal projection onto $\ell^2(\Lambda\setminus F)$.
	
	It follows now from the quasi-localization condition \ref{glocal} that (given any $\boldsymbol{\xi}'$ such that $|E_k(\boldsymbol{\xi}')-\overline{E}| \leq e^{-L_2}$) we have for all $\boldsymbol{\xi}$:
	\[1-CL_0^\delta \leq \|\psi_k(\boldsymbol{\xi})\|_{\ell^2(G)} \leq 1\]
	and more specifically (as a consequence)
	\[ \|(1-P_G)\psi_k(\boldsymbol{\xi})\|_{\ell^2(\Lambda)} \leq CL_0^{-\delta}\]
	where $P_G$ is the projection onto the set $G$. This, together with the orthogonality $|\langle \psi_k(\boldsymbol{\xi}), \psi_j(\boldsymbol{\xi})\rangle_{\ell^2(\Lambda)}- \delta_{kj}|=0$ gives us almost orthogonality for the $\ell^2(G)$ inner product:
	\[|\langle \psi_k(\boldsymbol{\xi}),\psi_j(\boldsymbol{\xi})\rangle_{\ell^2(G)}-\delta_{kj}| \leq CL_0^{-\delta} \leq (5|G|)^{-1/2}\]
	when $L_0$ is sufficiently large. The bound in \Cref{taobound} gives us that the count of such vectors is at most $C|G| \leq CL_0^{-\delta}$. This proves the claim, and now \Cref{bigweg} follows from basic estimates:
	\begin{align*}
		&\P[\|(H-\overline{E})^{-1}\| \geq e^{L_1}\,|\,V_F=v] \\ &\leq \P[\mathcal{E}_{uc}^C\,|\,V_F=v] + \E_{\boldsymbol{t}'}\left[\sum_{\substack{k_1,k_2 \in K(F,v,\boldsymbol{t}')\\0\leq \ell\leq CL_0^{\delta}}} \P[\mathcal{E}_{k_1,k_2,\ell}\cap \mathcal{E}_{uc}\,|\,V_F=v]\right]\\
		&\leq e^{-L_0^{\ve}} + CL_0^{1+3\delta}L_4^{-3/2}
	\end{align*}
	where $\E_{\boldsymbol{t}'}$ denotes averaging over all $\boldsymbol{t}'$. Finally, because $L_4 \geq L_0^{1-4\ve + O(\ve^2)}$ we get (for sufficiently small $\ve$):
	\begin{align*} \P[\|(H-\overline{E})^{-1}\| \geq e^{L_1}\,|\,V_F=v] &\leq e^{-L_0^\ve} + CL_0^{-\frac{1}{2}+9\ve + O(\ve^2)}\\
		&\leq L_0^{-\frac{1}{2}+10\ve}\end{align*}
	for $L_0$ sufficiently large.
\end{proof}
Having proven the Wegner estimate, the localization proof proceeds almost exactly as in \cite{Ding-Smart}; for the convenience of the reader we prove \Cref{endofmsa2} in appendices \ref{detmsasec} and \ref{msasec}, but we stress that this is not original; we are undertaking our own exposition of \cite[Sections 6,7 \& 8]{Ding-Smart}.

\appendix\section{Unique continuation from the key lemma}\label{key2uc}
Having proven \Cref{key}, \Cref{ucucthmfinal} follows more or less using the exact strategy used by Ding and Smart in the original paper \cite{Ding-Smart} to prove their Theorem 3.5 from their Lemma 3.12. The only change of some note, is that we correct what seem to be minor errors in \cite[Claims 3.23, 3.24]{Ding-Smart}, where they work with e.g. $Q$ instead of $Q\setminus F$. The changes are in Claims \ref{smallbc} and \ref{bigbd}.

The first necessary lemma is a growth lemma, whose content is more or less that an eigenfunction small on all of the ``deep interior'' of a tilted square and most of the whole of the square is not very big anywhere on the square. Specifically:
\begin{lem}\label{growth}
	Assume all the hypotheses of \Cref{ucucthmfinal}, most saliently that the random potential for $H$ satisfies \ref{aprime}, \ref{bprime} and \ref{cprime}. For a tilted square $Q$, any set $F \subset Q$, and real numbers $0 < \ve < 1 < \alpha$ we let $\mathcal{E}_{ex}(Q,F,\alpha,\ve)$ denote the event that
	\begin{equation}\label{exhyp}\begin{cases}|E - \overline{E}| \leq e^{-\alpha(\ell(Q)\log \ell(Q))^{1/2}}\\
			H\psi = E \psi\text{ in }2Q\\
			|\psi|\leq 1\text{ in }\frac{1}{2}Q\\
			|\psi|\leq 1\text{ in a }1-\ve(\ell(Q)\log\ell(Q))^{-1/2}\text{ fraction of } 2Q\setminus F\end{cases}
	\end{equation}
	implies that  $|\psi|\leq e^{\alpha\ell(Q)\log\ell(Q)}\text{ in }Q$.
	Then for any sufficiently small $\ve > 0$, there is some $\alpha > 1$ depending on $\ve$ such that if
	\begin{enumerate}[label=(\Alph*)]
		\item $\ell(Q) \geq \alpha$
		\item $F$ is $\ve$-sparse in $2Q$
	\end{enumerate}
	we have the bound $\P[\mathcal{E}_{ex}|V_F = v] \geq 1 - e^{-\ve \ell(Q)}$ for any $v:F\rightarrow [0,M]$. The smallness required of $\ve$ and the value of $\alpha$ depend on $M$, $\rho$ and $\gamma$.
\end{lem}
\begin{proof}
    We let $\mathcal{E}_{ex}'(\alpha,F,a)$ denote the event that
	\begin{equation}\label{exprimehyp}
		\begin{cases}|E - \overline{E}| \leq e^{-\alpha(a\log a)^{1/2}/2}\\
			H\psi = E \psi\text{ in }4R_{[1,a],[1,a]}\\
			|\psi|\leq\text{ in }R_{[1,a],[1,a]}\\
			|\psi|\leq 1\text{ in a }1-\ve(a \log a))^{-1/2}\text{ fraction of } 4R_{[1,a],[1,a]}\setminus F\end{cases}
	\end{equation}

	implies 
	\begin{equation}\label{exprimebound}
		|\psi|\leq e^{\alpha a \log a/2}\text{ in }R_{[1,a],[1,2a]}
	\end{equation}\\
	
	By symmetry, it suffices to show that for sufficiently large $\alpha$ and $a > \alpha/2$, $\P[\mathcal{E}_{ex}'|V_F=v] \geq 1-e^{-\ve a}$. Indeed, carrying out this lemma in four directions (using 90$^\circ$ rotational symmetry) gives us a diagonal cross of sorts; carrying it out four more times fills out said diagonal cross to get $R_{[-a+1,2a],[-a+1,2a]} \supset 2R_{[1,a],[1,a]}$, and so we obtain that with high probability ($1-8e^{-\ve a}$) our eigenfunctions are not larger than $e^{\alpha \ell(Q)\log \ell(Q)}$ anywhere on $2R_{[1,a],[1,a]}$. For a schematic depiction, see Figure \ref{fig:growth8}.

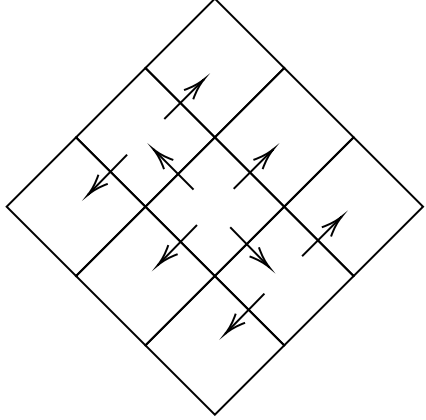
\begin{figure}[h]\caption{The growth lemma follows from growing in one direction eight times}\label{fig:growth8}

\tikzset{every picture/.style={line width=0.75pt}} 

\begin{tikzpicture}[x=0.75pt,y=0.75pt,yscale=-1,xscale=1]
	
	\draw   (135,126) -- (170,161) -- (135,196) -- (100,161) -- cycle ;
	\draw   (170,161) -- (205,196) -- (170,231) -- (135,196) -- cycle ;
	\draw   (100,91) -- (135,126) -- (100,161) -- (65,126) -- cycle ;
	\draw   (170,91) -- (205,126) -- (170,161) -- (135,126) -- cycle ;
	\draw   (205,126) -- (240,161) -- (205,196) -- (170,161) -- cycle ;
	\draw   (135,56) -- (170,91) -- (135,126) -- (100,91) -- cycle ;
	\draw   (205,56) -- (240,91) -- (205,126) -- (170,91) -- cycle ;
	\draw   (240,91) -- (275,126) -- (240,161) -- (205,126) -- cycle ;
	\draw   (170,21) -- (205,56) -- (170,91) -- (135,56) -- cycle ;
	\draw    (179.5,117) -- (198.09,98.41) ;
	\draw [shift={(199.5,97)}, rotate = 135] [color={rgb, 255:red, 0; green, 0; blue, 0 }  ][line width=0.75]    (10.93,-3.29) .. controls (6.95,-1.4) and (3.31,-0.3) .. (0,0) .. controls (3.31,0.3) and (6.95,1.4) .. (10.93,3.29)   ;
	\draw    (214,151) -- (232.59,132.41) ;
	\draw [shift={(234,131)}, rotate = 135] [color={rgb, 255:red, 0; green, 0; blue, 0 }  ][line width=0.75]    (10.93,-3.29) .. controls (6.95,-1.4) and (3.31,-0.3) .. (0,0) .. controls (3.31,0.3) and (6.95,1.4) .. (10.93,3.29)   ;
	\draw    (144.5,81.75) -- (163.09,63.16) ;
	\draw [shift={(164.5,61.75)}, rotate = 135] [color={rgb, 255:red, 0; green, 0; blue, 0 }  ][line width=0.75]    (10.93,-3.29) .. controls (6.95,-1.4) and (3.31,-0.3) .. (0,0) .. controls (3.31,0.3) and (6.95,1.4) .. (10.93,3.29)   ;
	\draw    (177.75,136.38) -- (196.1,155.19) ;
	\draw [shift={(197.5,156.63)}, rotate = 225.72] [color={rgb, 255:red, 0; green, 0; blue, 0 }  ][line width=0.75]    (10.93,-3.29) .. controls (6.95,-1.4) and (3.31,-0.3) .. (0,0) .. controls (3.31,0.3) and (6.95,1.4) .. (10.93,3.29)   ;
	\draw    (195,169.88) -- (176.65,188.69) ;
	\draw [shift={(175.25,190.13)}, rotate = 314.28] [color={rgb, 255:red, 0; green, 0; blue, 0 }  ][line width=0.75]    (10.93,-3.29) .. controls (6.95,-1.4) and (3.31,-0.3) .. (0,0) .. controls (3.31,0.3) and (6.95,1.4) .. (10.93,3.29)   ;
	\draw    (161,135.25) -- (142.65,154.07) ;
	\draw [shift={(141.25,155.5)}, rotate = 314.28] [color={rgb, 255:red, 0; green, 0; blue, 0 }  ][line width=0.75]    (10.93,-3.29) .. controls (6.95,-1.4) and (3.31,-0.3) .. (0,0) .. controls (3.31,0.3) and (6.95,1.4) .. (10.93,3.29)   ;
	\draw    (125.75,99.75) -- (107.4,118.57) ;
	\draw [shift={(106,120)}, rotate = 314.28] [color={rgb, 255:red, 0; green, 0; blue, 0 }  ][line width=0.75]    (10.93,-3.29) .. controls (6.95,-1.4) and (3.31,-0.3) .. (0,0) .. controls (3.31,0.3) and (6.95,1.4) .. (10.93,3.29)   ;
	\draw    (159.25,117.38) -- (140.41,98.42) ;
	\draw [shift={(139,97)}, rotate = 45.18] [color={rgb, 255:red, 0; green, 0; blue, 0 }  ][line width=0.75]    (10.93,-3.29) .. controls (6.95,-1.4) and (3.31,-0.3) .. (0,0) .. controls (3.31,0.3) and (6.95,1.4) .. (10.93,3.29)   ;
	
	\draw (162,117.4) node [anchor=north west][inner sep=0.75pt]  [font=\normalsize]  {};

\end{tikzpicture}

\end{figure}
	Now we show that for sufficiently large $\alpha$ and $a > \alpha/2$, $\P[\mathcal{E}_{ex}'(\alpha,F,a)|V_F=v] \geq 1-e^{-\ve a}$. We let $\alpha' > 1 > \ve' > 0$ denote a pair of constants such that the conclusions of \Cref{key} apply. We let $\mathcal{E}_{ni}(R_{[1,a],[c,d]})$ denote precisely the same event as in the statement of said lemma, i.e. that, roughly, if an eigenfunction is small on all of one long (length $a$) edge, and most of the other, then it is not too big anywhere. If we let 
	\[\mathcal{E}_{ni}(\alpha) = \bigcap_{\substack{[c,d]\subset [1,\frac{5}{2}a]\\ \alpha(d-c+1)^2 \log a \leq 2a}}\mathcal{E}_{ni}(R_{[1,a],[c,d]})\]
	then, so long as $\alpha\geq 2\alpha'$, we have by \Cref{key} and a union bound that
	\[ \P[\mathcal{E}_{ni}|V_F = v] \geq 1 - e^{-\ve'a + C \log a}\]
	with the logarithmic term coming from the polynomial upper bound on the count of intervals $[c,d]$ satisfying the conditions. Obviously for sufficiently small $\ve$, we have $1-e^{\ve a} \geq 1-e^{-\ve' a + C \log a}$. It will suffice to show that for $\ve < c\ve'(\alpha')^{1/2}$ we can find $\alpha$ so that $\mathcal{E}_{ex}'(\alpha) \supset \mathcal{E}_{ni}(\alpha)$.
	
	We now proceed under the assumption that we are working with a realization of the potential in $\mathcal{E}_{ni}$, and an eigenfunction satisfying (\ref{exprimehyp}), seeking to show that (\ref{exprimebound}) holds as well; recall that $\mathcal{E}_{ex}'$ is nothing more than the event where the implication (\ref{exprimehyp}) $\Rightarrow$ (\ref{exprimebound}) holds. By using \Cref{key}, it will suffice to show the following:
	\begin{claim}
		There is a sequence $b_0 <\cdots < b_K$, with $b_k \in [a,\frac{5}{2}a]$, such that
		\begin{enumerate}[label=(\Alph*)]
			\item $b_0 = a$
			\item $b_K \geq 2a$
			\item $\frac{1}{2}a \leq \alpha'(b_{k+1}-b_k+2)^2\log a \leq a\text{ for } 0 \leq k < K$
			\item $|\psi| \leq 1 \text{ on a } 1-\ve'\text{ fraction of }R_{[1,a],[b_{k+1}-1,b_{k+1}]}\setminus F\text{ for } 0 \leq k < K$
		\end{enumerate}
	\end{claim}
	We argue inductively. We set $B= \left\lfloor \left(\frac{a}{\alpha'\log a}\right)^{1/2} \right\rfloor$. Assuming $b_k$ is already defined, we decompose $R_{[1,a],[b_k+B/2,b_k+B]}$ into diagonals:
	\[R_{[1,a],[b_k+B/2,b_k+B]} = \bigsqcup_{b \in [b_k+B/2,b_k+B]} R_{[1,a],[b,b]}\]
	By assumption, specifically the last condition of (\ref{exprimehyp}), we know $|\psi| > 1$ on at most $\ve a^{3/2}\log a$ points in $4R_{[1,a],[1,a]}\setminus F$. Simultaneously, there are at most $\ve c B a \leq C \ve a^{3/2}(\alpha' \log a)^{-1/2}$ points in $R_{[1,a],[b_k+B/2,b_k+B]}\cap F$ by the $\ve$-sparsity.
	
	If there fails to be a single possible choice of $b_{k+1}$ such that $|\psi| \leq 1$ on a $1-\ve'$ fraction of $R_{[1,a],[b_{k+1}-1,b_{k+1}]}$, then clearly $|\psi| > 1$ on a $C\ve'$ fraction of $R_{[1,a],[b_k+B/2,b_k+B]}$. Recall that we've already shown that the count of such lattice points is at most $\ve(1+C(\alpha')^{-1/2}) a^{3/2}( \log a)^{-1/2}$. On the other hand, the total count of points in our rectangle $R_{[1,a],[b_k+B/2,b_k+B]}$  is at least $ca^{3/2}(\alpha' \log a)^{-1/2}$.
	
	Consequently, combining our two bounds and cancelling the common terms, we have\begin{equation}\label{failure}
		\ve(1+C(\alpha')^{-1/2}) \geq c\ve'
	\end{equation} if it is impossible to find the desired $b_{k+1}$. Recalling that $\alpha' > 1$, taking $\ve < c\ve' (\alpha')^{-1/2}$ suffices to make (\ref{failure}) impossible and guarantee existence of the desired $b_{k+1}$.
	Having shown the claim, we now apply the bound which holds because we are in the event $\mathcal{E}_{ni}$ successively to the sequence of boxes $R_{[1,a],[b_k-1,b_{k+1}]}$.\\
	
	We obtain:
	\[ \|\psi\|_{\ell^\infty(R_{[1,a],[b_k-1,b_{k+1}]})} \leq e^{\alpha' (b_{k+1}-b_k)\log a}(1 + \|\psi\|_{\ell^\infty(R_{[1,a],[b_k-1,b_k]})})\]
	At each successive application, our increase is at most by a multiplicative factor of $e^{C\alpha'B \log a}$. It takes at most $Ca/B$ iterations to reach $b_K \geq 2a$. Hence we obtain:
	\[ \|\psi\|_{\ell^\infty(R_{[1,a],[1,2a]})} \leq (e^{C\alpha' B \log a})^{Ca/B}\]
	Taking $\alpha \geq C\alpha'$ we get the desired bound.
\end{proof}
Having proven \Cref{growth}, we can now proceed to a proof of \Cref{ucucthmfinal}.
\begin{proof}[Proof of \Cref{ucucthmfinal}]
	We first reduce the problem of unique continuation for standard squares to that of proving the same for tilted squares. For a tilted square $Q \subset \Z^2$ and some set $F \subset \Z^2$, we let $\mathcal{E}_{uc}'$ denote the event that
	\begin{equation}\label{uctiltprimehyp}\begin{cases}
			|E-\overline{E}|\leq e^{-\beta(\ell(Q)\log\ell(Q))^{1/2}}\\
			H\psi = E\psi\text{ in }Q\\
			|\psi|\leq 1\text{ in a }1-\ve(\ell(Q)\log\ell(Q))^{1/2}\text{ fraction of }Q \setminus F
	\end{cases}\end{equation}
	implies $|\psi| \leq e^{\beta \ell(Q)\log\ell(Q)}$ in $\frac{1}{64}Q$.
	
	This event is essentially our unique continuation event, except holding on tilted squares. It suffices to show for appropriate $\beta$, $\delta$, $\ell(Q) \geq \beta$ and $F$ $\delta$-regular we have:

	\begin{equation}\label{uctilt}\P[\mathcal{E}_{uc}'|V_F=v]\geq 1-e^{-\delta \ell(Q)^{1/4}}\end{equation}
	This essentially follows by a covering argument.
	\begin{claim}\label{squareslist}
	Let $\Lambda$ be a square with $\ell(\Lambda) \geq 2^{100}$. There exists a collection of tilted squares satisfying the following:
	\begin{enumerate}[label=(\Alph*)]
		\item $\frac{1}{2}\Lambda \subset \cup_{n\leq N} \frac{1}{64}Q_n$
		\item $\cup_{n\leq N} Q_n \subset \Lambda$
		\item $\min_{n\leq N} \ell(Q_n) \geq \frac{1}{256}\ell(\Lambda)$
		\item $N \leq 2^{50}$
	\end{enumerate}
	\end{claim}
	Indeed, one can just start with the list of squares
	\[\{R_{[\frac{1}{256}\ell(\Lambda) + \frac{n}{1024}\ell(\Lambda),\frac{1}{256}\ell(\Lambda) + \frac{m}{1024}\ell(\Lambda)]}\,:\,m,n\in\Z^2\}\]
	and restrict to the list of such squares which non-trivially intersect $\frac{1}{2}\Lambda$. (We left out ceiling operations above for ease of reading, but of course one needs to round to get integer side lengths and corner coordinates.)
	
	Clearly if $\mathcal{E}_{uc}'(Q_n,F)$ holds for all $Q_n$ covering $\Lambda$ in such a way, then $\mathcal{E}_{uc}(\Lambda,F)$ holds. Moreover, if $F$ is $\ve$-regular in $\Lambda$, then it is $\delta$-regular in all the $Q_n$ for $\delta < \frac{1}{256}\ve$. In particular, proving \ref{uctilt} under the appropriate hypotheses for some $\delta$ proves \Cref{ucucthmfinal} for $\ve = \frac{1}{256^2} \delta$. (We essentially change variables so as to use $\ve$ for the constant coming from \Cref{growth}.) Indeed; the hypotheses of \Cref{ucucthmfinal} imply the hypotheses for all the $Q_n$, and the necessary bound follows by a union bound over the events associated to the $Q_n$, which are absolutely bounded in number.
	
	Towards showing the bound holds, we let $\alpha > 1 > \ve > 0$ be constants so that the conclusions of \Cref{growth} hold. We emphasize that this $\ve$ is not the $\ve$ from the hypotheses of \Cref{ucucthmfinal}; as we have said earlier, the $\ve$ from said hypotheses is $c\delta$. We will assume $\delta< \ve^2$ and $\beta > 2\alpha$; we may need to take $\delta$ even smaller, and $\beta$ even larger. In general, we allow $\beta$ to depend on $\delta$. For fixed $Q$ we let $\mathcal{Q}$ denote the set of tilted squares $Q' \subset Q$ such that
	\begin{enumerate}[label=(\Alph*)]
		\item $\ell(Q') \geq \ell(Q)^{1/4}$
		\item $2Q' \subset \frac{1}{2}Q$
		\item $F$ is $\ve$-sparse in $2Q'$
		\item $\frac{1}{4}Q' \cap \frac{1}{64}Q \neq \varnothing$
	\end{enumerate}
	We recall that $\mathcal{E}_{ex}(Q,F)$ is the event from \Cref{growth} which says, roughly, that whenever eigenfunctions are small on all of $\frac{1}{2}Q$ and most of $2Q\setminus F$, then they are not huge anywhere on $Q$. Then, by a straightforward union bound, if we define the event
	\[\mathcal{E}_{ex} = \bigcap_{Q'\in\mathcal{Q}} \mathcal{E}_{ex}(Q',F)\cap\mathcal{E}_{ex}(2Q',F)\]
	we have $\P[\mathcal{E}_{ex}|V_F=v] \geq 1-e^{-\ve\ell(Q)^{1/4}-C\log \ell(Q)} \geq 1- e^{-\delta \ell(Q)^{1/4}}$, with the logarithmic term again coming from the fact that $|\mathcal{Q}| \leq C\ell(Q)^3$. We've reduced the problem to showing $\mathcal{E}_{ex} \subset \mathcal{E}_{uc}'(Q,F)$, and henceforth seek to show (\ref{uctiltprimehyp}) and $\mathcal{E}_{ex}(Q)$ imply the conclusion of unique continuation for tilted squares, i.e. the bound
	\begin{equation}\label{eigenboundgrowth}|\psi| \leq e^{\beta \ell(Q)\log\ell(Q)}\end{equation} on $\frac{1}{64}Q$. Thus going forward, we consider eigenfunctions $\psi$ satisfying the hypotheses in (\ref{uctiltprimehyp}). We'll call these eigenfunctions with ``sparse large support'' going forward; the terminology leaves something to be desired as it lacks any reference to energy, whereas (\ref{uctiltprimehyp}) supposes the associated energy of the eigenfunction to live in a small band. Nevertheless, it's concise enough to work with.
	
	The basic idea of the rest of proof is to consider the subsquares $Q'$ on which the eigenfunctions with sparse large support obey the analogous bound, i.e. subsquares such that eigenfunctions (associated to the relevant energy range) satisfy (\ref{eigenboundgrowth})
	and show that there is such a square containing all of $\frac{1}{64}$.  As we have emphasized, our work in the latter half of this section very closely follows \cite[Section 3]{Ding-Smart}; we should mention, as was mentioned in said paper, that the idea is a random version of a proof from \cite{Buhovsky22}.
	
	We let $\mathcal{Q}_{bd}$ be the subset of $\mathcal{Q}$ consisting of $Q'$ obeying (\ref{eigenboundgrowth}). We let $\mathcal{Q}_{mbd}$ be the set of $Q' \in \mathcal{Q}_{bd}$ which are maximal in said set with respect to inclusion. For such $Q'$, at least one of the following is true:
	\begin{enumerate}[label=(\Alph*)]
		\item \label{toobig}$4Q' \not\subset \frac{1}{2}Q$
		\item \label{nosparse}$F$ is not $\delta$-sparse in $2Q'$
		\item \label{largesupportlarge}$|\{|\psi|\geq 1\}\cap 4Q'\setminus F|\geq \ve(\ell(Q')\log \ell(Q'))^{-1/2}|4Q'\setminus F|$
	\end{enumerate}
	In particular, if all these conditions fail for some $Q' \in \mathcal{Q}_{bd}$, the $\mathcal{E}_{ex}$ event will allow us to find $Q'' \in \mathcal{Q}_{bd}$ strictly containing $Q'$. Indeed, the failure of \ref{nosparse} means in particular $F$ is also $\ve$-sparse in $2Q'$. Taken together with the failure of condition \ref{toobig}, we conclude that $2Q' \in \mathcal{Q}$ and hence $\mathcal{E}_{ex} \subset \mathcal{E}_{ex}(2Q',F)$. If \ref{largesupportlarge} also fails, it follows from the fact that we are in the $\mathcal{E}_{ex}(2Q',F)$ event that:
	\[ \|\psi\|_{\ell^\infty(2Q')} \leq e^{\alpha \ell(2Q')\log(\ell(2Q'))}\max\{1,\|\psi\|_{\ell^\infty(Q')}\}\]
	on $2Q'$. In particular, if $\psi$ is an eigenfunction with sparse large support and $\|\psi\|_{\ell^\infty(Q')} \geq 1$, it's obvious that its $\ell^\infty(Q')$ normalization also has sparse large support, and then this normalized version satisfies \ref{exhyp}, whence we get the bound. If $\|\psi\|_{\ell^\infty(Q')} \leq 1$, no normalization is needed. Taking $\beta \geq 2\alpha$, the above equation (together with $Q' \in \mathcal{Q}_{bd}$) will give
	\begin{equation}\|\psi\|_{\ell^\infty(2Q')} \leq e^{\beta \ell(2Q')\log \ell(2Q')}
	\end{equation}
	This of course implies $2Q' \in \mathcal{Q}_{bd}$ and $Q'\notin \mathcal{Q}_{mbd}$. It turns out that the squares in $\mathcal{Q}_{mbd}$ for which either condition \ref{nosparse} or \ref{largesupportlarge} hold cover a small amount of $Q$. Specifically, we let $\mathcal{Q}_{mbd}^\ast$ denote the subset of $\mathcal{Q}_{mbd}$ of $Q'$ such that condition ($\ast$) fails, where $\ast = A,B,C$. (Our eventual aim is to show that $\mathcal{Q}_{mbd}\setminus (\mathcal{Q}^B_{mbd}\cup \mathcal{Q}_{mbd}^C)$ is non-empty.) We will set as a shorthand $\cup \mathcal{Q}_{mbd}^\ast :=\bigcup_{Q' \in \mathcal{Q}^\ast_{mbd}}Q'$ and define $\cup \mathcal{Q}_{mbd}$ and $\cup \mathcal{Q}_{bd}$ analogously.
	\begin{claim}\label{smallbc}
		For side lengths $L$ sufficiently large:
		\[ |\cup\mathcal{Q}^B_{mbd}| + |\cup\mathcal{Q}^C_{mbd}| \leq C\left(\delta + \frac{\delta}{\ve}\right)|Q|\]
	\end{claim}
	We will let $\cup \mathcal{Q}_{mbd}^{C\setminus B}$ denote $\cup \mathcal{Q}^C_{mbd}\setminus \cup\mathcal{Q}_{mbd}^B$. Using the Vitali covering lemma, we can find for $\ast = B, C\setminus B$ a collection of squares $Q^\ast_1,\dots,Q^\ast_{N_\ast} \in \mathcal{Q}^\ast_{mbd}$ so that for $n\neq m$ we have $4Q^\ast_n \cap 4Q^\ast_m = \varnothing$ and moreover \begin{equation}\label{vitcov}
		|\cup_{n\leq N} Q^\ast_n|\geq \frac{1}{32}|\cup \mathcal{Q}^\ast_{mbd}|
	\end{equation} Specifically, the Vitali covering lemma guarantees we can find such a disjoint collection such that $\cup_{n\leq N} 4Q^\ast_n$ covers $\cup 4\mathcal{Q}^\ast_{mbd}$; heuristically then $|\cup_{n\leq N}Q^\ast_n| \geq \frac{1}{16}|\cup \mathcal{Q}_{mbd}^\ast|$. Given the lower bound on the length $\ell(Q') \geq \ell(Q)^{1/4}$, this heuristic is correct up to an $O(\ell(Q)^{-1/4})$ correction (this comes from the ratio of the perimeter to the area), so that (\ref{vitcov}) holds for sufficiently large $\ell(Q)$. It then suffices to bound $|\cup_{n\leq N}Q_n^\ast|$.
	
	Because $F$ is $\delta$-regular in $Q$, it's immediate that $|\cup_{n\leq N} Q^B_n| \leq \delta |Q|$.  To control $|\cup_{n\leq N}Q^{C\setminus B}_n|$, we first note that by the $\delta$-sparseness of $F$ in all these squares, we have:
	\[|\cup_{n\leq N}Q^{C\setminus B}_n\setminus F| \geq c|\cup_{n\leq N}Q^{C\setminus B}_n|\]
	whence the necessary bound follows almost immediately. Indeed, the proportion of $Q\setminus F$ on which $|\psi| > 1$ is bounded above by $\ve (\ell(Q)\log \ell(Q))^{-1/2}$; the proportion  of sites in $\cup_{n\leq N}Q^{C\setminus B}_n \setminus F$ so that $|\psi| > 1$ is at least 
	\[\delta\min_{n\leq N} \ell(Q_n^{C\setminus B})\log (\ell(Q_n^{C\setminus B}) )^{-1/2} \geq \delta (\ell(Q)\log \ell(Q))^{-1/2}\] 
	Given the upper bound on the total count of sites in $Q\setminus F$ where $|\psi| > 1$, this gives $|\cup_{n\leq N} Q_n^{C\setminus B}| \leq \frac{\delta}{\ve}$. Having shown this, we will obtain the existence of some $Q' \in \mathcal{Q}_{mbd}\setminus(\mathcal{Q}_{mbd}^B \cup \mathcal{Q}_{mbd}^C)$ from the following claim:
	\begin{claim}\label{bigbd}
		For all sufficiently small $\delta >0$, we have $|\cup \mathcal{Q}_{mbd}| = |\cup \mathcal{Q}_{bd}| \geq c |Q|$.
	\end{claim}
	By taking e.g. squares of side length $\lceil\ell(Q)^{1/4}\rceil$ with corners at the points in $\lceil 4\ell(Q)^{1/4}\rceil\Z^2 \cap \frac{1}{128}Q$ (see the construction for \Cref{squareslist}), we can produce a list of tilted squares $Q'_1,\dots, Q'_N$ such that
	\begin{enumerate}[label=(\Alph*)]
		\item\label{constantbounds} $\ell(Q)^{1/4} \leq \ell(Q'_n)  \leq 2\ell(Q)^{1/4}$
		\item $Q'_n \subset \frac{1}{64}Q$
		\item $2Q'_n\cap 2Q'_m = \varnothing$ for $n\neq m$
		\item $K \geq c\ell(Q)^{3/2}$
	\end{enumerate}
	In particular, the $Q'_n$ and their doublings $2Q'_n$ both cover a positive proportion of $Q$. As such, by $\delta$-regularity of $F$, we have (for $\delta$ sufficiently small) $F$ $\delta$-sparse in at least a $1-C\delta$ proportion of the doublings $2Q'_n$ when they are weighted by their cardinalities. Given \ref{constantbounds}, the same thing holds (with a different constant) even when the squares are weighted equally.
	
	Note also that there are at most $C\ve \ell(Q)^{3/2} \log \ell(Q)^{-1/2}$ points in $Q\setminus F$ such that $|\psi| > 1$; by either taking $\ve$ small or taking $\ell(Q)$ large enough that $\ve\log(\ell(Q))^{-1/2}$ is small, we can obtain that for a positive proportion of these $Q'_n$, $F$ is $\delta$-sparse and $|\psi|\leq 1$ on all of $Q'_n$. This positive proportion of the $Q'_n$ are all in $\mathcal{Q}$ (the only thing to verify for this is the $\delta$-sparsity), and $|\psi| \leq 1$ implies that they are further in $\mathcal{Q}_{bd}$. Hence $\mathcal{Q}_{bd}$ contains at least $c\ell(Q)^{3/2}$ disjoint squares of at area at least $\ell(Q)^{1/2}$, yielding $|\cup\mathcal{Q}_{bd}| \geq c|Q|$.
	
	It follows immediately, combining the previous two claims, that for $\delta > 0$ sufficiently small there is $Q' \in \mathcal{Q}_{mbd}$ such that condition \ref{toobig} holds, but conditions \ref{nosparse} and \ref{largesupportlarge} fail. In particular, $\mathcal{E}_{ex}(2Q',F) \subset \mathcal{E}_{ex}$ and given the failure of the later two conditions, $2Q'$ satisfies (\ref{eigenboundgrowth}). Simultaneously, $4Q' \not\subset \frac{1}{2}Q$ and $\frac{1}{4Q'} \cap \frac{1}{64} Q$ together imply a lower bound on $\ell(2Q')$.
	
	Specifically, the diameter of any tilted square $Q'$ (in the Euclidean metric) is $\sqrt{2}\ell(Q')$, up to $O(1)$ corrections. We also have, up to $O(1)$ corrections, which we will stop mentioning now, $\text{dist}(\frac{1}{64}Q,\Z^2\setminus \frac{1}{2}Q) =   \frac{31}{128}\sqrt{3}\ell(Q)$ with respect to this same metric. Combining these bounds gives e.g. $\ell(4Q') \geq \frac{1}{4} \ell(Q)$ for our $Q'$, 1/4 being smaller than $\frac{31}{128}\sqrt{\frac{3}{2}}$, because $4Q'$ non trivially intersects both $\frac{1}{64}Q$ and the complement of $\frac{1}{2}Q$.
	
	This gives the lower bound $\ell(Q') \geq \frac{1}{16}\ell(Q)$. Taken together with $\frac{1}{4}Q' \cap \frac{1}{64}Q=\varnothing$ this gives $\frac{1}{64}Q \subset 2Q'$. Indeed, $2Q'$ contains a ball of radius $\frac{3}{2}\ell(Q') \geq \frac{3}{32}\ell(Q)$ around any element of $\frac{1}{2}Q'$. Given that the diameter of $\frac{1}{64}Q$ is $\frac{\sqrt{2}}{64}\ell(Q)$, we conclude $\frac{1}{64}Q \subset 2Q'$ and hence (\ref{eigenboundgrowth}) holds on $\frac{1}{64}Q$; this completes the proof.
\end{proof}

\section{Deterministic preliminaries to MSA}\label{detmsasec}

We now prove (or formulate) various entirely deterministic results necessary to carry out the MSA; the first two are the main results of \cite[Section 6]{Ding-Smart} and the third is analogous to the main result of \cite[Section 7]{Ding-Smart}. To formulate the first result succinctly, we introduce the notion of the boundary of a square $\Lambda'$ within $\Lambda$.

\begin{deffo}
	The boundary of $\Lambda' \subset \Lambda$, denoted $\partial \Lambda'$, is the set of all pairs of lattice sites $(u,v)$ with $u \in \Lambda'$, $v \in \Lambda\setminus \Lambda'$, and $|u-v| = 1$.
\end{deffo}
Throughout this section, $R_\Lambda = (H_\Lambda - E)^{-1}$ for some fixed energy $E$; this inverse is taken in the Hilbert space $\ell^2(\Lambda)$ and not $\ell^2(\Z^2)$. For $\Lambda'\subset \Lambda$, $R_{\Lambda'} = (H_{\Lambda'}-E)^{-1}$, with this inverse taken in $\ell^2(\Lambda)$. The first result we need is a straightforward consequence of the geometric resolvent identity
\begin{equation}
	R_\Lambda(x,y) = R_{\Lambda'}(x,y) + \sum_{(u,v) \in \partial\Lambda'} R_{\Lambda'}(x,u)R_{\Lambda}(v,y)
\end{equation}
which holds for squares (or more general regions) $\Lambda' \subset \Lambda$.
\begin{lem}\label{geomres}
	If $x \in \Lambda' \subset \Lambda$ and $y \in \Lambda$, then
	\begin{equation*}
		|R_\Lambda(x,y)| \leq |R_{\Lambda'}(x,y)| + |\partial\Lambda'|\max_{(u,v)\in\partial\Lambda'} |R_{\Lambda'}(x,u)|\cdot|R_{\Lambda}(v,y)|
	\end{equation*}
\end{lem}
It is also convenient to formulate this in terms of the maximal boundary pair; if $(u,v)$ maximizes $|R_{\Lambda'}(x,u)R_\Lambda(v,y)|$, we have:
\begin{equation*}
	|R_\Lambda(x,y)| \leq |R_{\Lambda'}(x,y)| + |\partial\Lambda'||R_{\Lambda'}(x,u)|\cdot|R_{\Lambda}(v,y)|
\end{equation*}

In the specific case where $\Lambda'$ is a square, we have $|\partial \Lambda'| \leq 4\ell(\Lambda') +4$; in the cases where $\ell(\Lambda') \ll \ell(\Lambda)$, we have the bound below, which will suffice going forward.
\begin{equation}
	|R_\Lambda(x,y)| \leq |R_{\Lambda'}(x,y)| + |\Lambda||R_{\Lambda'}(x,u)|\cdot|R_{\Lambda}(v,y)|
\end{equation}
We now use this result to prove the deterministic part of the MSA. The following is, in some form, well established and well known by now, but we find it worthwhile to include a proof, so we present the one from \cite{Ding-Smart}.
\begin{lem}\label{detmsa}
	If all of the following hold:
	\begin{enumerate}[label=(\Roman*)]
		\item $\ve > \nu> 0$ are small enough
		\item $K \geq 1$ is an integer
		\item $L_0\geq \cdots \geq L_6 \geq C_{\ve,\nu,K}$ are dyadic scales satisfying $(1-\ve) \log_2 L_k \geq \log_2 L_{k+1}$
		\item $1 \geq m \geq 2L_5^{-\nu}$ is an exponential decay rate
		\item $\Lambda \subset \Z^2$ is a square with $\ell(\Lambda) = L_0$
		\item $\Lambda_1',\dots,\Lambda_K' \subset \Lambda$ are disjoint $L_2$-squares with $|R_{\Lambda'_k}|\leq e^{L_4}$
		\item for all $x \in \Lambda$ at least one of the following holds
		\begin{enumerate}[label=(\roman*)]
			\item there is $\Lambda'_k$ such that $x \in \Lambda'_k$ and $\text{dist}(x,\Lambda\setminus \Lambda') \geq \frac{1}{8}\ell(\Lambda'_k)$ i.e. $x$ is ``deep inside'' a defect square $\Lambda'_k$
			\item there is a square $\Lambda'' \subset \Lambda$ with $\ell(\Lambda'') = L_5$ such that $x \in \Lambda''$, $\text{dist}(x,\Lambda\setminus \Lambda'') \geq \frac{1}{8} \ell(\Lambda'')$, and $|R_{\Lambda''}(y,z)| \leq e^{L_6-m|y-z|}$ for $y,z \in \Lambda''$, i.e. $x$ is ``deep inside'' a square obeying the appropriate resolvent bound. (Note that these squares are much smaller than the defect squares.)
		\end{enumerate}
	\end{enumerate}
	then $|R_\Lambda(x,y)| \leq e^{L_1-\tilde{m}|x-y|}$ for $x,y \in \Lambda$, where $\tilde{m} = m - \L_5^{-\nu}$.
\end{lem}
\begin{proof} 
	
	Throughout, $C> 1 > c > 0$ will depend on the parameters $\ve,\nu, K$. We start by putting a weighted multigraph structure on $\Lambda$. For any $x,y \in \Lambda$, we add an edges $x \rightarrow y$ and $y\rightarrow x$ with weight $|x-y|$. If for some $k$ we have $x \in \Lambda'_k$, $y \in \Lambda\setminus \Lambda'$, $\text{dist}(x,\Lambda\setminus \Lambda') \geq\frac{1}{8}\ell(\Lambda'_k)$ and $\text{dist}(y,\Lambda'_k) =1$, we add an edge $x \rightarrow y$ with weight $-L_3$. (So, we add edges to $y$ on the ``outside boundary'' of defect squares from $x$ ``deep inside'' those same defect squares.)
	
	For sufficiently large scales, $\frac{1}{8}L_2 \geq L_3$, so (using also that the defect squares are disjoint) there are no cycles with negative total weight. As a consequence, if we set $d(x,y)$ to be the minimum weight over all paths from $x$ to $y$, we have the triangle inequality
	\[ d(x,y) \leq d(x,z) + d(z,y)\]
	
	Because there are only finitely many defects, we have the following:
	\[ |x-y|-KL_2 \leq d(x,y) \leq |x-y| \]
    Indeed, because $\frac{1}{8}L_2 \geq L_3$, moving into a defect, out of it, and into it again cannot have negative weight. Hence between any $x$ and $y$ there is a weight-minimizing path exiting each defect square at most once.\\
	
	We now set
	\[ \alpha := \max_{x,y \in \Lambda} e^{\tilde{m}d(x,y)}|R_\Lambda(x,y)|\]
	and the result will follow from estimates on $\alpha$.
	
	For $x,y \in \Lambda$, and $x$ ``deep inside'' a defect square, i.e. case (i) from our hypothesis, then \Cref{geomres} gives us $(u,v) \in \partial \Lambda'_k$ such that
	\[ |R_\Lambda(x,y)| \leq |R_{\Lambda'_k}(x,y)| + L_0^2|R_{\Lambda'_k}(x,u)|\cdot|R_\Lambda (v,y)|\]
	By the triangle inequality established for $d$, we have $d(x,y) \leq d(v,y)-L_3$. Further, by the definition of $\tilde{m}$ and conditions on our scales, we have \[\tilde{m}L_3 - L_4 - 2 \log L_0 \geq \log 2\]for sufficiently large scales. The estimates below follow:
	\begin{align*}
		|R_\Lambda(x,y)| &\leq e^{L_4} 1_{\Lambda_j'}(y) + L_0^2\cdot e^{L_4}\cdot \alpha e^{-\tilde{m}d(v,y)}\\
		&\leq e^{L_4} 1_{\Lambda_j'}(y) + L_0^2\cdot e^{L_4-\tilde{m}L_3}\cdot \alpha e^{-\tilde{m}d(x,y)}\\
		&\leq e^{L_4} 1_{\Lambda_j'}(y) + \frac{1}{2}\alpha e^{\tilde{m}d(x,y)}
	\end{align*}
	As $\tilde{m}d(x,y) \leq 2\tilde{m}L_2$ for $y \in \Lambda'_k$, and $L_2 > L_4$, we obtain (under the assumption $x$ is ``deep inside'' a defect) the bound:
	\begin{equation}\label{estweak}
		e^{\tilde{m}d(x,y)}|R_\Lambda(x,y)| \leq e^{CL_2} +\frac{1}{2}\alpha
	\end{equation}
	We now show a stronger estimate holds in the case where $x$ is ``deep inside'' a smaller good square, i.e. case (ii) from our hypothesis. Again, for the good square $\Lambda''$, there is some pair $(u,v) \in \partial\Lambda''$ such that
	\[ |R_\Lambda(x,y)| \leq |R_{\Lambda'_k}(x,y)| + L_0^2|R_{\Lambda'_k}(x,u)|\cdot|R_\Lambda (v,y)|\]
	We have by definition of $\tilde{m}$ and the conditions on the scales the estimate
	\[(m-\tilde{m})|x-u| - L_6 - 2\log L_0 \geq \log 2\]
	whence (making use of the triangle inequality for $d$) we obtain the following:
	\begin{align*}
		|R_\Lambda(x,y)| &\leq e^{-L_5}1_{\Lambda''}(y) + L_0^2 \cdot e^{L_6}\cdot \alpha e^{-\tilde{m}d(v,y)}\\
		&\leq e^{L_6}1_{\Lambda''}(y) + L_0^2\cdot e^{L_6-(m-\tilde{m})|x-u|+1}\cdot \alpha e^{-\tilde{m}d(x,y)}\\
		&\leq e^{L_6} 1_{\Lambda''}(y) + \frac{1}{2}\alpha e^{-\tilde{m}d(x,y)}
	\end{align*}
	In case $y \in \Lambda''$ as well, we have $\tilde{m}d(x,y) \leq CL_5$, and so we have
	\begin{equation}\label{eststrong}e^{\tilde{m}d(x,y)}|R_\Lambda(x,y)| \leq e^{CL_5} + \frac{1}{2}
	\end{equation}
	In particular, all $x$ fall into one of the two cases, and so the weaker estimate (\ref{estweak}) applies globally; the immediate implication is that
	\[ \alpha \leq e^{CL_2}\]
	whence it follows that 
	\[ |R_\Lambda(x,y)| \leq e^{CL_2-\tilde{m}d(x,y)} \leq e^{L_1-\tilde{m}|x-y|}\]
	for sufficiently large scales where $L_1 > C(L_2+L_3)$.	
\end{proof}

The next result we need says, roughly, that if at some energy we have nice resolvent bounds, then we have almost as nice bounds for energies nearby. The proof is based on that presented in \cite{Ding-Smart}, but our bounds are slightly better due to slight changes to the argument.
\begin{lem}\label{contresbound}
	Let $\Lambda \subset \Z^2$ be a square of side length $L$ and $\overline{E}$ some energy. If $\alpha > \beta > 0$ and we have for $\overline{E}$:
	\begin{equation}\label{keyresbound} |(H_\Lambda-\overline{E})^{-1}(x,y)| \leq e^{\alpha - \beta|x-y|} \text{ for } x,y \in \Lambda \end{equation}
	then for energies $E$ satisfying\[|E-\overline{E}| \leq \frac{1 }{2L^2 e^{\alpha}}\]
	we have the bound
	\[|(H_\Lambda - E)(x,y)| \leq 2e^{\alpha - \beta|x-y|}\text{ for }x,y\in \Lambda \]
\end{lem}
\begin{proof}
	We drop the $\Lambda$ for the rest of the proof, just writing $H$ for $H_\Lambda$, and we use the resolvent identity
	\begin{equation}\label{resid}(H-E)^{-1} = (H-E)^{-1} + (H-\overline{E})^{-1}(E-\overline{E})(H-\overline{E})^{-1}\end{equation}
	We set
	\[ \gamma:= \max_{x,y \in \Lambda} e^{\beta|x-y|-\alpha} |(H-E)^{-1}(x,y)|\]
	and seek to show $\gamma \leq 2$.
	
	By assumption, $(H-\overline{E})^{-1}(x,y) \leq e^{\alpha-\beta|x-y|}\gamma$, so that in particular \begin{align*}\left|[(H-\overline{E})^{-1}(H-E)^{-1}](x,y)\right| &\leq \sum_{z \in \Lambda} e^{\alpha-\beta|x-z|}e^{\alpha-\beta|z-y|}\gamma\\
	&\leq L^2 e^{2\alpha - \beta|x-y|}\gamma
	\end{align*}
	However, by (\ref{resid}), we obtain
	\[ |(H-E)^{-1}(x,y)| \leq e^{\alpha-\beta|x-y|} + |\overline{E}-E|L^2e^{2\alpha-\beta|x-y|}\gamma \]
	for all $x,y \in \Lambda$. Dividing both sides by $e^{\alpha-\beta|x-y|}$ and taking a maximum gives
	\[ \gamma \leq 1 + |E-\overline{E}|L^2 e^\alpha \gamma \]
	or, after rearranging,
	\[ \gamma \leq \frac{1}{1-|E-\overline{E}|L^2e^\alpha}\]
	Taking $|E-\overline{E}| \leq \frac{1}{2L^2e^\alpha}$ gives $\gamma \leq 2$, and so the desired result.
\end{proof}
Finally, we need a result which will serve as the base case of the MSA; we need to recall a definition from the theory of metric spaces here.
\begin{deffo}
	If $(Y,d)$ is a metric space and $R>0$, we say a non-empty subset $ X \subset Y$ is an $R$-net in $Y$ if $\sup_{y\in Y} \inf_{x \in X} d(x,y) \leq R$.
\end{deffo}
The necessary result is (a slight variation on) \cite[Lemma 7.2]{Ding-Smart}:
\begin{lem}[\cite{Ding-Smart}]\label{rnet}
	If $\kappa > 0$, there are $\ve  >0$ and $0 < c< 1 < C$ (depending on $\kappa$) such that if $R$ is sufficiently large and
	\begin{enumerate}[label=(\Alph*)]
		\item  $\Lambda \subset \Z^2$ is a box with side length at least $\tilde{R} := R^2\log R$
		\item $\{n \in\Z^2 \,:\,V_n \geq \kappa\}$ is an $R$-net in $\Lambda$
	\end{enumerate} then we have the estimate
	\[ |H_\Lambda^{-1}(x,y)|\leq e^{C \tilde{R}-c \tilde{R}^{-1}|x-y|}\]
\end{lem}
Ding and Smart proved this for $\kappa =1$, but their proof works for any positive $\kappa$, and we will use this fact to obtain our initial scale estimate with $\kappa = \gamma$.

\begin{remm}\label{ltail}
	This result will be used to derive an ``initial scale estimate'' which serves as a base case for our inductive argument in \Cref{msa}. The argument will be completed by basic lower bounds for the probability that $\{V \geq \kappa\}$ is an $R$-net for $R$ chosen appropriately. This argument in particular demonstrates that we are unlikely to find eigenvalues close to 0; this is a manifestation of a phenomenon known as ``Lifschitz tails''.
	
	A detailed discussion of this phenomenon is beyond the scope of this paper, but we will comment briefly on this phenomenon; our methods suffice to prove localization in energy ranges where this phenomenon holds ``strongly enough''. Roughly, Lifschitz tails describe a scarcity of eigenvalues near 0, or more generally near the boundary of the spectrum; with high probability the number of eigenvalues in $[0,E]$ ``per unit volume'' is of order $e^{-CE^{-d/2}}$, where $d$ is the dimension, with high probability. Precisely, Lifschitz tails corresponds to such asymptotics for a limiting object which describes the expected count of ``eigenvalues per unit volume'' called the integrated density of states.
	
	For models like ours, this phenomenon (or at least the upper bound) was shown to hold rigorously in \cite{Simon1985}; the main result was the bounds for the limiting object, but these were obtained by getting probabilistic bounds in finite volume. The corresponding lower bounds hold if a certain quantitative strengthening of the hypotheses in \Cref{suffcond} hold universally, i.e. without an exceptional density zero subset.
	
	As was discussed earlier, in general we know very little about the spectrum of $H$; only the essential spectrum is deterministic and we know very little about its topological structure. In full generality it is unlikely that we can say much. In analogous one dimensional models there are explicit examples where the behavior is very different from what is seen for i.i.d. models \cite[Proposition A.1]{gorodetski2024nonstationary}.
	
	It is nevertheless reasonable to expect that for many concrete models of interest either the whole spectrum or at least the essential spectrum is a union of intervals; this is the case for the i.i.d. model.  If the spectrum does have such a structure, it is also natural to expect ``internal'' Lifschitz tails, as were found for i.i.d. models in \cite{PhysRevB.32.6272}, see also \cite{Simon1987}. Such results, once found, should imply localization at these ``internal edges''  via our methods, so long as the proof is robust enough to be compatible with the frozen sites formalism.
\end{remm}

Finally, we need a covering lemma \cite[Lemma 8.1]{Ding-Smart}, which roughly says that given $K$ defect squares we can cover them by the same amount of squares, larger by a constant only depending on $K$, so that they sit well inside these larger squares.
\begin{lem}[\cite{Ding-Smart}]\label{coveringlemma}
	For $K$ a positive integer and $\alpha \geq C^K$ dyadic, dyadic scales $L_0 \geq \alpha L_1 \geq L_1 \geq \alpha L_2 \geq L_2$, an $L_0$ square $\Lambda$, and $K$ many $L_2$ squares $\Lambda''_1,\dots,\Lambda''_K$, there is a dyadic scale $\tilde{L} \in [L_1,\alpha L_1]$, and corresponding $\tilde{L}$ boxes $\Lambda_1',\dots,\Lambda_K'$, such that for every $j =0,\dots, K$, there is some $k=0,\dots,K$ such that $\Lambda_j'' \subset \Lambda_j'$ and $\text{dist}(\Lambda_j'',\Lambda \setminus \Lambda_j') \geq \frac{L_1}{8}$.
\end{lem}
	This lemma is entirely deterministic, so the original proof holds in our non-stationary context.
\section{Multiscale Analysis}\label{msasec}
We can now carry out a multiscale analysis, which implies \Cref{endofmsa2} and is our version of \cite[Theorem 8.3]{Ding-Smart}:
\begin{thm}\label{msa}
	For $\gamma \in (0,1/2)$, there are:
	\begin{enumerate}[label=(\Roman*)]
		\item \label{smallparam} small parameters $1 > \ve > \nu > \delta > 0$
		\item \label{basescalecount}an integer $\tilde{M} \geq C_{\ve,\delta}$
		\item \label{scales}dyadic scales $L_k$ for $k \geq 0$ satisfying
		\begin{equation*}
			\log_2 L_{k+1}  = \left\lfloor \frac{1}{1-6\ve}\log_2 L_k \right\rfloor
		\end{equation*}
		\item \label{decays}decay rates $m_k$ satisfying  $1 \geq m_k \geq L_k^{-\delta}$
		\item \label{sets} random sets $F_k \subset F_{k+1}$ such that
		\begin{enumerate}[label=(\roman*)]
			\item \label{regularity}$F_k$ is $\eta_k$-regular in $\Lambda$ for $\ell(\Lambda) \geq L_k$, with $\eta_k := \ve^2 + C_{\ve,\delta,M}\sum_{j=0}^k L_{j}^{-\ve} <\ve$
			\item \label{measurable}$F_k \cap \Lambda$ is $V_{F_{k-1}\cap 2\Lambda}$-measurable, i.e. determined totally by that data, for $\ell(\Lambda) \geq L_k$
			\item \label{resbound} if $\ell(\Lambda) = L_k$ and $\mathcal{E}_g(\Lambda)$ denotes the event that, for any $0 \leq E \leq e^{-L_{\tilde{M}}^{\delta}}$ we have:
			\begin{equation}\label{resboundexp}
				|(H_\Lambda-E)^{-1}(x,y)| \leq 2e^{L_k^{1-\ve}-m_k|x-y|}\text{ for } x,y \in \Lambda
			\end{equation}
			then $\mathcal{E}_g(\Lambda)$ holds with high probability; explicitly
			\begin{equation}\label{conditionalgood}
				\P[\P[\mathcal{E}_g(\Lambda)|V_{F_k\cap \Lambda}]=1] \geq 1-L_k^{-\gamma}
			\end{equation}
		\end{enumerate}
		\item \label{decaybound} $m_k \geq m_{k-1} \geq m_{k-1} - L_k^{-\nu}$ for $k >\tilde{M}$
	\end{enumerate}
\end{thm}
\begin{proof}
	We let $\ve$, $\delta$, $\nu$, $\tilde{M}$ and $L_k$ such that they satisfy the first three conditions. Further conditions on these will be imposed in the course of the proof. We set for the base case (which actually consists of the first $\tilde{M}$ scales) $F_k$ as follows: We set $F_k = \lceil \frac{2}{\ve^2}\rceil\Z^2$ to start, and throw away those points landing in every other box. We set $m_k = L_k^{-\delta}$ for $k=0,\dots,\tilde{M}$. We will also throughout the proof assume that all the squares we deal with are ``half-aligned.'' Specifically, we work with squares $\Lambda$ of the form
	\begin{equation}\label{align}
		\Lambda = x + [0,2^n)^2 \cap \Z^2\text{ for } x \in 2^{n-1}\Z^2
	\end{equation}
	Our first task is to verify that things work for our base case.
	\begin{claim}
		For $L_0$ taken sufficiently large, specifically $L_0 \geq C(\ve,\delta)$, all six conditions are satisfied for $k=0,\dots,\tilde{M}$.
	\end{claim}
	That \ref{smallparam}-\ref{decays} are satisfied is by assumption or construction, and \ref{decaybound} holds vacuously. All that really needs to be verified for the base case is \ref{sets}. Of course, there are really three conditions; the regularity condition \ref{regularity} is straightforward, $F_k$ are clearly $\ve^2$-regular, and regularity is ``monotone'' in the sense that regularity with respect to a parameter implies regularity with respect to a larger parameter; $\ve^2 \leq \eta_k$. To ensure $\eta_k < \ve$, we take $L_0$ large enough that $\sum_{j=0}^\infty L_j^{-\delta} < \ve - \ve^2$. Similarly, the measurability condition \ref{measurable} is straightforward; these $F_k$ are deterministic and hence automatically measurable.
	
	The condition \ref{resbound} for the scales $k=0,\dots, \tilde{M}$ follows from applying \Cref{rnet}, then \Cref{contresbound}. Specifically, we note that
	\[ |B_{L_k^{\delta}}(x)\cap F_k| \leq c\ve^4L_k^{2\delta/3}\]
	For every $y \in B_{L_k^{\delta}}(x)$, we have $\P[V_y < \gamma]\leq 1-\rho < 1$. By independence,
	\[\P[V_y < \gamma\text{ for all }y \in B_{L_k^{-\delta}}\cap F_k] \leq e^{-c\ve^4L_k^{\delta}}\]
	By a union bound, it follows that
	for a box $\Lambda$ of side length $L_k$:
	\[\P[\{V \geq y\}\cap F_k \text{ is not an } L_k^{2\delta/3}\text{-net in }\Lambda] \leq L_k^2e^{c\ve^4L_k^{\delta}}\]
	In particular, as soon as this bound holds, we have the following by \Cref{rnet}, and importantly this holds regardless of what happens outside $F_k$:
	\begin{align*}
		H_\Lambda^{-1}(x,y) &\leq e^{CL_k^{2\delta/3}\log L_k-cL_k^{-2\delta/3}\log^{-1}L_k|x-y|}\\ 
		&\leq e^{L_k^{1-\ve} - m_k|x-y|}
	\end{align*}
	with the second inequality holding for $L_0$ sufficiently large. Moreover, by \Cref{contresbound}, we can extend this bound to small $E$. Then we have, for $0 \leq E \leq e^{L_k^{-\delta}}$, that (\ref{resboundexp}) holds. We reiterate that this is true independent of what the potential is on the complement of $F_k$ as long as $\{V \geq \gamma\} \cap F_k$ is an $L_k^{2\delta/3}$-net.
	
	Finally, for $L_0$ sufficiently large, we have
	\[ L_k^2e^{-c\ve^4L_k^{2\delta/3}} \leq L_k^{-\gamma}\]
	verifying the last property \ref{resbound} for $k=0,\dots,\tilde{M}$.
	
	Having shown the base case, we now proceed to define the $m_k$ and $F_k$ for $k > \tilde{M}$ and prove some claims which set up the inductive step. It is necessary to show that we can choose $\tilde{M}$ such that
	\begin{equation}\label{emmbound} \frac{\delta}{4} \log_2 L_{k+\tilde{M}} \leq \log_2 L_k \leq \delta \log_2 L_{k+\tilde{M}}
	\end{equation}
	Specifically:
	\begin{claim}
		There are $C$ and $\tilde{M}$ (depending on $\ve$ and $\delta$) such that if $L_0 \geq C$, then (\ref{emmbound}) holds for all $k \geq 0$.
	\end{claim}
	That $\tilde{M}$ does not meaningfully depend on the choice of initial scale is crucial. To show this, we first set for notational simplicity $\kappa = -\frac{1-6\ve}{\log_2 \delta}$. 
	
	Clearly (for $L_0 \geq C_{\ve,\delta}$) we have:
	\[ \frac{1}{1-(6-\kappa)\ve} \log_2 L_k \leq \log_2 L_{k+1} \leq \frac{1}{1-6\ve}\log_2 L_k\]
	Iterating,
	\begin{equation*}\left(\frac{1}{1-(6-\kappa)\ve}\right)^{\tilde{M}} \log_2 L_k\leq \log_2 L_{k+\tilde{M}} \leq \left(\frac{1}{1-6\ve}\right)^{\tilde{M}} \log_2 L_k\end{equation*}
	so that it suffices to find $\tilde{M}$ such that
	\[\log_2 \delta -2 \leq -\tilde{M} \log_2(1-(6-\kappa)\ve) \leq - \tilde{M}\log_2(1-6\ve) \leq \log_2\delta \]
	or, rewriting,
	\[ \frac{-\log_2\delta}{1-6\ve} \leq \tilde{M} \leq \frac{2-\log_2\delta}{1-(6-\kappa)\ve}\]
	By a straightforward computation, the two quantities on either end of the above inequality are separated by a distance of greater than one, and so there is at least one integer between them; we take e.g. $\tilde{M} := \lceil -\frac{\log_2 \delta}{1-6\ve}\rceil$.
	
	Having shown the existence of appropriate $\tilde{M}$, we introduce a notion of ``good'' and ``bad'' boxes; such notions are common to MSA.
	\begin{deffo}
		An $L_k$ square $\Lambda$ is good if $\P[\mathcal{E}_g(\Lambda)|V_{F_k\cap \Lambda}] = 1$. An $L_k$ square $\Lambda$ is bad if it is not good.
	\end{deffo}
	The inequality (\ref{conditionalgood}) then amounts to saying that boxes are good with high probability. Note that $\Lambda$ being good is a $V_{F_k\cap \Lambda}$-measurable event. We also want to consider a notion of nested defects at $\tilde{M}$ different scales.
	\begin{deffo}
		If $\Lambda_{\tilde{M}} \subset \Lambda_{\tilde{M}-1} \subset \cdots \subset \Lambda_1 \subset \Lambda$ are bad squares with $\ell(\Lambda) = L_k$, and $\ell(\Lambda_j) = L_{k-j}$, we call $\Lambda_{\tilde{M}}$ a hereditary bad subsquare of $\Lambda$.
	\end{deffo}
	The count of hereditary bad subsquares of $\Lambda$ is $V_{F_{k-1}\cap \Lambda}$ measurable. Our inductive step will entail, more or less, showing that if the number of hereditary bad subsquares is small with high probability, then by applying \Cref{detmsa}, we will have that $\Lambda$ is good with high probability. Specifically, we will show the following:
	\begin{claim}\label{readyhi}
		For all $k> \tilde{M}$, if $\ve\leq c$ and $N \geq C_{\tilde{M},\gamma,\delta}$, then
		\[ \P[\Lambda \text{ has fewer than }N\text{ hereditary bad subsquares }] \geq 1 - L_k^{-1}\]
		under the assumption that (\ref{conditionalgood}) holds for scales smaller than $L_k$.
	\end{claim}
	
	Recall that we deal with aligned squares, and the count of $L_j$ aligned subsquares of $L_k$ is of order $(L_k/L_j)^2$. We first note that if $\Lambda$ has a lot of hereditary bad squares, then we can locate some subsquare with a lot of (not necessarily hereditary) bad subsquares one scale down. Setting $N = (N')^{\tilde{M}}$ for some $N'$ to be determined, we have in probabilistic terms:
	\begin{align}\label{hb} \P[\Lambda\text{ has}>N\text{ h.b. subsquares}] \leq \sum_{\substack{ \Lambda' \subset \Lambda\\ \ell(\Lambda')=L_j\\k-\tilde{M}<j\leq k}} \P[\Lambda'\text{ has}>N'\text{ bad } L_{j-1}\text{ subsquares}]\end{align}
	Indeed, if $\Lambda$ does have more than $N=(N')^{\tilde{M}}$ hereditary bad subsquares, then we have chains of bad subsequares $\Lambda \supset \Lambda^1_1\supset \dots \Lambda^1_{\tilde{M}}$. In particular, $\cup_{j=1}^{N+1}\Lambda^j_1$ contains $(N+1)M \geq N'^{\tilde{M}}+1$ bad subsquares. (This is just what is given by the listing, and we will only consider the bad subsquares given by said listing.)
	
	There is thus some (minimal) $j'$ such that the count of bad $L_{k-j'}$ in this union exceeds $(N')^{j'}$, whence we conclude the ratio of $L_{j'}$ bad boxes to $L_{j'-1}$ bad boxes listed out exceeds $N'$. From this it follows that there exists $\Lambda'$ an $L_{j'-1}$ box with more than $N'$ bad $L_{j'-1}$ subsquares, showing the bound above.
	
	Having reduced the problem to bounding the right hand side of (\ref{hb}), we recall that we've assumed (\ref{conditionalgood}) to hold for smaller scales, and also that (for $L_0$ sufficiently large) $L_j^{1-C\ve} \geq L_{j-1}$. Using these facts, we get the estimates:
	\begin{align*}
		\text{R.H.S. of }(\ref{hb}) &\leq \sum_{\substack{\Lambda' \subset \Lambda\\ \ell(\Lambda') = L_j\\ k-\tilde{M} < j \leq k}} (L_j/L_{j-1})^CN'(L_{j-1}^{-\gamma})^{cN'}\\
		&\leq \sum_{k-\tilde{M} < j \leq k} (L_k/L_j)^C (L_j/L_{j-1})^{CN'}(L_{j-1}^{-\gamma})^{cN'}\\
		&\leq C\tilde{M}L_k^C L_{k-\tilde{M}}^{(C\ve-c\gamma)N'}
	\end{align*}
	with the last estimate only valid under the assumption $C\ve < c\gamma$. (Importantly, this is the only place $\ve$ small is needed, and so the requisite smallness does not depend on any other parameters; this is important as changing $\ve$ may necessitate a change in $\tilde{M}$.) Finally, by (\ref{emmbound}) we obtain:
	\[\text{R.H.S. of }(\ref{hb}) \leq C\tilde{M}L_k^C L_k^{(C\ve-c\gamma)\delta N'}\]
	Taking $N'$ large enough (and $L_0$ large enough) gives the desired bound, and so the claim is proved.
	\begin{deffo}
		If $k >\tilde{M}$, we call an $L_k$ square $\Lambda$ ready if it has less than $N$ hereditary bad subsquares, where $N$ is the same as the constant appearing in \Cref{readyhi}.
	\end{deffo}
	Then by showing \Cref{readyhi} we've shown that squares $\Lambda$ are ready with high probability. For the rest of the proof, it will be necessary to consider various subsquares of an $L_k$ square $\Lambda$. 

\begin{claim}\label{varsubsqclaim}
    Given any aligned $L_k$ square $\Lambda$, with $k>\tilde{M}$, there is a dyadic scale $L' \in [L_{k-1},CL_{k-1}]$ and lists of subsquares:
    \begin{enumerate}
        \item $\Lambda'''_1,\dots,\Lambda'''_N$ a list of $L_{k-\tilde{M}}$ squares which, if $\Lambda$ is ready, contains every hereditary bad subsquare
        \item $\Lambda''_1,\dots,\Lambda''_N$ a list of $L_{k-1}$ subsquares which, if $\Lambda$ is ready, contains every bad $L_{k-1}$ subsquare
        \item $\Lambda'_1,\dots, \Lambda'_N$ a list of $L'$ subsquares so that all $\Lambda''_j$ are contained inside some $\Lambda'_{j'}$ and moreover $\text{dist}(\Lambda''_j,\Lambda\setminus \Lambda_{j'}) \geq \frac{1}{8}L'$. 
    \end{enumerate}
    Moreover, $L'$ and all the squares $\Lambda'_j, \Lambda''_j, \Lambda'''_J$ can be chosen in a $F_k\cap \Lambda$ measurable way.
\end{claim}
That such $\Lambda'''_j$ exist is an immediate consequence of readiness; that they can be chosen measurably is a consequence of readiness and hereditary badness both being $F_k\cap \Lambda$ measurable properties. Every $L_{k-1}$ bad subsquare contains a hereditarily bad subsquare, so one similarly gets existence (and measurability) of $\Lambda''_j$. (Indeed; since our notion of good boxes includes both a non-resonance condition and decay of off-diagonal elements, absence of hereditary bad squares implies absence of bad $L_{k-1}$ subsquares via the deterministic part of prior MSA variants, see e.g. \cite[Lemma 4.2]{vonDreifus1989}.) Finally, One obtains the $\Lambda'_j$ from the entirely deterministic \Cref{coveringlemma}.

Finally, we define our frozen sites $F_k$ at scale $k$ to be the union of $F_{k-1}$ with the subsquares $\Lambda_j'$ so defined corresponding to any ready $L_k$ subsquares of $\Lambda$. Recall that we are considering only subsquares which satisfy the alignment condition (\ref{align}). Note that readiness of any $L_k$ square, and the $\Lambda_j'$, $\Lambda_j''$ and $\Lambda_j'''$ subsquares corresponding to the ready subsquares are all $V_{F_{k-1}\cap\Lambda}$ measurable, or at least can be so. Indeed, readiness concerns goodness/badness at scales $L_{k-1}$ and below, and by fixing a deterministic scheme to choose the $\Lambda_j'$, $\Lambda_j''$ and $\Lambda_j'''$ among the candidates, one gets $V_{F_{k-1}\cap \Lambda}$ measurability; the condition of being a candidate is clearly $V_{F_{k-1}\cap\Lambda}$ measurable.
	
Having defined all the relevant quantities, it remains to verify the conditions \ref{smallparam}-\ref{decaybound} for our inductively defined objects. Conditions \ref{smallparam} through \ref{decays} and \ref{decaybound} are all true by definition; our work lies in showing \ref{sets}, i.e. the regularity condition \ref{regularity}, the measurability condition \ref{measurable} and the resolvent bound \ref{resbound}. The first two are straightforward; the measurability condition \ref{measurable} especially so. The $L'_k$ corresponding to the bad subsquares of $\Lambda$ are $V_{F_{k-1}\cap\Lambda}$ measurable. However, $F_k\cap \Lambda$ may actually contain such squares from overlapping $L_k$ squares; all such squares are contained in $2\Lambda$, and so $F_k\cap \Lambda$ is $V_{F_{k-1}\cap2\Lambda}$ measurable. (Recall that our alignment condition (\ref{align}) does not preclude overlaps.) Moreover, such $Q$ must not be too large.
	\begin{claim}
		Property \ref{regularity} holds
	\end{claim}
	Note that we can assume $F_{k-1}$ is $\eta_{k-1}$ regular in $\Lambda$ as part of our induction hypothesis. So we are interested in controlling the total area of squares $Q$ in which $F_{k-1}$ is $\eta_{k-1}$ regular but $F_k$ is not $\eta_k$ regular. In particular, we need to show the area one can cover by such squares is less than $(\eta_k -\eta_{k-1})|\Lambda| = L_k^{-\ve}|\Lambda|$.
	
	By monotonicity, these squares intersect $F_k\setminus F_{k-1}$ non-trivially. However, $F_k\setminus F_{k-1}$ consists of $N$ boxes of side length $CL_{k-1}$, with $N$ and $C$ universal. At large enough scales, we have $L_{k-1} \leq L_k^{1-5\ve}$, and so in particular $|F_k \setminus F_{k-1}| \leq L_k^{-2\ve}$. Thus the area of boxes in which $F_k$ is not $\eta_k$ sparse but $F_{k-1}$ was $\eta_{k-1}$ was $\eta_{k-1}$ sparse is bounded by $\frac{1}{\eta_k - \eta_{k-1}}L^{-3\ve} \leq L_k^{-\ve} $.

	All that remains now is the proof of property \ref{resbound}, the sought exponential decay. We have already shown that boxes $\Lambda$ are ready with high probability, and will now show that they are good with high probability if they are ready. Specifically:
	\begin{claim}
		We define, for an $L_k$ square $\Lambda$ a family of events $\mathcal{E}_1(\Lambda),\dots,\mathcal{E}_N(\Lambda)$ in terms of the $\Lambda'''_j$ defined in \Cref{varsubsqclaim}: $\mathcal{E}_j(\Lambda)$ is precisely the event that \[\P[\|(H_{\Lambda_j'}-\overline{E})^{-1}\|\leq e^{L_k^{1-4\ve}}\,|\,V_{F_k\cap 4\Lambda}] = 1\]
	Then $\Lambda$ is good if $\mathcal{E}_1(\Lambda),\dots,\mathcal{E}_N(\Lambda)$ all hold.
	\end{claim}
	The proof of this claim is essentially just the deterministic MSA, \Cref{detmsa}. More precisely, for our six scales we take $L_k \geq L_k^{1-\ve} \geq L_k^{1-2\ve} \geq L_k^{1-3\ve} \geq L_k^{1-4\ve} \geq L_{k-1} \geq L_{k-1}^{1-\ve}$ with small parameters $\ve > \nu$. Our defect squares are the $\Lambda_j'$; by the way things have been set up it follows immediately that
	\[|(H_\Lambda-\overline{E})^{-1}(x,y)|\leq e^{L_k^{1-\ve}-m_k|x-y|}\]
	Combining this with \Cref{contresbound}, one gets that $\Lambda$ is good under the presumption of the events $\mathcal{E}_1(\Lambda),\dots,\mathcal{E}_N(\Lambda)$.
	
	So we would like to estimate the likelihood of the events $\mathcal{E}_j$. For this, we first need to prove a quasi-localization result, which we will ultimately use to get the requisite bounds from our Wegner estimate, \Cref{bigweg}.
	
	Towards this, we fix $\Lambda$ an $L_k$ ready square and $\Lambda_j'$ one of its associated $L_{k-1}$ squares, we define $G_+ = \Lambda_j' \setminus \cup_\ell \Lambda_\ell''$ and $G_- = \Lambda_j' \cap \cup \Lambda_\ell'''$, where in both cases $\ell$ ranges over $1,\dots, N$ and all squares are as laid out in \Cref{squareslist}. $G_+$ is in some sense the ``boundary'' of $\Lambda_j'$ and $G_-$ in some sense the ``deep interior'', and our next claim amounts to saying that very little mass (even in the $\ell^\infty$ sense) lives on the boundary, and the overwhelming majority lives in the ``deep interior'', the hereditary bad squares.
	\begin{claim}\label{quasilocal2}
		Let $G_-$ and $G_+$ as defined above, $\Lambda$ an $L_k$ ready square and $\Lambda_j'$ one of the associated $L'$ squares comprising $F_k$. If $|E-\overline{E}| \leq e^{L_{k-1}^{1-\ve}}$ and $\psi\neq 0$ solves $H_{\Lambda_j'}\psi = E\psi$, then we have the bounds:
		\[ e^{cL_{k-1}^{1-\delta}}\|\psi\|_{\ell^\infty(G_+)} \leq \|\psi\|_{\ell^2(\Lambda_j')} \leq (1+e^{-cL_{k+M}^{1-\delta}})\|\psi\|_{\ell^2(G_-)}\]
	\end{claim}
	If $x \in \Lambda_i' \setminus G_-$, then $x$ is contained in an $L_{k-j}$ good square $\Lambda_\ell$ such that moreover $x$ is ``deep inside'' $\Lambda''$, i.e. $\text{dist}(x,\Lambda_i'\setminus \Lambda'') \leq \frac{L_{k-j}}{8}$. (One can take $j=1$ if $x \in G_+$.)
	
	From the existence of said good box and the continuity of resolvent bounds in energy from \Cref{contresbound}:
	\[|\psi(x)| \leq 2e^{L_{k-j}^{1-\ve}-\frac{1}{8}m_{k-j}L_{k-j}}\|\psi\|_{\ell^2(\Lambda_i')} \leq e^{-cL_{k-j}^{1-\delta}}\|\psi\|_{\ell^2(\Lambda_i')}\]
	for all such $x$. In particular, we get
	\[ \|\psi\|_{\ell^\infty(G_+)} \leq e^{-cL_{k-1}^{1-\delta}}\|\psi\|_{\ell^2(\Lambda_i')}\]
	and
	\[ \|\psi\|_{\ell^\infty(\Lambda_i'\setminus G_-)} \leq e^{-cL_{k-M}^{1-\delta}}\|\psi\|_{\ell^2(\Lambda_i')}\]
	The first part of the sought inequality follows immediately; the second part follows nearly immediately by using the elementary bound
	\[ \|\psi\|_{\ell^2(G)} \leq |G|^{1/2} \|\psi\|_{\ell^\infty(G)}\]
	\begin{claim}
		For $1 \leq j \leq N$, we have
		\[\P[\mathcal{E}_j(\Lambda)] \geq 1-L_k^{C\ve-1}\]
	\end{claim}
	The proof of this fact follows from the Wegner lemma, \Cref{bigweg}. Specifically, we take scales $L' \geq L_k^{1-4\ve} \geq L_k^{1-5\ve} \geq L_{k-1} \geq L_{k-1}^{1-\delta}\geq L_5 = L_{k-1}^{1-\ve}$. We take our same $\ve$ and $\delta$ to be the small parameters. We take $F_{k-1}$ to be our frozen set, and $G:= \cup\{\Lambda_{j'}'''\,:\,\Lambda_{j'}''' \subset \Lambda_j'\}$ be our set of quasi-localization. (Note that this is $G_-$ from above.) $\Lambda_1' \subset F_k$ under the readiness assumption, so that in particular
	\[ \P[\mathcal{E}_j\,|\,\Lambda\text{ is ready}\,|\,V_{F_{k-1}}] \geq 1-(L')^{C\ve-1/2}\]
	which proves the claim when combined with the lower bound $L' \geq cL_k^{1-2\ve}$ and \Cref{readyhi}. Finally:
	\begin{claim}
		Property \ref{resbound} holds.
	\end{claim}
	By a union bound,
	\[\P[\cap_j\mathcal{E}_j] \geq 1-NL_k^{C\ve-1/2}\]
	hence, taking $\ve$ small enough with respect to $\gamma < 1/2$, one gets (\ref{resboundexp}) for sufficiently large scales. \end{proof}
\printbibliography
\end{document}